\documentclass{article}
\usepackage{wright}
\usepackage{ytableau}
\usepackage{subcaption}
\numberwithin{equation}{section}
\allowdisplaybreaks
\usepackage{tikz}
\usetikzlibrary{quantikz2}

\newcommand{\pbar}{\overline{P}}
\newcommand{\qbar}{\overline{Q}}
\newcommand{\rbar}{\overline{R}}
\DeclareMathOperator{\dist}{dist}
\DeclareMathOperator{\diam}{diam}

\newcommand{\Lbig}{\lind^{\mathrm{big}}}
\newcommand{\Lsmall}{\lind^{\mathrm{small}}}
\newcommand{\lambdabig}{\lambda^{\mathrm{big}}}
\newcommand{\lambdasmall}{\lambda^{\mathrm{small}}}
\newcommand{\Wbig}{W^{\mathrm{big}}}
\newcommand{\Wsmall}{W^{\mathrm{small}}}
\newcommand{\indicator}{\mathds{1}}

\newcommand{\pauli}[1]{\mathcal{P}_{#1}}

\newcommand{\cluster}[1]{\boldsymbol{#1}}
\usepackage[linesnumbered,ruled,vlined]{algorithm2e}
\zcRefTypeSetup{algocf}{
Name-sg = Algorithm ,
name-sg = algorithm ,
Name-pl = Algorithms ,
name-pl = algorithms ,
}

\newcommand{\round}{\operatorname{Round}}
\newcommand{\trunc}{\operatorname{Trunc}}
\renewcommand{\err}{\operatorname{err}}

\AtEveryBibitem{ % delete unnecessary information
    \clearfield{day}
    \clearfield{month}
    \clearfield{series}
    \clearfield{venue}
    \clearname{editor}
    \clearlist{publisher}
    \clearlist{location} % alias to field 'address'
    \clearfield{venue}
    \clearfield{issn}
    \clearfield{isbn}
    \clearfield{urldate}
    \clearfield{eventdate}
    \clearfield{pages}
    \clearfield{number}
    \clearfield{volume}
}
\newcommand{\locality}{k} % the locality of the Hamiltonian
\newcommand{\locals}{\mathcal{P}} % Pauli local terms
\newcommand{\energy}{{g}} % degree of the graph
\newcommand{\degree}{{d}} % degree for dual interaction graph
\DeclarePairedDelimiterXPP{\lonorm}[1]{}{\lVert}{\rVert}{_{B_1}}{#1} % local one norm
\newcommand{\tet}{t_{\textup{total}}} % total evolution time
\newcommand{\tres}{t_{\textup{min}}} % time resolution
\newcommand{\supp}{\mathrm{supp}} % support
\newcommand{\lind}{\mathcal{L}} % Lindbladian
\newcommand{\func}{\mathcal{F}} % function being optimized

\title{Learning the structure of open quantum systems}
\author{
    Laura Lewis\thanks{UC Berkeley. \texttt{\{lllewis,ewin,jswright\}@berkeley.edu}}
    \and Ewin Tang\footnotemark[1]
    \and John Wright\footnotemark[1]
}

\date{}

\begin{document}

\maketitle

\begin{abstract}
    We design an algorithm for learning the coefficients of an $n$-qubit constant-local Lindbladian to $\varepsilon$ error with $\mathcal{O}(g d^2 \log(n) / \varepsilon^2)$ total evolution time, where $g$ is the single-site energy and $d$ is the (approximate) degree of the interaction graph.
    Though Lindbladians present new challenges not present in the special case of Hamiltonians, our algorithm achieves the suite of desiderata attained by state-of-the-art Hamiltonian learning algorithms: (1) it uses non-adaptive, ancilla-free randomized Pauli measurement circuits with a time resolution of only $\Theta(1/g)$; (2) it works without knowledge of the structure of the unknown Lindbladian; (3) it depends on a smooth form of degree, thereby supporting the learning of quasi-local and power-law Lindbladians.
    Moreover, we prove a lower bound showing that our algorithm is optimal in each parameter up to logarithmic factors.
    
    Our algorithm is a simple iterative method, where the objective function consists of Fourier coefficients of the Lindbladian restricted to few-site regions.
    Its analysis identifies the difficulty unique to open systems, which we call ``confusing'' terms.
    For settings where the ``confusion'' is limited, the performance of the algorithm improves.
    We demonstrate this for the case of structure learning of Hamiltonians from access to real-time evolution, where we obtain a new, control-free algorithm that is significantly simpler than previous work.
    In addition, using the same iterative method, we design the first efficient algorithm for structure learning Hamiltonians from high-temperature Gibbs states.
\end{abstract}
%
%\newpage
\hypersetup{linktocpage}
\tableofcontents
%
%\newpage

\section{Introduction}

Much of modern physics has been built by probing quantum mechanical systems.
With the rise of controllable quantum systems as a means to run experiments at large scale~\cite{arute2019quantum,zhong2020quantum,wu2021strong,zhu2023interactive,google2025observation}, we must demand a more precise understanding of how to learn the behavior of quantum systems.
This drives the field of quantum learning theory, wherein one of the most active topics is Hamiltonian learning from real-time evolution: given some form of access to the dynamics $e^{-\ii Ht}$ of an unknown Hamiltonian, estimate $H$.
However, this assumes that the underlying evolution is performed in a closed system.
Far less is known about the analogous question for open systems, which evolve according to a Lindbladian $e^{\mathcal{L} t}$ (assuming that the evolution is time-independent and Markovian).
This problem is a natural extension of Hamiltonian learning because Hamiltonians form a subclass of Lindbladians.
Moreover, this question is particularly timely with the current influx of research interest in Lindbladians, due to recent advances in designing Lindbladians for efficiently preparing Gibbs states~\cite{chen2025efficient,chen2023efficient,ding2024efficient,scandi2026thermalization,rouze2026optimal,bakshi2026dobrushin,bakshi2024high,bergamaschi2024quantum,bergamaschi2026fast}.

In this work, we present an algorithm for structure learning an unknown (geometrically) local Lindbladian given access to its real-time dynamics\footnote{
    Hamiltonian and Lindbladian learning fall into two broad regimes.
    The first is the (geometrically) local regime, where the terms respect an underlying locality graph, and the total time evolution scales only logarithmically with the system size.
    The second is the sparse regime, where this condition is not imposed, and in exchange the time evolution scales polynomially with the system size.
    This work focuses on the first.
}.
As with Hamiltonian learning, one may hope to optimize many different figures of merit associated to a Lindbladian learning algorithm, including the simplicity of its circuits, total evolution time, time resolution, generality, and classical runtime.
Our algorithm is able to achieve performance matching state-of-the-art Hamiltonian learning results in all of these figures of merit.
We also prove a lower bound showing that the total time evolution used by our algorithm is optimal in all parameters up to logarithmic factors.

Moreover, to showcase the generality of our framework, a simplification of our main result yields new algorithms for structure learning Hamiltonians from both real-time evolution and high-temperature Gibbs states.
To our knowledge, this is the first result for structure learning Hamiltonians from the Gibbs state access model.

We first define the Lindbladian learning problem formally.
Consider a quantum system consisting of $n$ qubits.
The \emph{Lindbladian} $\mathcal{L}$ of this system describes its evolution under Markovian dynamics via the master equation~\cite{lindblad1976generators,gorini1976completely}:
\begin{equation}
    \label{eq:lind-main}
    \mathcal{L}(\rho) = \frac{1}{2}\sum_{P\neq I} \alpha_P [P, \rho] + \sum_{P_1,P_2 \neq I} D_{P_1,P_2}\left(P_1 \rho P_2 - \frac{1}{2} \{P_2P_1, \rho\}\right),
\end{equation}
where $P$, $P_1$, and $P_2$ are $n$-qubit Pauli operators.
Throughout, we consider $k$-local Lindbladians, which means that each term only acts on at most $k$ qubits.
Given access to the real-time evolution $e^{\mathcal{L}t}$, the goal is to learn the coefficients of the Lindbladian to error $\epsilon > 0$ in $\infty$-norm.

We refer to the first sum in \Cref{eq:lind-main} as the \emph{coherent} or \emph{Hamiltonian} part, while the second sum is the \emph{dissipative} part.
The dissipative part models interactions between the system and the environment.
In particular, if $D_{P_1,P_2} = 0$ for all $P_1,P_2$, then $e^{\mathcal{L}t}$ is simply a Hamiltonian evolution.
While evolution under a Hamiltonian $e^{-\ii Ht}$ is always unitary, including dissipative terms in the Lindbladian evolution $e^{\mathcal{L}t}$ results in more complex, non-unitary dynamics.

\subsection{Results}
We focus on learning a physically-motivated class of Lindbladians, where we require two main assumptions.
First, we assume that the total interaction strength of terms that act on any given qubit is bounded:
\begin{equation}
    \lonorm{\mathcal{L}} \triangleq \max_{i \in [n]} \parens[\Bigg]{\sum_{P : \supp(P) \ni i} |\alpha_P| + \sum_{\substack{P_1,P_2\\\supp(P_1) \cup \supp(P_2) \ni i}} |D_{P_1,P_2}|} \leq g.
\end{equation}
Such a bound implies that any qubit can only be involved in few terms with large interaction strengths.
This norm is sometimes referred to as the one-spin energy~\cite{akl16,alhambra2023quantum}, and an analogous condition has been considered in the Hamiltonian learning literature~\cite{bakshi2024structure}.

Second, we require the notion of \emph{approximate degree}.
Recall the typical notion of degree is the maximum number of Lindbladian terms acting on any site.
The approximate degree of $\lind$ can be viewed as a smooth version of the degree: it measures the degree of a modified Lindbladian, obtained by disregarding sufficiently small interactions from $\lind$.
For a parameter $\eta > 0$, it is defined as
\begin{equation}
    \deg_\eta(\lind) \triangleq \min_{\Lbig : \lonorm{\lind - \Lbig} < \eta} \deg(\Lbig).
\end{equation}
To our knowledge, this parameter, although a natural generalization of the degree, has not appeared in any prior Lindbladian or Hamiltonian learning paper.
However, similar parameters, e.g., the ``effective sparsity'' parameter in~\cite{bakshi2024structure}, have been considered previously.
We defer to \Cref{sec:lindblad} for further discussion of these definitions.

With this, we can state our main result, which obtains an algorithm for structure learning local Lindbladians.
A more detailed statement can be found in \Cref{cor:main-detailed}, and the full algorithm is presented in \Cref{alg:full}.

\begin{theorem}[Learning local Lindbladians from real-time evolution]
    \label{thm:main}
    Let $\epsilon > 0$, and let $k = \mathcal{O}(1)$.
    Let $\mathcal{L}$ be a $\locality$-local Lindbladian with unknown structure and known bounds on the local one-norm $\lonorm{\lind} \leq \energy$ and approximate degree, $\deg_{\epsilon/(100\cdot 16^{\locality})}(\mathcal{L}) \leq d$.
    Then there exists a quantum algorithm $\calA$ that outputs estimates $(\widehat{\alpha}, \widehat{D})$ with the following guarantees:
    \begin{enumerate}
        \item (Accuracy) With probability 0.99, $\lonorm{\widehat{\lind} - \lind} \leq \epsilon$, where $\widehat{\lind}$ is the Lindbladian with coefficients $\widehat{\alpha},\widehat{D}$. This implies that $\norm{\widehat{\alpha} - \alpha}_\infty \leq \epsilon$ and $\norm{\widehat{D} - D}_\infty \leq \epsilon$.
        \item (Evolution time) $\calA$ applies $e^{\calL t}$ with a total evolution time of $\tet = \mathcal{O}(g d^2\log(n)/\epsilon^2)$.
        \item (Time resolution) $\calA$ only applies $e^{\calL t}$ with $t \geq t_{\min} = \Theta(1/g)$.
        \item (Quantum measurements) $\calA$ performs $\mathcal{O}(g^2 d^2 \log(n)/\epsilon^2)$ quantum experiments of the following form: (i) prepare a Pauli eigenstate, (ii) apply $e^{\calL t}$, (iii) measure in a Pauli eigenbasis.
        \item (Classical overhead) $\calA$ has classical runtime $\mathcal{O}\left(\frac{n^kg^2d^2}{\epsilon^2}\log\left(\frac{g}{\epsilon}\right)\log\left(n\log\left(\frac{g}{\epsilon}\right)\right)\right)$.
    \end{enumerate}
\end{theorem}

We highlight that our learning algorithm only utilizes extremely simple quantum experiments: prepare a Pauli eigenstate, apply time evolution under the unknown Lindbladian, and measure in a Pauli eigenbasis (see \Cref{fig:experiments}).
We also remark that for $g,k,d = \mathcal{O}(1)$, our classical time complexity is $\widetilde{\mathcal{O}}(n^k/\epsilon^2)$.
We complement this upper bound with a lower bound showing that our algorithm is optimal in all parameters up to logarithmic factors\footnote{We note that this is optimal only with respect to \emph{each individual} parameter. It does not yet show full optimality, as the full known lower bound is $\widetilde{\Omega}(gd^2/\epsilon^2 + g\log(n)/\epsilon^2)$.}.
Moreover, it elucidates a fundamental difference between Hamiltonian and Lindbladian learning, where we show that the latter requires a strictly larger total time evolution (see the end of this section and \Cref{sec:related} for further discussion).

\begin{theorem}[Lower bound; Informal version of \Cref{thm:lower-bound}]
    \label{thm:lower-bound-main}
    Let $g > 0, 0 < \epsilon \leq g/512$.
    Let $d > 0$ be sufficiently large.
    Any learning algorithm which learns the coefficients of an unknown two-local, degree-$d$ Lindbladian with $\lonorm{\lind} \leq g$ to $\epsilon$ error in $\infty$-norm must use a total evolution time of $\tet = \Omega(gd^2/(\epsilon^2 \log d))$.
    This holds for any learning algorithms which can query the dynamics adaptively, have access to ancillary systems, and use entangled operations/measurements.
\end{theorem}

Moreover, the assumptions required for our algorithm are fairly general and encompass a wide range of natural, physically motivated settings (see \Cref{sec:apps}).
Namely, for suitable choices of $g$ and $d$, we can specialize to the four following cases.
We emphasize that all of these results apply to the problem of structure learning: the algorithm does not know a priori which of the unknown Lindbladian coefficients are large.

\begin{corollary}[Geometrically local Lindbladians; Informal version of \Cref{coro:geo-local}]
    Let $\epsilon > 0$.
    Let $\mathcal{L}$ be a $k$-local Lindbladian with bounded coefficients $|\alpha_P|, |D_{P_1,P_2}| \leq 1$ such that each qubit only interacts with $\mathcal{O}(1)$ nonzero terms.
    Then, there exists an algorithm which finds estimates $\widehat{\alpha}, \widehat{D}$ such that $\lonorm{\widehat{\lind} - \lind} \leq \epsilon$ with probability at least $0.99$ using $\tet = \mathcal{O}(\log(n)/\epsilon^2)$ total time evolution and time resolution $\tres = \Theta(1)$, where $\widehat{\lind}$ is the Lindbladian with coefficients $\widehat{\alpha},\widehat{D}$.
\end{corollary}

\begin{corollary}[Local Lindbladians; Informal version of \Cref{coro:local}]
    Let $\epsilon > 0$.
    Let $\lind$ be a $k$-local Lindbladian with $\lonorm{\mathcal{L}} = \mathcal{O}(1)$.
    Then, there exists an algorithm which finds estimates $\widehat{\alpha}, \widehat{D}$ such that $\lonorm{\widehat{\lind} -\lind} \leq \epsilon$ with probability at least $0.99$ using $\tet = \widetilde{\mathcal{O}}(n^{2k-2}\log(n)/\epsilon^2)$ total time evolution and time resolution $\tres = \Theta(1)$, where $\widehat{\lind}$ is the Lindbladian with coefficients $\widehat{\alpha},\widehat{D}$.
\end{corollary}

\begin{corollary}[Quasi-local Lindbladians; Informal version of \Cref{coro:quasi-local}]
    Let $\epsilon > 0$.
    Let $\mathcal{L}$ be a $k$-local, quasi-local Lindbladian on a $p$-dimensional lattice.
    If $p, k =\mathcal{O}(1)$, then there exists an algorithm which finds estimates $\widehat{\alpha}, \widehat{D}$ such that $\lonorm{\widehat{\lind} -\lind} \leq \epsilon$ with probability at least $0.99$ using $\tet = \mathcal{O}(\log(n)(\log(1/\epsilon))^{2pk}/\epsilon^2)$ total time evolution, where $\widehat{\lind}$ is the Lindbladian with coefficients $\widehat{\alpha},\widehat{D}$.
\end{corollary}

\begin{corollary}[Power-law Lindbladians; Informal version of \Cref{coro:algebraic}]
    Let $\epsilon > 0$.
    Let $\lind$ be a $k$-local Lindbladian on a $p$-dimensional lattice with $\gamma$-power-law decay for $\gamma - p > 0$.
    Let
    \begin{equation}
        \kappa = \frac{2pk}{\gamma - p}.
    \end{equation}
    Then, there exists an algorithm which finds estimates $\widehat{\alpha}, \widehat{D}$ such that $\lonorm{\widehat{\lind} -\lind} \leq \epsilon$ with probability at least $0.99$ using $\tet = \mathcal{O}(2^{\gamma \kappa} \log(n)/\epsilon^{2+\kappa})$ total time evolution, where $\widehat{\lind}$ is the Lindbladian with coefficients $\widehat{\alpha},\widehat{D}$.
\end{corollary}

The last two applications are especially interesting because such Lindbladians with long-range interactions are precisely those used for recent quantum Gibbs state preparation algorithms~\cite{chen2025efficient,ding2024efficient,scandi2026thermalization}.
Furthermore, for Lindbladians with power-law decay, as the decay strength $\gamma$ increases, $\kappa$ approaches $0$, so we recover the scaling $\tet = \mathcal{O}(\log(n)/\epsilon^2)$ in the fast-decay limit.

We illustrate the versatility of our framework further by applying it to not only learn different classes of Lindbladians, but also Hamiltonians.
In this setting, we consider two different access models: the ability to evolve under the dynamics $e^{-\ii Ht}$ and access to copies of the Gibbs state $\rho_\beta \triangleq e^{-\beta H}/\tr(e^{-\beta H})$, where $H$ is the unknown Hamiltonian.
In both cases, we obtain new, simple algorithms for structure learning Hamiltonians.
Moreover, prior to our work, there were no results in the literature for structure learning Hamiltonians from any Gibbs state access model.

\begin{theorem}[Structure learning Hamiltonians from real-time evolution; Informal version of \Cref{thm:ham-time}]
    \label{thm:ham-time-main}
    Let $\epsilon > 0$.
    Let $H$ be a $k$-local Hamiltonian with bounded coefficients $|\lambda_P| \leq 1$ and $\lonorm{\lambda} \leq g$.
    Then, given access to $e^{-\ii Ht}$, there exists an algorithm which finds estimates $\widehat{\lambda}$ such that $\norm{\widehat{\lambda} - \lambda}_\infty \leq \epsilon$ with probability at least $0.99$ using $\tet = \mathcal{O}(g\log(n)/\epsilon^2)$, time resolution $\tres = \Theta(1/g)$, and classical runtime $\mathcal{O}\left(\frac{n^k g^2}{\epsilon^2}\log\left(\frac{1}{\epsilon}\right)\log\left(\frac{n}{\delta \epsilon}\right)\right)$.
\end{theorem}

\begin{theorem}[Structure learning Hamiltonians from high-temperature Gibbs states; Informal version of~\Cref{thm:ham-gibbs}]
    \label{thm:ham-time-gibbs}
    Let $H$ be a $k$-local Hamiltonian with bounded coefficients $|\lambda_P| \leq 1$ and each qubit only interacts with $\mathcal{O}(1)$ nonzero terms.
    Let $\epsilon > 0$, and let $\beta < \beta_c$ for some critical inverse temperature $\beta_c$.
    Then, given access to copies of the Gibbs state $\rho_\beta$, there exists an algorithm which finds estimates $\widehat{\lambda}$ such that $\norm{\widehat{\lambda} - \lambda}_\infty \leq \epsilon$ with probability at least $0.99$ using $\mathcal{O}(\log(n)/(\beta\epsilon)^2)$ copies.
    The classical runtime of this algorithm is $\mathcal{O}(n^k\log(n)/(\beta\epsilon)^2)$.
\end{theorem}

Prior works on learning local Lindbladians either (1) have guarantees for structure learning Lindbladians with single-qubit dissipative terms~\cite{stilck2024efficient,stilck2025learning} or (2) only perform parameter learning~\cite{montana2025efficiently,ivashkov2026ansatz}.
In the first case,~\cite{stilck2024efficient,stilck2025learning} also both have a time resolution scaling as $\tres = \mathcal{O}(1/\polylog(1/\epsilon))$.
In the second case,~\cite{montana2025efficiently} performs parameter learning with a total time evolution scaling exponentially in $d$ and time resolution $\tres = \mathcal{O}(1/\polylog(1/\epsilon))$.
Moreover,~\cite{ivashkov2026ansatz} has a complexity dependent on the condition number of a large linear system, which, a priori may be exponentially large in $n$.
In contrast, our work achieves structure learning of local Lindbladians for arbitrary dissipative terms, with polynomial dependence on all parameters, and no condition number dependence.
We discuss related work in more detail in \Cref{sec:related}.

We also remark that the Lindbladian learning guarantees of our \Cref{thm:main} are comparable to state-of-the-art Hamiltonian learning results~\cite{bakshi2024structure}.
However, our total evolution time has a slightly worse dependence which is quadratic in the approximate degree versus \cite{bakshi2024structure}'s linear dependence in the analogous sparsity parameter\footnote{The Heisenberg limit is not possible to attain for learning Lindbladians, as they define a quantum channel. Thus, the standard quantum limit of $1/\epsilon^2$ is the optimal dependence on the error parameter.}.
Our lower bound in \Cref{thm:lower-bound-main} shows that this difference in total time evolution is fundamental.
Our classical runtime is also slightly worse, scaling with $n^k$ in contrast to the $n^2$ scaling of~\cite{bakshi2024structure}.
It is not known whether such $n^2$ classical runtime is attainable without quantum control.
We can also compare our Hamiltonian learning result (\Cref{thm:ham-time-main}) to~\cite{bakshi2024structure}.
Note that~\cite{bakshi2024structure}'s effective sparsity parameter $r$ satisfies $r \leq g/\epsilon$, so their algorithm also achieves the same total time evolution as our \Cref{thm:ham-time-main}.
However, our algorithm is significantly simpler, does not use quantum control, and has an improved time resolution.

\subsection{Technical overview}

\subsubsection{Our algorithm}
We explain the ideas behind our algorithm for the well-studied special case of geometrically local Lindbladians.
We also restrict to fully dissipative Lindbladians (i.e., $\alpha_P = 0$ for all $P$) for simplicity, as the inclusion of coherent terms can be analyzed similarly.
For the proof of our more general \Cref{thm:main}, we refer the reader to \Cref{sec:algo}.
At a high level, our paper extends the techniques of the prior works~\cite{haah2024learning,bakshi2024structure}, which were developed for learning Hamiltonians from their time evolutions, and adapts them to the problem of learning Lindbladians.

\paragraph{Review of~\cite{haah2024learning}.}
We begin by reviewing the approach of~\cite{haah2024learning}, which performs parameter learning of a geometrically local Hamiltonian $H = H(\alpha) = \sum_P \alpha_P P$.
The algorithm of~\cite{haah2024learning} follows a two-step procedure.
First, it produces estimates $\widehat{E}_P(\alpha)$ of expectation values of the form $E_P(\alpha) \triangleq \ntr(O_{1,P} e^{-\ii H(\alpha)t} O_{2,P} e^{\ii H(\alpha)t})$ for particular choices of Pauli observables $O_{1,P}, O_{2,P}$ for all $P$ in the known structure of $H$.
Then, it converts these estimates involving the time evolution $e^{-\ii Ht}$ into estimates for the Hamiltonian $H$ by finding zeros of the function
\begin{equation}
    \mathcal{F}_P(x) \triangleq \frac{1}{t}E_P(x) - \frac{1}{t}\widehat{E}_P(\alpha),
\end{equation}
using tools from convex optimization.
In particular, it does so via an iterative method, where, at the $j$-th iteration, the updates to the estimated Hamiltonian parameters are
\begin{equation}
    \label{eq:richardson-main}
    x^{(j+1)} = x^{(j)} - \frac{1}{2}\mathcal{F}(x^{(j)}).
\end{equation}
Here, we display a simplified version of the Newton-Raphson iteration used in~\cite{haah2024learning} which obtains the same guarantees\footnote{This version is not present in the literature and was observed by the authors of the current manuscript. The key insight behind this simplification is that~\cite{haah2024learning} does not leverage the fast quadratic convergence of the Newton-Raphson method. Instead, to obtain their result, they only require linear convergence, which can be attained by many other convex optimization methods, including Richardson iterations, upon which this simplified iteration is based. See \Cref{sec:ham} for an analysis of a variant of this iteration.}.
The point is that, for the specific choice of observables $O_{1,P}, O_{2,P}$,
\begin{equation}
    \label{eq:f-linear-ham}
    \mathcal{F}_P(x) = 2(x_P - \alpha_P) + \mathcal{O}(t).
\end{equation}
This can be seen by taking a linear approximation to the exponential.
Thus, up to a linear order approximation, this is the correct update to apply.
The key technical contributions of~\cite{haah2024learning} are bounding the contribution of the higher-order terms of $\mathcal{F}$ and showing how to efficiently approximate $\mathcal{F}_P(x)$ via a truncated series expansion.

\paragraph{Parameter learning Lindbladians.}
It is natural to attempt to generalize the algorithm of~\cite{haah2024learning} to parameter learning for Lindbladians.
Note that~\cite{haah2024learning} does not perform structure learning, so we only focus on parameter learning for now, i.e., the algorithm knows a priori which Lindbladian coefficients are nonzero.
Unfortunately, this approach quickly encounters obstacles.
Namely, what observables $O_1, O_2$ should we measure?
In the Hamiltonian case, there exists a choice of $O_{1,P}, O_{2,P}$ such that $E_P(\alpha) = 2t\cdot\alpha_P + \mathcal{O}(t^2)$, which is how we obtain \Cref{eq:f-linear-ham}.
In other words, for this choice of observables, the expectation values directly estimate the unknown Hamiltonian coefficients.
However, even for a single-qubit, fully dissipative Lindbladian, no such choice of observables exists.
Instead, expectation values yield a linear combination of Lindbladian coefficients, making it necessary to solve a linear system of equations to recover the coefficients using this approach.
This is the source of unanalyzed condition numbers and restrictions to single-qubit dissipators in prior work~\cite{stilck2024efficient,stilck2025learning,ivashkov2026ansatz}.

One contribution of our work is to overcome this difficulty using tools from Fourier analysis.
Inspired by Fourier inversion, we consider expectation values of the form
\begin{equation}
    \label{eq:exp-vals-main}
    E_{P_1,P_2}(D) \triangleq \frac{1}{2^n} \E_{R}[\tr(e^{\mathcal{L}_{D}t}(R)P_2 RP_1)],
\end{equation}
where $\mathcal{L}_{D}$ denotes a (fully dissipative) Lindbladian with dissipative coefficients $D$.
When $R$ is uniformly random over $n$-qubit Paulis, then this quantity can be naturally interpreted as a \emph{Fourier coefficient} of the channel $e^{\mathcal{L}_{D}t}$.
Moreover, these expectations satisfy similar properties as in the Hamiltonian case, where $E_{P_1,P_2}(D) = t \cdot D_{P_1,P_2} + \mathcal{O}(t^2)$.
Now, one may hope that the guarantees of~\cite{haah2024learning} apply when instead estimating the expectation values from \Cref{eq:exp-vals-main}.

However, when $R$ is sampled uniformly over $n$-qubit Paulis, it is not possible to parallelize the measurements of these expectation values, resulting in a large total time evolution.
Nevertheless, estimating the \emph{local} Fourier coefficients of $e^{\mathcal{L}_{D}}$, defined as in \Cref{eq:exp-vals-main} except where $R$ is instead a uniformly random Pauli on $\supp(P_1) \cup \supp(P_2)$, turn out to be sufficient.
Moreover, for a $k$-local Lindbladian, $|\supp(P_1) \cup \supp(P_2)| \leq k$, so these local Fourier coefficients can be estimated efficiently.
Luckily, using these local Fourier coefficients, it turns out that the iterative method of~\cite{haah2024learning} can be extended to obtain a parameter learning algorithm for geometrically local Lindbladians.

\paragraph{Structure learning Lindbladians.}
We now discuss pushing these ideas further to the problem of structure learning Lindbladians.
While~\cite{bakshi2024structure} can be viewed as a way to extend~\cite{haah2024learning} to structure learning for Hamiltonians, performing similar modifications to our parameter learning algorithm for Lindbladians is not straightforward.
In particular,~\cite{bakshi2024structure} uses techniques which are specialized to the Hamiltonian setting.
Namely, it begins by running a base learning algorithm to produce a coarse approximation $\widehat{H}$ to the true Hamiltonian $H$.
Then, to improve its estimate, it uses Trotterization to simulate access to $e^{-\ii(H - \widehat{H})}$ and obtain an estimate of $H - \widehat{H}$, which can in turn be added back to $\widehat{H}$ to produce a better estimate of $H$.
However, Trotterization requires access to the inverse time evolution,
which cannot be performed for Lindbladians because they are dissipative and their evolutions cannot be reversed.
Thus, we attempt a different modification to perform structure learning.

A simple approach one may take is to keep track of all $k$-local coefficients instead of only the coefficients in the known structure.
In other words, an algorithm may update coefficient estimates using the errors
\begin{equation}
    \mathcal{F}_{P_1,P_2}(x) \triangleq \frac{1}{t}E_{P_1,P_2}(x) - \frac{1}{t}\widehat{E}_{P_1,P_2}(D)
\end{equation}
for all $k$-local $P_1,P_2$.
The problem is that the recovery of the Fourier coefficients of $e^{\mathcal{L}_D t}$ from the local Fourier coefficients then becomes more difficult, as many Lindbladian terms can interfere with each other.
Moreover, the contribution of Lindbladian terms which are small but nevertheless still part of the structure can be obscured by the contribution of large terms.

We quantify this ``confusion'' as follows.
Because the expectation over $R$ in \Cref{eq:exp-vals-main} is not taken over all $n$-qubit Paulis, the local Fourier coefficient $E_{P_1,P_2}(D)$ does not precisely approximate the dissipative coefficient $D_{P_1,P_2}$.
Instead, the local Fourier coefficients also include contributions from other Paulis $(Q_1, Q_2)$, which are ``confused'' with the correct term $(P_1, P_2)$.
We write $(Q_1, Q_2) \succeq (P_1, P_2)$ to denote such Paulis, which are defined as $(Q_1, Q_2)$ that agree with $(P_1, P_2)$, respectively, on $\supp(P_1) \cup \supp(P_2)$ and that agree with each other outside of this support.
In \Cref{sec:local-fourier}, we prove that
\begin{equation}
    \label{eq:confusion-main}
    E_{P_1,P_2}(D) \approx t\sum_{(Q_1, Q_2) \succeq (P_1, P_2)} D_{Q_1, Q_2} \triangleq t(AD)_{P_1,P_2},
\end{equation}
up to a linear approximation of $e^{\mathcal{L}_Dt}$.
The problem described above, that recovering the Lindbladian coefficients from the local Fourier coefficients becomes difficult, can be made precise in that the matrix $A$ does not have a well-behaved inverse.
In particular, $\norm{A^{-1}}_{\infty\to\infty}$ can scale polynomially in $n$.
Operationally, this means that if one attempts to perform an iteration of the form
\begin{equation}
    x^{(j+1)} = x^{(j)} - A^{-1}\mathcal{F}(x^{(j)}),
\end{equation}
which is a natural extension of \Cref{eq:richardson-main}, the error in each iteration increases by a factor of $\poly(n)$.
To counteract this, one would need to estimate the local Fourier coefficients to $\epsilon/\poly(n)$ error.
However, this results in an abysmal total time evolution of $\mathcal{O}(\poly(n)/\epsilon^2)$ for learning geometrically local Lindbladians, whereas one would typically expect an exponentially smaller total time evolution of $\mathcal{O}(\log(n)/\epsilon^2)$.

This is an obstacle unique to Lindbladian learning.
In contrast, for Hamiltonian learning, there is no ``confusion'': the local Fourier coefficients still approximate the corresponding Hamiltonian coefficient directly.
Thus, the matrix $A$ in \Cref{eq:confusion-main} is simply the identity matrix, which has a bounded $\infty\to\infty$ norm.
We refer the reader to \Cref{sec:ham} to see how this greatly simplifies the analysis.

The critical problem here is that to estimate $D_{P_1, P_2}$, all possible pairs $(Q_1, Q_2)$ which can be confused with $(P_1, P_2)$, of which there can be roughly $\mathcal{O}(n^k)$, contribute some error, whereas we should only actually have $\mathcal{O}(d)$ large terms which contribute large error.
To remedy this, we round small entries of both $\mathcal{F}(x^{(j)})$ and our estimate after an update to zero, i.e., we consider the iteration
\begin{equation}
    x^{(j+1)} = \round_{\tau_j}(x^{(j)} - A^{-1}\round_{\tau_j'}(\mathcal{F}(x^{(j)}))),
\end{equation}
for some carefully chosen thresholds $\tau_j, \tau_j' > 0$.
Here, $\round_\tau(x)$ is zero if $|x| \leq \tau$ and $x$ otherwise.
After rounding, the remaining nonzero coefficients correspond to the structure of $\mathcal{L}$ discovered so far.
Rounding in this way ensures that our estimated Lindbladian in each iteration always has degree $\mathcal{O}(d)$, so we effectively only incur a total time evolution cost comparable to parameter learning a Lindbladian with degree $\mathcal{O}(d)$.

Technically, analyzing this new rounded algorithm requires bounding $\norm{A^{-1}}_{B_1\to B_1}$ (instead of $\norm{A^{-1}}_{\infty\to\infty}$, see \Cref{sec:local-lindblad-fourier}) and maintaining the error of our iterates in $B_1$-norm.
Throughout this discussion, we have also been ignoring errors arising from linear approximation, and a significant portion of our analysis is dedicated to showing that this error is not too large (see \Cref{sec:series}).

\subsubsection{Our lower bound}

\paragraph{Our construction.}
To discuss our lower bound in \Cref{thm:lower-bound-main}, we begin by describing our lower bound instance.
We consider a family of 2-local Lindbladians with a star graph topology, in which the central vertex, labeled by 0, has pairwise interactions with the star vertices $i = 1, \ldots, n-1$.
In particular, the interactions only involve $Z$ Paulis, and we can write the Lindbladian as
\begin{equation}\label{eq:lindblad-def}
    \mathcal{L}(\rho) = \frac{1}{2} \ii a [Z_0, \rho] + \sum_{i=1}^{n-1} J_i \rho J_i^{\dagger} - \frac{1}{2} \{J_i^\dagger J_i, \rho\},
    \qquad \text{where } J_i = b_i \cdot Z_i + c_i \cdot Z_0 Z_i,
\end{equation}
where $a,b_i,c_i$ are coefficients to be chosen later.
We consider two families of random Lindbladians, the first of which typically chooses the coefficient $a$ as a reasonably large positive number, and the second of which typically chooses $a$ as a reasonably large negative number.
Therefore, to prove a lower bound on the total time evolution required to learn the coefficients of a Lindbladian, it suffices to prove a lower bound on the total time needed to distinguish these two families.

\paragraph{Thwarting a distinguisher.}
To understand how best to sample the coefficients of these hard Lindbladian families, we first describe how one would typically attempt to estimate the coefficient $a$.
The standard method is to prepare the initial state
\begin{equation}
    \rho = \frac{I + X_0}{2^n},
\end{equation}
apply the time evolution $e^{\calL t}$, and measure with the $Y_0$ observable.
In expectation, the measurement outcome is
\begin{align}
    \tr(Y_0 e^{\lind t}(\rho)) &= \tr(Y_0(\rho + t\lind(\rho))) + \mathcal{O}(t^2) = t\cdot \tr(Y_0 \lind(\rho)) + \mathcal{O}(t^2) = \frac{t}{2^n}\tr(Y_0 \lind (X_0)) + \mathcal{O}(t^2),
\end{align}
where we expand $e^{\lind t} = \rho + t\lind(\rho) + \mathcal{O}(t^2)$.
In the second and third equalities, we use that $\tr(Y_0) = \tr(Y_0 Z_0) = 0$.
Now, the contribution to $\tr(Y_0 \calL(X_0))$ from the coherent term is given by
\begin{equation}
    \frac{1}{2} \ii a  \tr(Y_0 [Z_0, X_0])
    = \frac{1}{2} \ii a (\tr(Y_0 Z_0 X_0) - \tr(Y_0 X_0 Z_0))
    = - a 2^n.
\end{equation}
Similarly, for the dissipative terms, one can show that the dissipative term indexed by $i$ contributes
\begin{equation}
    \tr\Big(Y_0 \Big(J_i X_0 J_i^\dagger - \frac{1}{2} \{J_i^\dagger J_i, X_0\}\Big)\Big)
    = \ii 2^n (b_i^* c_i - b_i c_i^*).
\end{equation}
Putting everything together, then
\begin{equation}\label{eq:distinguisher}
    \tr(Y_0 e^{\lind t}(\rho)) = t\left(-a + \sum_{i=1}^{n-1} \ii (b_i^*c_i - b_i c_i^*)\right) + \mathcal{O}(t^2).
\end{equation}
Notably, up to first order in $t$, this expectation value is heavily biased in the direction of $a$ (either very positive or very negative), allowing us to learn its value and distinguish the two ensembles.
To remedy this, we use the contribution from the dissipative terms to cancel it out.

Concretely, we specify the two families of Lindbladians by two probability distributions, $\mathcal{D}_+$ and $\mathcal{D}_-$ on real numbers.
We construct $\mathcal{D}_+$ such that it tends to produce positive real numbers, and $\mathcal{D}_-$ such that it tends to produce negative real numbers.
For the positive case, say, then we sample $n-1$ numbers $u_1, \ldots, u_{n-1}$ independently from $\mathcal{D}_+$ and set
\begin{equation}
a = - 2 \sum_{i=1}^{n-1} u_i,
\quad b_i = 1,
\quad c_i = \ii u_i,
\end{equation}
for all $1 \leq i \leq n -1$.
The negative case is the same except we draw $u_1,\dots,u_{n-1}$ from $\mathcal{D}_-$.
Here, because $\mathcal{D}_+$ tends positive, then $a$ is very positive in that case, and analogously for $\mathcal{D}_-$.
However, with this choice of coefficients, we see that, up to first order in $t$, the expectation value in~\Cref{eq:distinguisher} is exactly 0 in both cases.
Thus, this strategy no longer distinguishes these two cases.

\paragraph{Thwarting another distinguisher by moment matching.}
In designing our Lindbladian families to be resistant to the distinguisher described above,
we have introduced another weakness that can be used to distinguish them.
In particular, suppose that we start in the initial state $\ket{+}^{\otimes n}$, apply the time evolution $e^{\calL t}$, and then measure in the Hadamard basis.
Before measurement, the state is
\begin{equation}
    e^{\calL t}(\ketbra{+}{+}^{\otimes n}) = \ketbra{+}{+}^{\otimes n} + t \calL(\ketbra{+}{+}^{\otimes n}) + \mathcal{O}(t^2).
\end{equation}
From the form of the Lindbladian in~\Cref{eq:lindblad-def}, the action of the Lindbladian results in a single flip from $\ket{+}$ to $\ket{-}$ on the $i$-th star vertex with probability $t |b_i|^2 = t$ and a pair of such flips (one on vertex 0 and another on vertex $i$) with probability $t |c_i|^2 = t u_i^2$.
We focus on the latter case.
This tells us that the probability we see two flips is, up to first order in $t$,
\begin{equation}
    t \sum_{i=1}^{n-1} u_i^2.
\end{equation}
Since each $u_i$ is sampled i.i.d. from either $\mathcal{D}_+$ or $\mathcal{D}_-$, observing this quantity can allow us to estimate the second moment $\E[u_i^2]$ of a random variable from either distribution.
Thus, if we want the two distributions to be indistinguishable, we should ensure that they have matching second moments.
Similarly, suppose we perform this experiment twice, and we check whether in each experiment there is a pair of flips and the two experiments both flip the same star vertex (i.e., the same vertex $i$, for some $1 \leq i \leq n -1$).
The probability that this occurs is proportional to
\begin{equation}
    \sum_{i=1}^{n-1} u_i^4,
\end{equation}
which can be viewed as an estimate of the fourth moment of the distribution.
Repeating this logic, we observe that we want random variables drawn from $\mathcal{D}_+$ to have large positive first moments (i.e., means) and those from $\mathcal{D}_-$ to have large negative first moments, but otherwise these two random variables should have matching second, third, and higher moments.

Designing random variables which behave noticeably differently but have matching lower-order moments is a common ingredient in lower-bound proofs.
To our knowledge, the specific properties we require of our random variables have not appeared in the literature.
However, we are able to design appropriate random variables through a straightforward modification to a previous moment-matching construction of  Wu and Yang~\cite{wu2019chebyshev}.

\paragraph{Proving the lower bound.}
With this, we have defined our two families of Lindbladians, but it still remains to show that they are hard to distinguish by algorithms which can query them adaptively, have access to ancillary systems, and use entangled operations/measurements.
This is a powerful class of algorithms capable of complicated behavior, so proving limitations on its abilities at first blush seems quite challenging.
However, we are able to do so by an idea inspired by recent works of Patel and Wilde~\cite{patel2023wave1,patel2023wave}; namely we show how to simulate access to the time evolution under a Lindbladian from either of our families, given one copy of a properly designed ``program'' state.
More specifically, we show that given a Lindbladian $\calL_{\alpha, D}$ in either one of our families, there is a program state $\Psi_{\alpha, D}(T)$ which can be used simulate arbitrary adaptive queries to $e^{\lind_{\alpha, D}}$, as long as the total time evolution used is at most $T$.
Thus, proving that the two families of Lindbladians are hard to distinguish reduces to proving that their two respective program states are hard to distinguish, i.e., proving a bound on
\begin{equation}
    \dtr\left(\E_{(\alpha, D)\sim \mathcal{D}_+}[\Psi_{\alpha, D}(T)], \E_{(\alpha, D)\sim\mathcal{D}_-}[\Psi_{\alpha, D}(T)]\right).
\end{equation}
This is significantly more straightforward to calculate, and we do so using our moment matching conditions.

The construction of our program state can be viewed as a simplification of the program states in\cite{patel2023wave}.
For a single term of a Lindbladian in our families (i.e., a term indexed by $i$), the program states of~\cite{patel2023wave} can be identified with
\begin{equation}
    \ket{\psi_{z,N}} = \frac{\ket{0} + \frac{z}{\sqrt{N}}\ket{1}}{\sqrt{1 + |z|^2/N}},
\end{equation}
for $N \in \mathbb{Z}_{\geq 1}$ and some $z \in \mathbb{C}$ which depends on the coefficient $u_i$ and the time $T$.
\cite{patel2023wave} shows how to simulate the Lindbladian time evolution using multiple copies of $\ket{\psi_{z,N}}$.
Our program state $\Psi_{\alpha, D}(T)$ can be seen as $\ket{\psi_{z,N}}^{\otimes N}$ in the limit as $N \to \infty$.
In this limit, one can show that $\ket{\psi_{z,N}}^{\otimes N}$ converges to a bosonic coherent state (\Cref{eq:coherent-state}), which is how we define our program states.
Working with bosonic states allows us to use only a single copy of the program state to simulate the time evolution rather than needing many copies.
In addition, working with bosonic states yields a cleaner analysis, and it allows us to generalize their simulation result from nonadaptive queries to fully adaptive queries.
We refer the reader to \Cref{sec:lower-bound} for the full proof.

\subsection{Related work}
\label{sec:related}

\paragraph{Hamiltonian learning.}
For the simpler task of Hamiltonian learning, there is an extensive literature solving this task in a variety of access models, e.g., from copies of the Gibbs state~\cite{bairey2019learning,anshu2021sample,haah2024learning,bakshi2024learning,chen2025learning,qi2019determining}, access to the Hamiltonian's real-time evolution~\cite{zubida2021optimal,haah2024learning,huang2023learning,bakshi2024structure,zhao2024learning,ma2024learning,hu2025ansatz,abbas2025nearly,dutkiewicz2024advantage,caro2024learning,odake2024higher,gutierrez2024simple,arunachalam2024testing,castaneda2023hamiltonian,chen2025lower,sinha2025improved,bluhm2025certifying,shin2026heisenberg,mirani2024learning,mobus2025heisenberg,li2024heisenberg,ni2024quantum}, and more restrictive settings~\cite{brahmachari2026learning,pradenne2026learning,chen2025probe}.
The works most relevant to ours are those that consider Hamiltonian learning from access to dynamics, where the unknown Hamiltonian is promised to be (geometrically) local.
We only detail the results of some of these works and refer to, e.g.,~\cite{bakshi2024structure}, for a more thorough review.

In this setting, early works, e.g.,~\cite{bairey2019learning}, designed an algorithm using $\tet = \mathcal{O}(\log(n)/\epsilon^3)$ and $\tres = \mathcal{O}(\epsilon)$.
This approach can be modified to achieve structure learning.
For known structure, \cite{haah2024learning} improved this to $\tet = \mathcal{O}(\log(n)/\epsilon^2)$ and $\tres = \Omega(1)$.
Moreover, \cite{huang2023learning} achieved the Heisenberg scaling with $\tet = \mathcal{O}(\log(n)/\epsilon)$ and $\tres = \Omega(\sqrt{\epsilon})$.
Both \cite{haah2024learning,huang2023learning} only work for parameter learning, not structure learning.
\cite{bakshi2024structure} later achieved structure learning with the Heisenberg scaling while maintaining a constant time resolution, i.e., $\tet = \mathcal{O}(\log(n)/\epsilon)$ and $\tres = \Theta(1)$.

Our work is most comparable to~\cite{haah2024learning,bakshi2024structure}, as we achieve the optimal scaling of $\tet = \mathcal{O}(\log(n)/\epsilon^2)$ with a constant time resolution $\tres = \Theta(1)$.
At a high level, our algorithm resembles those of~\cite{haah2024learning,bakshi2024structure}, as ours is also an iterative procedure based on convex optimization algorithms.
While~\cite{bakshi2024structure} can be seen as a way to extend~\cite{haah2024learning} to learn the structure of Hamiltonians, performing a similar modification for Lindbladians is not straightforward.
In particular, in each iteration, \cite{bakshi2024structure} updates the learned parameters based on expectations with respect to $\exp(-\ii(H - \widehat{H}))$, where $\widehat{H}$ is the current estimate of the Hamiltonian, and access to $H - \widehat{H}$ can be simulated via a new constant-time Trotterization formula.
However, a similar Trotterization formula is not expected to be possible for Lindbladians.
Thus, one main conceptual contribution of our work is to overcome this barrier and find a new way to update our estimates of the Lindbladian parameters.

The above discussion takes all parameters to be constant, i.e., $k, g, d = \mathcal{O}(1)$.
When considering the scaling with $g$ and $d$, our algorithm for Lindbladian learning achieves $\tet = \mathcal{O}(gd^2\log(n)/\epsilon^2)$.
Meanwhile,~\cite{bakshi2024structure} has a better scaling of $\tet = \mathcal{O}(d\log(n)/\epsilon)$\footnote{\cite{bakshi2024learning} considers a slightly different parameter than our approximate degree $d$, which they call the effective sparsity. These capture similar physical settings, so we state their complexity in terms of $d$ for comparison.}.
While the Heisenberg scaling $1/\epsilon$ is not possible for learning Lindbladians, the total time evolution of~\cite{bakshi2024structure} is still better by a factor of $d$ and $g$.
However, our lower bound shows that this difference in total time evolution is fundamental.
For our Hamiltonian learning result in \Cref{thm:ham-time-main}, \cite{bakshi2024structure} obtains the same total time evolution (seen by bounding their effective sparsity parameter as $r \leq g/\epsilon$).
However, our algorithm is significantly simpler, does not use quantum control, and has an improved time resolution (our algorithm has $\tres = \Theta(1/g)$ compared to their $\tres = \Theta(1/d)$).

\paragraph{Lindbladian learning.}
Early works studied the task of recovering a description of the Lindbladian from access to its dynamics or steady states but lacked rigorous guarantees~\cite{buvzek1998reconstruction,BGP+20}.
Since then, interest in this problem has gained momentum with the development of many heuristic/numerical algorithms~\cite{liu2025robust,onorati2023fitting,olsacher2025hamiltonian,pastori2022characterization} and even some experimental demonstrations~\cite{kraft2025bounded,birke2026demonstrating,lam2026pairwise,berg2025large}.

The works most relevant to the present manuscript are those with provable guarantees on the total time evolution required to learn the Lindbladian given access to its time evolution operator.
However, no prior work achieves a rigorous guarantee for learning local Lindbladians with optimal performance in total evolution time and time resolution, even for the easier task of parameter learning.
\cite{da2011practical} uses time derivative estimation, resulting in $\tet = \mathcal{O}(\nu\log(n)/\epsilon^3)$ (when combined with randomized measurements as in~\cite{haah2024learning}) and $\tres = \mathcal{O}(\epsilon)$.
Here, $\nu$ is a condition number factor, which is implicit in the complexity, as their algorithm requires inverting a linear system which is, a priori, not well-conditioned.
In addition, recent works studying this problem fall into two main categories:
\begin{enumerate}
    \item The work gives algorithms for structure learning, but with guarantees which only hold in very restricted settings.
    \item The work only gives algorithms for parameter learning, and the total evolution time scales poorly in some parameters.
\end{enumerate}
In particular, \cite{stilck2024efficient,stilck2025learning} both fall into the first category, where they achieve structure learning of local Lindbladians with $\tet = \widetilde{\mathcal{O}}(\log(n)/\epsilon^2)$ and $\tres = \mathcal{O}(1/\polylog(1/\epsilon))$, but their guarantees only apply to Lindbladians with single-qubit dissipative terms.
Their algorithm's dependence on $d$ is not clear, as they take $d = \mathcal{O}(1)$ throughout.
\cite{stilck2024efficient} also considers Lindbladians with single-qubit dissipators and algebraic decay on a $p$-dimensional lattice.
In this case, their algorithm achieves a similar complexity as ours.
However, for a decay rate of $\gamma > 0$, their guarantee only holds for $\gamma \geq 5p$.

\cite{montana2025efficiently,ivashkov2026ansatz} belong to the second category\footnote{We only quote the result from~\cite{ivashkov2026ansatz} relevant to our paper, which is parameter learning for local Lindbladians. In the sparse setting, they obtain a structure learning algorithm.}.
\cite{montana2025efficiently} solves the problem of parameter learning using $\tet = \mathcal{O}(\exp(d) \log(n)/\epsilon^2)$ and $\tres = \mathcal{O}(1/\polylog(1/\epsilon))$.
\cite{ivashkov2026ansatz} achieves $\tet = \widetilde{\mathcal{O}}(d^2\nu\log(n)/\epsilon^2)$ and $\tres = \Theta(1/d)$ for parameter learning.
Here, the total time evolution $\tet$ depends on the condition number $\nu$ of a large linear system, which is not analyzed and could be exponentially large in $n$.
It is also worth noting that for general $k$-local Lindbladians, \cite{ivashkov2026ansatz} presents a structure learning algorithm with $\tet = \widetilde{\mathcal{O}}(\nu n^{4k}/\epsilon^4)$, but this complexity still hides a condition number factor.

In contrast, our work achieves structure learning of local Lindbladians with $\tet = \mathcal{O}(gd^2\log(n)/\epsilon^2)$ and $\tres = \Theta(1/g)$, even for arbitrary dissipative terms and without condition number dependence.
Our total time evolution also scales polynomially in all parameters.
Moreover, for Lindbladians with algebraic decay, our guarantee holds for decay rates $\gamma > p$, beyond which there is a natural barrier~\cite{defenu2023long}.
Our algorithm also applies to general $k$-local Lindbladians, where we achieve $\tet = \mathcal{O}(gn^{2k-2}\log(n)/\epsilon^2)$.

\paragraph{Concurrent work.}
While preparing this manuscript, we became aware of several independent and concurrent works~\cite{romanov2026learning,arad2026near,sinha2026efficient,mobus2026robust,zhou2026characterizing} which study the problem of learning Lindbladians.
Three of these works~\cite{romanov2026learning,sinha2026efficient,zhou2026characterizing} operate in the sparse setting and are not comparable to our work.
On the one hand, they obtain structure learning guarantees for a broader class of Lindbladians without sparsity assumptions, but, on the other hand, their algorithms utilize a total evolution time which scales polynomially in the system size.
Moreover, \cite{romanov2026learning,sinha2026efficient} require at least $n$ ancillary qubits.
In contrast, our algorithm is ancilla-free and uses a total evolution time scaling logarithmically in system size.

\cite{arad2026near} addresses the same setting as our work, but their algorithm requires worse parameter dependencies.
In particular, for $\locality = \mathcal{O}(1)$,~\cite{arad2026near} uses a total evolution time of $\tet = \tilde{\mathcal{O}}(\Lambda \degree^{2\locality}_{\mathrm{dis}}\log(n)/\epsilon^2)$ and time resolution $\tres = \Theta(1/\Lambda)$ for Lindbladians with ``dissipative site degree'' $\degree_{\mathrm{dis}}$ and ``local dynamical strength'' $\Lambda$.
Their parameters of $\degree_{\mathrm{dis}}$ and $\Lambda$ are comparable to our parameters $d$ and $\energy$, respectively.
Moreover, \cite{arad2026near} does not consider learning Lindbladians with decaying long-range interactions.
In contrast, our algorithm achieves $\tet = \mathcal{O}(g d^2\log(n)/\epsilon^2)$, which is fixed-parameter tractable, and a similar time resolution of $\tres = \Theta(1/g)$.
We also obtain guarantees for learning Lindbladians with exponentially decaying and power-law decaying interactions.
\cite{arad2026near} also proves a lower bound of $\tet = \Omega(g/\epsilon^2)$, which shows optimality with respect to the $g$ and $\epsilon$ parameters.
Meanwhile, our lower bound of $\tet= \Omega(gd^2/(\epsilon^2 \log d))$ includes the $d$-dependence and shows the optimality of our algorithm with respect to all parameters.

\cite{mobus2026robust} also considers the local setting, and their algorithm has similar scalings as~\cite{arad2026near}.
Namely, for $\locality, g = \mathcal{O}(1)$, their result uses $\tet = \tilde{\mathcal{O}}(\degree^{2\locality}\log(n)/\epsilon^2)$.
It is not clear how their algorithm scales with $\energy$.
Also, for Lindbladians with algebraic decaying interactions, \cite{mobus2026robust} does not perform structure learning: the set of terms with large interaction strengths are provided as input to the algorithm.
In comparison, our result is fixed-parameter tractable, achieving a significantly better $\degree$ dependence.
Both of our results also hold for learning general $k$-local Lindbladians and achieve similar complexities.
For long-range interactions, our results are strictly stronger than~\cite{mobus2026robust}, as we perform structure learning and are not given the set of large terms.

The above comparisons are made with respect to the quantum resources required, i.e., the total time evolution and time resolution.
The classical time complexities of these works are all comparable with ours.
Namely, our classical runtime is $\widetilde{\mathcal{O}}(n^kg^2 d^2/\epsilon^2)$.
Meanwhile,~\cite{arad2026near} has a classical runtime of $\widetilde{\mathcal{O}}(\Lambda^2 n^k d_{\mathrm{dis}}^{2k}/\epsilon^2)$, and \cite{mobus2026robust} has a classical runtime of $\mathcal{O}(n^k)$.

\subsection{Discussion}
In this work, we develop a general framework for structure learning an unknown Lindbladian given access to its dynamics and show that our algorithm is optimal in all parameters up to logarithmic factors.
We show that our result can be instantiated in several physically motivated settings, including geometrically local, general $k$-local, quasi-local, and power-law Lindbladians.
We can also specialize our proof to apply to two problems in Hamiltonian learning.
Namely, we give a new, simple algorithm for structure learning Hamiltonians from real-time evolution.
In addition, we design the first algorithm for structure learning Hamiltonians from high-temperature Gibbs states.

There are still several interesting open questions to explore.
\begin{enumerate}
    \item What is the optimal complexity for structure learning Hamiltonians? By adapting our framework, we show a new guarantee of $\tet = \mathcal{O}(g\log(n)/\epsilon^2)$ and $\tres = \Theta(1/g)$ for Hamiltonian learning. Is the dependence on effective sparsity $r$ in~\cite{bakshi2024structure} required to obtain the Heisenberg scaling?
    Could this be related to the lack of inverse access in this model~\cite{tw25a}?
    \item Our techniques for structure learning Hamiltonians from high-temperature Gibbs states do not extend to the low-temperature regime because the cluster expansion diverges at low temperatures. Can one efficiently learn the structure of Hamiltonians from Gibbs states at any temperature? \cite{bakshi2024learning,chen2025learning} achieve analogous results in the parameter learning case.
\end{enumerate}

\section{Preliminaries}\label{sec:prelims}

Throughout, $[n] = \{1,\dots, n\}$, and $\widetilde{\mathcal{O}}(f) = \mathcal{O}(f\polylog(f))$.
We use the Iverson bracket: $\iver{G} = 1$ if the proposition $G$ is true and $0$ otherwise.
We denote the complement of a set $S$ by $\overline{S} \subseteq [n]$, where $\overline{S} \triangleq [n] \setminus S$.
For a matrix $M$, we use $M^\dagger$ to denote its conjugate transpose.
Given $\tau \geq 0$, we define
\begin{equation}
    \round_\tau(x) = \begin{cases}
        0 & |x| \leq \tau\\ x & \text{otherwise}
        \end{cases},
        \qquad
\text{and}
\qquad
\trunc_\tau(x) = \begin{cases}
        x & |x| \leq \tau\\ 0 & \text{otherwise}
        \end{cases}.
\end{equation}

Throughout, we also let $n$ denote the number of qubits, $N \triangleq 2^n$, and $\ntr \triangleq \tr/N$.
We let $\pauli{n} \triangleq \{I,X,Y,Z\}^{\otimes n}$ denote the set of $4^n$ $n$-qubit Pauli matrices.
For a Pauli $P \in \pauli{n}$, we denote $S_P \triangleq \supp(P)$, where $\supp(P) \triangleq \{i \in [n] : P_i \neq I\}$.
For a pair of Paulis $P_1, P_2 \in \pauli{n}$, we sometimes write $\pbar \triangleq (P_1, P_2)$ for brevity.
Moreover, similarly to the single Pauli case, we write $S_{\pbar} \triangleq \supp(P_1) \cup \supp(P_2)$.
We also write $s_P \triangleq |S_P|$ and $s_{\pbar} \triangleq |S_{\pbar}|$.
Finally, consider a matrix $M = a \cdot P$, where $a \in \{\pm 1, \pm \ii\}$ and $P \in \pauli{n}$;
such matrices arise from products of Pauli matrices.
Then we will write $c(M) = a$ and $\pauli{}(M) = P$, so that
\begin{equation}\label{eq:pauli-factor}
    c(M) \cdot \pauli{}(M) = M.
\end{equation}

For a vector $x$, we use $\norm{x}_\infty \triangleq \max_i |x_i|$ to denote the $\infty$-norm.
We use $\norm{x}_2^2 \triangleq \sum_i |x_i|^2$ to denote the $2$-norm.
For a matrix $M$, we use $\norm{M}_{\mathrm{op}} \triangleq \max_{x\neq 0} \frac{\norm{Mx}_2}{\norm{x}_2}$ to denote the spectral norm.
We also write $\norm{M}_{\tr} \triangleq \tr(\sqrt{M^\dagger M})$ to denote the trace norm and $\norm{\mathcal{M} - \mathcal{N}}_\diamond \triangleq \max_\rho \norm{(\mathcal{M} \otimes \mathrm{id})(\rho) - (\mathcal{N}\otimes \mathrm{id})(\rho)}_{\tr}$ to denote the diamond distance.
For a matrix $M \in \mathbb{C}^{N \times N}$, we define the operator norm corresponding to norms $\norm{\cdot}_a$ and $\norm{\cdot}_b$ on $\mathbb{C}^N$ as
\begin{equation}
    \norm{M}_{a\to b} \triangleq \sup_{\norm{x}_a \leq 1} \norm{Mx}_b.
\end{equation}
Note that $\norm{Mx}_b \leq \norm{M}_{a\to b}\norm{x}_a$, and $\norm{LM}_{a\to c} \leq \norm{L}_{b\to c}\norm{M}_{a\to b}$.
Moreover, we define a \emph{superoperator} as a linear map from $\mathbb{C}^{N\times N}$ to $\mathbb{C}^{N\times N}$.

\subsection{Lindbladians}
\label{sec:lindblad}

First, we define a Lindbladian, which defines the Markovian dynamics of an open quantum system via the Lindblad master equation.

\begin{definition}[Lindbladian]
    A \emph{Lindbladian} is a linear map $\mathcal{L}: \mathbb{C}^{N\times N} \to \mathbb{C}^{N \times N}$ that, applied to a quantum state $\rho \in \mathbb{C}^{N \times N}$, can be written as
    \begin{equation}
        \label{eq:lind}
        \mathcal{L}(\rho) = \frac{1}{2}\sum_{P\neq I} \alpha_P [P, \rho] + \sum_{P_1,P_2 \neq I} D_{P_1,P_2}\left(P_1 \rho P_2 - \frac{1}{2} \{P_2P_1, \rho\}\right),
    \end{equation}
    where $D$ is a positive semidefinite matrix and $\alpha$ is purely imaginary.
    This Lindbladian is \emph{$\locality$-local} if every $P$ satisfies $|\supp(P)| \leq \locality$ and every $P_1,P_2$ satisfies $|\supp(P_1) \cup \supp(P_2)| \leq \locality$.
\end{definition}

Throughout, we assume that we are working with nonzero Lindbladians.
In particular, we always assume that $k$ and $g$ are nonzero.
We often find it convenient in our analysis to consider combining the coherent and dissipative coefficients into a single vector, i.e., writing
\begin{equation}
    \label{eq:lindblad-combined-coefs}
    \mathcal{L}(\rho) = \sum_{\substack{P_1,P_2\\P_1 \neq I}} \lambda_{P_1,P_2}\left(P_1 \rho P_2 - \frac{1}{2}\{P_2 P_1, \rho\}\right),
\end{equation}
where now $P_2$ is allowed to be the identity.
This simply allows us to index into the vector $\lambda$ with a pair of Paulis instead of one, making our notation simpler throughout.
These are clearly equivalent definitions,\footnote{The reason we include the factor of $1/2$ in the definition in \Cref{eq:lind} is so that $\lambda_{P,I} = \alpha_P$, instead of $2\alpha_P$.} as one can set $\lambda_{P,I} = \alpha_P$ and $\lambda_{P_1,P_2} = D_{P_1,P_2}$.
To make the dependence on the coefficients explicit, we sometimes write $\mathcal{L}_{\alpha, D}$ or $\mathcal{L}_\lambda$.
When the subscript is a single vector, one should consider the representation in \Cref{eq:lindblad-combined-coefs}.

\begin{definition}[Parameterized Lindbladian]
    \label{def:param-lind}
    For a Lindbladian and a coefficient vector $(\alpha, D) \in \C^m$, we let $\lind_{\alpha, D}$ refer to the corresponding Lindbladian in \Cref{eq:lind}.
    Moreover, we denote by $\lind_\lambda$ the superoperator of the form \Cref{eq:lindblad-combined-coefs}, which we call a \emph{parameterized Lindbladian}.
    For notational simplicity, we will sometimes index into $\lambda$, not by an explicit Lindbladian term $P_1, P_2$ (where $P_2$ is arbitrary and $P_1 \neq I$), but by an index $a \in [m]$.
    The term associated to $a$ will then be denoted $\pbar_a$.
\end{definition}

Though we sometimes refer to $\lind_\lambda$ as a Lindbladian, for general $\lambda$, this may not define a valid physical Lindbladian and may only refer to a superoperator mapping $\C^{N \times N} \to \C^{N \times N}$.
We also often use the adjoint of the Lindbladian.

\begin{definition}[Adjoint of a Lindbladian]
    The \emph{adjoint} $\mathcal{L}^\dagger$ of a Lindbladian $\mathcal{L}$ is defined such that $\tr(X \mathcal{L}(Y)) = \tr(\mathcal{L}^\dagger(X) Y)$ for any operators $X,Y$.
    Explicitly, using the Pauli expansion above, one can write the adjoint as
    \begin{equation}
        \mathcal{L}^\dagger(O) = -\frac{1}{2}\sum_{P\neq I} \alpha_P [P, O] + \sum_{P_1,P_2 \neq I} D_{P_1,P_2} \left(P_2 O P_1 - \frac{1}{2}\{P_2P_1, O\}\right),
    \end{equation}
    where $O \in \mathbb{C}^{N\times N}$ is an observable. 
\end{definition}

The analogous definition for $\mathcal{L}_\lambda$ is clear.
Similarly to \cite{bakshi2024structure}, our results depend on a ``local norm,'' which is defined as follows.

\begin{definition}[Local norm of a Lindbladian]
    Let $\mathcal{L}_{\lambda}$ be a superoperator.
    Then, we define the \emph{$B_1$-norm} of $\lambda$ as
    \begin{equation}
        \lonorm{\lambda} \triangleq \max_{i \in [n]} \left(\sum_{\pbar : S_{\pbar} \ni i} |\lambda_{\pbar}|\right).
    \end{equation}
\end{definition}
We sometimes write $\lonorm{\mathcal{L}_\lambda} = \lonorm{\lambda}$.
Moreover, note that $\lonorm{\lind_\lambda} = \lonorm{\lind_{\alpha, D}}$.

\begin{definition}[Degree and approximate degree of a Lindbladian]
    The \emph{degree} $\deg(\lind_\lambda)$ of a superoperator $\mathcal{L}_\lambda$ is the maximum number of terms supported on a site, i.e.,
    \begin{equation}
        \deg(\lind_\lambda) \triangleq \max_{i \in [n]} |\{\pbar : i \in S_{\pbar},\; \lambda_{\pbar} \neq 0\}|
    \end{equation}

    The \emph{approximate degree} $\deg_\eps(\lind_\lambda)$ of a superoperator $\lind_\lambda$ is the minimum $\deg(\Lbig)$ over ways to split $\lind_\lambda = \Lbig + \Lsmall$ such that $\lonorm{\Lsmall} < \eps$, i.e.,
    \begin{equation}
        \deg_\eps(\lind_\lambda) \triangleq \min_{\Lbig : \lonorm{\lind_\lambda - \Lbig} < \epsilon} \deg(\Lbig).
    \end{equation}
    Here, $\Lsmall = \lind - \Lbig$.
    Sometimes, we also write $\deg(\lambda) = \deg(\lind_\lambda)$ and $\deg_\eps(\lambda) = \deg_\eps(\lind_\lambda)$.
\end{definition}

Consider the optimal splitting $\lind_\lambda = \Lbig + \Lsmall$, and let $\lambdabig$ and $\lambdasmall$ denote the coefficients of $\Lbig$ and $\Lsmall$, respectively.
We observe that, without loss of generality, the supports of $\lambdabig$ and $\lambdasmall$ are disjoint.
This is because, if a coefficient is nonzero in both, then one could remove it from $\lambdasmall$ and keep it in $\lambdabig$ while maintaining the same approximate degree.

\begin{remark}
    \label{rmk:deg-bound}
    We note that $\deg(\round_\tau(\lambda)) \leq \lonorm{\lambda}/\tau$.
\end{remark}

\subsection{Fourier analysis of quantum channels} 

We begin by recalling standard facts about the Fourier analysis of quantum channels.
For a more thorough introduction to this topic, we refer the reader to~\cite{bao2023testing}.

Consider a superoperator $\Phi$ which acts on $n$-qubit states.
We can expand $\Phi$ in terms of Pauli matrices by defining the functions
\begin{equation}
    \chi_{P_1, P_2}(\rho) = P_1 \rho P_2,
\end{equation}
where $P_1, P_2 \in \{I, X, Y, Z\}^{\otimes n}$.
Such $\chi_{P_1, P_2}$ form a basis for the set of superoperators~\cite[Proposition 6]{bao2023testing},
so $\Phi$ can be written as
\begin{equation}
    \Phi = \sum_{P_1, P_2} \widehat{\Phi}(P_1, P_2) \cdot \chi_{P_1, P_2},
\end{equation}
where the $\widehat{\Phi}(P_1, P_2)$'s are referred to as the \emph{Fourier coefficients} of the superoperator $\Phi$.
The next lemma shows that these functions are actually orthonormal to each other.

\begin{lemma}[The Fourier basis is orthonormal]
    \label{lem:fourier-orthonormal}
    Let $P_1, P_2, Q_1, Q_2 \in \pauli{n}$.
    Then
    \begin{equation}
    \E_{R \sim \pauli{n}}\left[\ntr\Big( \chi_{P_1, P_2}(R)^{\dagger}\cdot \chi_{Q_1, Q_2}(R)\Big)\right] = \begin{cases}
        1 & \text{if $P_1 = Q_1$ and $P_2 = Q_2$}\\
        0 & \text{otherwise}.
    \end{cases},
\end{equation}
\end{lemma}

\begin{proof}
    If $P_1 = Q_1$ and $P_2 = Q_2$, 
    \begin{equation}
        \E_R\left[\ntr\Big(\chi_{P_1, P_2}(R)^{\dagger} \cdot \chi_{Q_1, Q_2}(R)\Big)\right] = \E_R\left[\ntr(P_2 R P_1 \cdot Q_1 R Q_2)\right] = \E_R[\ntr(R R)] = 1.
    \end{equation}
    Otherwise, let us assume without loss of generality that $P_1 \neq Q_1$.
    Then, $P_1 Q_1$ is some nonidentity Pauli matrix, and a random Pauli $R$ will commute with $P_1 Q_1$ with probability $1/2$ and anticommute with probability $1/2$.
    Thus, half of the time, we have
    \begin{equation}
        \ntr\Big(\chi_{P_1, P_2}(R)^{\dagger}\cdot \chi_{Q_1, Q_2}(R) \Big) = \ntr(P_2 R P_1 Q_1 R Q_2) = \ntr(P_2 P_1 Q_1 RR Q_2) = \ntr(P_2 P_1 Q_1 Q_2),
    \end{equation}
    and the other half of the time, we have
    \begin{equation}
        \ntr\Big(\chi_{P_1, P_2}(R)^{\dagger}\cdot \chi_{Q_1, Q_2}(R) \Big) = \ntr(P_2 R P_1 Q_1 R Q_2) = -\ntr(P_2 P_1 Q_1 RR Q_2) = -\ntr(P_2 P_1 Q_1 Q_2).
    \end{equation}
    These two average out to zero, which completes the proof.
\end{proof}

\subsubsection{Local Fourier coefficients}
\label{sec:local-fourier}

By linearity, \Cref{lem:fourier-orthonormal} gives us the following Fourier inversion formula for the Fourier coefficients of $\Phi$:
\begin{equation}\label{eq:fourier-inversion}
    \widehat{\Phi}(P_1, P_2) = \E_R\left[\ntr\Big(\chi_{P_1, P_2}(R)^{\dagger} \cdot \Phi(R)\Big)\right].
\end{equation}
In principle, this suggests a natural experiment that we could carry out to learn the Fourier coefficient $\widehat{\Phi}(P_1, P_2)$ of $\Phi$ in the case that $\Phi$ is a quantum channel.
However, for our application we will be interested in efficiently estimating many expectations of this form,
and in this case it will only be possible to do so if we restrict the random Paulis $R$ in these expectations to have local support.
Motivated by this, let us first show an analogue of \Cref{lem:fourier-orthonormal} for Paulis with local support.

\begin{lemma}[Local orthonormality relations]\label{lem:restricted-orthogonality}
    Let $P_1, P_2, Q_1, Q_2 \in \pauli{n}$. 
    Let $S \subseteq [n]$ contain $S_{\pbar}$.
    Then
    \begin{equation}
        \E_{R \sim \pauli{S}}\left[\ntr\Big(\chi_{P_1, P_2}(R)^{\dagger} \cdot \chi_{Q_1, Q_2}(R)\Big)\right] = \begin{cases}
        1 & \text{if $Q_1^S = P_1$, $Q_2^S = P_2$, and $Q_1^{\overline{S}} = Q_2^{\overline{S}}$},\\
        0 & \text{otherwise}.
    \end{cases}
    \end{equation}
\end{lemma}
\begin{proof}
    By definition,
    \begin{align}
        \E_{R \sim \pauli{S}}\left[\ntr\Big(\chi_{P_1, P_2}(R)^{\dagger} \cdot \chi_{Q_1, Q_2}(R)\Big)\right]
        &=\E_{R \sim \pauli{S}}\left[\ntr(P_2 R P_1 Q_1 R Q_2)\right]\\
        &= \E_{R \sim \pauli{S}}\left[\ntr(P_2^S R P_1^S Q_1^S R Q_2^S)\right] \cdot \overline{\mathrm{tr}}(Q_1^{\overline{S}} Q_2^{\overline{S}})\\
        &= \E_{R \sim \pauli{S}}\left[\ntr\Big(\chi_{P_1^S, P_2^S}(R)^{\dagger} \cdot \chi_{Q_1^S, Q_2^S}(R)\Big)\right] \cdot \overline{\mathrm{tr}}(Q_1^{\overline{S}} Q_2^{\overline{S}}).
    \end{align}
    By \Cref{lem:fourier-orthonormal},
    the expectation is $1$ if $Q_1^S = P_1^S$ and $Q_2^S = P_2^S$ and 0 otherwise.
    In addition, the second term is $1$ if $Q_1^{\overline{S}} = Q_2^{\overline{S}}$ and 0 otherwise.
    This completes the proof.
\end{proof}

It will turn out that we only ever care about the case in which $S = S_{\pbar}$.
In this case, \Cref{lem:restricted-orthogonality} says that the functions $\chi_{P_1, P_2}$ are no longer orthonormal over Paulis $R$ restricted to $S_{\pbar}$;
instead, $\chi_{P_1, P_2}$ can be ``confused'' for certain other functions $\chi_{Q_1, Q_2}$ which agree with it on $S_{\pbar}$.
We will use $\qbar \succeq \pbar$ to mean that $\qbar$ can be confused with $\pbar$, i.e.
\begin{equation}
    \qbar \succeq \pbar
    \quad
    \iff
    \quad
    \text{$Q_1^{S_{\pbar}} = P_1$, $Q_2^{S_{\pbar}} = P_2$, and $Q_1^{\overline{S_{\pbar}}} = Q_2^{\overline{S_{\pbar}}}$}
\end{equation}
We define the ``local Fourier coefficients'' of $\Phi$ via
\begin{equation}
    \widehat{\Phi}_{\mathsf{loc}}(P_1, P_2) \triangleq \E_{R\sim \pauli{S_{\pbar}}}\left[\ntr\Big(\chi_{P_1, P_2}(R)^{\dagger} \cdot \Phi(R)\Big)\right].
\end{equation}
Note that, unlike the actual Fourier coefficients of $\Phi$, these do not have an interpretation as the coefficients of $\Phi$ when expanded in a particular basis of superoperators.
Instead, we only have that
\begin{equation}\label{eq:not-fourier-inversion}
    \widehat{\Phi}_{\mathsf{loc}}(P_1, P_2) = \sum_{\qbar \succeq \pbar} \widehat{\Phi}(Q_1, Q_2).
\end{equation}
However, it turns out that the actual Fourier coefficients of $\Phi$ can still be recovered from these ``local Fourier coefficients'' by a simple linear transformation.
We discuss this in \Cref{sec:local-lindblad-fourier} below.

\subsubsection{Estimating the local Fourier coefficients}
\label{sec:estimate-coefs}

We conclude by designing an algorithm to estimate the local Fourier coefficients via random Pauli measurements.
First, let us establish some notation which will help to specify the algorithm.
Given a single qubit Pauli $P \in \{X, Y, Z\}$ and a bit $b \in \{\pm 1\}$, we write $\ket{P, b}$ for the eigenstate of $P$ with eigenvalue $b$.
We also set $\ket{I, + 1} \triangleq \ket{0}$ and $\ket{I,- 1} \triangleq \ket{1}$, and note that both of these are $+1$ eigenstates of $I$.
For an $n$-qubit Pauli $Q$ and a vector $v \in \{\pm 1\}^n$, we define the vector $\ket{Q, v}$ by extending the single-qubit definition via the tensor product.
We note that $\ket{Q, v}$ is an eigenvector of $Q$ with eigenvalue $\chi_{S_Q}(v)$, where
\begin{equation}
    \chi_{S_Q}(v) = \prod_{i \in S_Q} v_i
\end{equation}
is the standard Boolean Fourier character. (The fact that we take the product of the $v_i$'s only within the support of $Q$ accounts for the fact that $Q$ may have coordinates which are equal to $I$.)
As a result, we have the eigendecomposition
\begin{equation}\label{eq:eigen}
    Q = \sum_{v \in \{\pm 1\}^n} \chi_{S_Q}(v) \cdot \ketbra{Q, v}{Q,v}.
\end{equation}
We are now ready to state our algorithm.

\SetKwInput{KwData}{Input}
\SetKwInput{KwResult}{Output}

\begin{algorithm}
    \caption{Local Fourier coefficient estimation}
    \label{alg:local-fourier}
    \KwData{Black-box ability to apply an unknown $n$-qubit quantum channel $\Phi$; locality parameter $k$; error parameters $\epsilon, \delta > 0$.}
    \KwResult{Estimates $\widehat{E}_{P_1, P_2}$ for all $P_1, P_2 \in \pauli{n}$ with $s_{\pbar} \leq k$ such that
    \begin{equation}
    |\widehat{E}_{P_1, P_2} - \widehat{\Phi}_{\mathsf{loc}}(P_1, P_2)| \leq \epsilon.
    \end{equation}}
    Set $M = \Theta(C^k \log(n/\delta) /\epsilon^2)$ for some absolute constant $C > 0$.\\
    \For{$i=1,\dots, M$}{
        Initialize a uniformly random Pauli eigenstate $\ket{A^{(i)}, v^{(i)}}$.\\
        Apply $\Phi$ to this state.\\
        Sample a uniformly random $B^{(i)} \in \pauli{n}$ and measure in this basis.\\
        Record the outcome $w^{(i)} \in \{\pm 1\}^n$.
    }
    \For{$P_1, P_2$ with $s_{\pbar} \leq k$}{
        For each $i = 1, \ldots, M$, define the Pauli matrices
    \begin{equation}
        R^{(i)} = (A^{(i)})^{S_{\pbar}} \otimes I^{\overline{S_{\pbar}}},
        \quad
        \text{and}
        \quad
        Z^{(i)} = (B^{(i)})^{S_{\pbar}} \otimes I^{\overline{S_{\pbar}}}.
    \end{equation}
    If $Z^{(i)} = \pauli{}(P_2 R^{(i)} P_1)$, set
        \begin{equation}
            E_{P_1, P_2}^{(i)} =
            4^{s_{\pbar}} \cdot c(P_2 R^{(i)} P_1) \cdot\chi_{S_{R^{(i)}}}(v^{(i)}) \cdot \chi_{S_{Z^{(i)}}}(w^{(i)}).
        \end{equation}
        Otherwise, set it to 0.\\
    Compute the estimate
    \begin{equation}
        \widehat{E}_{P_1, P_2} \triangleq \frac{1}{M} \cdot \sum_{i=1}^M E_{P_1, P_2}^{(i)}.
    \end{equation}
    }
\end{algorithm}

\begin{proposition}[Single sample unbiased estimators]
Let $P_1, P_2 \in \pauli{n}$ with $s_{\pbar} \leq k$, and let $1 \leq i \leq M$. Then the quantity $E_{P_1, P_2}^{(i)}$ from \Cref{alg:local-fourier} is an unbiased estimator for $\widehat{\Phi}_{\mathsf{loc}}(P_1, P_2)$.
\end{proposition}
\begin{proof}
    We drop the $(i)$ superscript from our Pauli matrices for notational convenience.
    We also write $S$ for $S_{\pbar}$ and $s$ for $s_{\pbar}$.
    
    First, we consider the expectation of $E^{(i)}_{P_1, P_2}$ conditioned on a fixed $A,B$.
    If $Z \neq \pauli{}(P_2 R P_1)$, then $E^{(i)}_{P_1, P_2}$ is set to 0, so the expectation is 0.
    Otherwise, we receive the measurement outcome $\ket{B, w}$ with probability
    \begin{equation}
        \tr\Big(\ketbra{B, w}{B,w} \cdot \Phi(\ketbra{A, v}{A,v})\Big).
    \end{equation}
    Thus, we can write the expectation as
    \begin{align}
        \E[E_{P_1, P_2}^{(i)} \mid A, B]
        &= \E_{v \in \{\pm 1\}^n}\sum_{w \in \{\pm 1\}^n} 4^s \cdot c(P_2 R P_1) \cdot \chi_{S_R}(v) \cdot \chi_{S_Z}(w)\cdot \tr\Big(\ketbra{B, w}{B,w} \cdot \Phi(\ketbra{A, v}{A,v})\Big)\\
        &=2^{-n}4^s \cdot c(P_2 R P_1) \cdot\tr\Big(\Big(\sum_w \chi_{S_Z}(w) \cdot  \ketbra{B, w}{B,w}\Big) \cdot \Phi(\sum_v \chi_{S_R}(v) \cdot\ketbra{A, v}{A,v}\Big)\Big)\\
        &=2^{-n}4^s \cdot c(P_2 R P_1) \cdot\tr(Z \cdot \Phi(R)) \\
        &=4^s\ntr(P_2 R P_1 \cdot \Phi(R)),
    \end{align}
    where we used \Cref{eq:eigen,eq:pauli-factor} in the last two steps.
    Now, note that as $A$ varies uniformly over $\pauli{n}$, $R$ varies uniformly over $\pauli{S}$.
    Furthermore, conditioned on the value of $A$, we have $Z = \pauli{}(P_2 R P_1)$ with probability exactly $1/4^s$.
    Hence,
    \begin{equation}
        \E[E_{P_1, P_2}^{(i)}] = \E_{R\sim \pauli{S_{\pbar}}}\left[\ntr(\chi_{P_1, P_2}(R)^{\dagger} \cdot \Phi(R))\right].
    \end{equation}
    This concludes the proof.
\end{proof}

\begin{lemma}
    \label{lem:estimate-coefs}
    Let $k > 0$ be a locality parameter.
    Let $M = \Theta(C^k\log(n/\delta)/\epsilon^2)$, where $C > 0$ is an absolute constant.
    Then, the outputs of
    \Cref{alg:local-fourier} satisfy
    \begin{equation}
    |\widehat{E}_{P_1, P_2} - \widehat{\Phi}_{\mathsf{loc}}(P_1, P_2)| \leq \epsilon
    \end{equation}
    for all $P_1, P_2 \in \pauli{n}$ with $s_{\pbar} \leq k$ with probability at least $1-\delta$, using $M$ queries to the channel $\Phi$.
    Moreover, \Cref{alg:local-fourier} runs in time $\mathcal{O}(M n^k)$.
\end{lemma}
\begin{proof}
    Consider a fixed $P_1$ and $P_2$.
    Each $E^{(i)}_{P_1, P_2}$ is an unbiased estimator for $\widehat{\Phi}_{\mathsf{loc}}(P_1, P_2)$ which is bounded in magnitude by $4^k$. As a result, Hoeffding's inequality, applied to the real and imaginary parts of the estimator, implies that
    \begin{equation}
        \Pr[|\widehat{E}_{P_1, P_2} - \widehat{\Phi}_{\mathsf{loc}}(P_1, P_2)| \geq \epsilon] \leq 4 e^{-4M \epsilon^2/16^k}.
    \end{equation}
    Now, the number of $P_1, P_2 \in \pauli{n}$ with $s_{\pbar} \leq k$ is at most $\binom{n}{k} \cdot 16^k \leq (16n)^k$.
    Hence, by the union bound, the probability that there exists an $E_{P_1, P_2}$ with error more than $\epsilon$ is at most
    \begin{equation}
        (16n)^k \cdot 2 e^{-2M \epsilon^2/16^k} \leq \delta,
    \end{equation}
    by our choice of $M$. This completes the proof.
\end{proof}

\section{Local Fourier coefficients of the time evolution operator}\label{sec:local-lindblad-fourier}

An important ingredient of this work is the local Fourier coefficients of the time evolution operator corresponding to our Lindbladian $e^{\calL_x t}(\cdot)$. 
We begin by introducing some notation we will use to represent these local Fourier coefficients.

\begin{definition}[Vector of expectation values]
    For a (parameterized) Lindbladian $\lind_x$ and a time $t \in \R$, we define the vector $E: \C^m \to \C^m$ as follows:
    \begin{align}
        (E(x))_{P_1, P_2} = \E_{R \sim \pauli{S_{\pbar}}}[\ntr(\chi_{P_1, P_2}(R)^{\dagger} \cdot e^{\calL_x t}(R))].
    \end{align}
\end{definition}
By a Taylor series expansion, we can write
\begin{equation}
    e^{\calL_x t}(R) = \sum_{\ell=0}^\infty \frac{t^\ell}{\ell!} \calL_x^{\ell}(R),
\end{equation}
where $R \in \pauli{n}$.
Hence, we have
\begin{align}
    (E(x))_{P_1, P_2}
    &= \sum_{\ell=0}^\infty \frac{t^\ell}{\ell!}\E_{R \sim \pauli{S_{\pbar}}}[\ntr(\chi_{P_1, P_2}(R)^{\dagger} \cdot \calL_x^{\ell}(R))]\\
    &= \sum_{\ell=1}^\infty \frac{t^\ell}{\ell!}\E_{R \sim \pauli{S_{\pbar}}}[\ntr(\chi_{P_1, P_2}(R)^{\dagger} \cdot \calL_x^{\ell}(R))].\label{eq:remove-zero}
\end{align}
In the second step, we used the fact that the $\ell = 0$ expectation is
\begin{equation}
    \E_{R \sim \pauli{S_{\pbar}}}[\ntr(\chi_{P_1, P_2}(R)^{\dagger} \cdot \calL_x^{0}(R))]
    = \E_{R \sim \pauli{S_{\pbar}}}[\ntr(\chi_{P_1, P_2}(R)^{\dagger} \cdot R)]
    = 0,
\end{equation}
due to \Cref{lem:restricted-orthogonality}, and the fact that at least one of $P_1, P_2$ is non-identity.
Precisely understanding the infinite sum in \Cref{eq:remove-zero} is challenging;
however, we show in \Cref{sec:series} that it is well-approximated by its linear term (the $\ell=1$ term).
Motivated by this, we dedicate this section to understanding this linear term.

The $\ell = 1$ term in \Cref{eq:remove-zero} is given by
\begin{equation}
    t\cdot \E_{R \sim \pauli{S_{\pbar}}}[\ntr(\chi_{P_1, P_2}(R)^{\dagger} \cdot \calL_x(R))]
\end{equation}
Recalling the definition of a (parameterized) Lindbladian, we have
\begin{align}
    \mathcal{L}_x(R) &=\sum_{Q_1\neq I, Q_2} x_{Q_1,Q_2}\left(Q_1 R Q_2 - \frac{1}{2} \{Q_2Q_1, R\}\right)\\
    &=\sum_{Q_1\neq I, Q_2} x_{Q_1,Q_2}\left(\chi_{Q_1, Q_2}(R) - \frac{1}{2} c(Q_2 Q_1) \cdot \chi_{\pauli{}(Q_2 Q_1), I}(R)- \frac{1}{2} c(Q_2 Q_1) \cdot \chi_{I, \pauli{}(Q_2 Q_1)}(R)\right).\label{eq:whatevs}
\end{align}
Now, we can use \Cref{lem:restricted-orthogonality} to calculate the expectation.
The simplest case is when $P_1, P_2 \neq I$. In this case,
\begin{equation}\label{eq:not-i}
    \E_{R \sim \pauli{S_{\pbar}}}[\ntr(\chi_{P_1, P_2}(R)^{\dagger} \cdot \calL_x(R))]
    = \sum_{\qbar \succeq \pbar} x_{\qbar}.
\end{equation}
On the other hand, let $\pbar = (P, I)$, with $P \neq I$ (we do not use the $\pbar = (I, P)$ case).
Of the three terms in \Cref{eq:whatevs}, the first can be ``confused'' with $(P, I)$ exactly when $\qbar \succeq (P, I)$;
the second can be ``confused'' with $(P, I)$ when $\pauli{}(Q_2Q_1)=P$;
and the third can never be ``confused'' with $(P, I)$.
In the second of these cases, note that $Q_2 = \pauli{}(PQ_1)$
and $c(Q_2 Q_1) = \overline{c(PQ_1)}$;
this is because
\begin{equation}
    P Q_1
    = c(P Q_1) \pauli{}(P Q_1)
    = c(P Q_1) Q_2
    \quad
    \Rightarrow
    \quad
    Q_2 Q_1
    = \overline{c(P Q_1)} P
    = c(Q_2 Q_1) P.
\end{equation}
As a result,
 \begin{align}\label{eq:ok-i}
    \E_{R \sim \pauli{S_{\pbar}}}[\ntr(\chi_{P, I}(R)^{\dagger} \cdot \calL_x(R))]
    &= \sum_{\qbar\succeq(P,I)}x_{\qbar}
      -\frac12\sum_{\substack{Q_1\neq I,\,Q_2,\\\pauli{}(Q_2Q_1)=P}}
      c(Q_2Q_1)\cdot x_{Q_1,Q_2}\\
      &= \sum_{\qbar\succeq(P,I)}x_{\qbar}
      -\frac12\sum_{Q \neq I}
      \overline{c(P Q)}\cdot x_{Q, \pauli{}(PQ)}.
 \end{align}
To help us analyze these expressions, we introduce the following notation.
\begin{notation}
    \label{not:A}
    For $0 \leq k \leq n$, we write $A_k$ for the square matrix whose rows and columns are indexed by pairs $P_1, P_2$ with $P_1 \neq I$ which acts as follows:
    \begin{equation}
        (A_k x)_{\pbar} \triangleq \E_{R \sim \pauli{S_{\pbar}}}[\ntr(\chi_{\pbar}(R)^{\dagger} \cdot \calL_x(R))].
    \end{equation}
    From \Cref{eq:not-i,eq:ok-i}, we have
    \begin{equation}\label{eq:a-def}
        (A_k x)_{\pbar} =
        \left\{\begin{array}{ll}
        \sum_{\qbar \succeq \pbar} x_{\qbar}& \text{if $P_1, P_2 \neq I$},\\
        \sum_{\qbar\succeq(P,I)}x_{\qbar}
      -\frac12\sum_{Q \neq I}
      \overline{c(P Q)}\cdot x_{Q, \pauli{}(PQ)}& \text{if $\pbar = (P, I)$}.
        \end{array}\right.
    \end{equation}
\end{notation}
\noindent
Applying this notation to \Cref{eq:remove-zero}, we have that
\begin{equation}\label{eq:first-order}
    E_{\pbar}(x) = t \cdot (A_k x)_{\pbar} + \sum_{\ell=2}^\infty \frac{t^\ell}{\ell!}\E_{R \sim \pauli{S_{\pbar}}}[\ntr(\chi_{P_1, P_2}(R)^{\dagger} \cdot \calL_x^{\ell}(R))].
\end{equation}

The linear $\ell = 1$ term of our expansion has a nice interpretation as the local Fourier coefficients of the Lindbladian $\calL_x$.
From \Cref{eq:a-def}, we see that these expressions are a linear combination of the true Lindbladian parameters $x$.
We are interested in the inverse of $A$, i.e., how to recover the true Lindbladian parameters if we know either the local Fourier coefficients or approximations of them.
To understand this, we first introduce the following matrix.

\begin{notation}
    For $0 \leq k \leq n$, we write $V_k$ for the $m \times m$ matrix which acts as follows. For $P_1, P_2 \neq I$,
    \begin{align}
        (V_k y)_{\pbar}
        \triangleq \sum_{\qbar \succeq \pbar} (-1)^{s_{\qbar} - s_{\pbar}} y_{\qbar}.
    \end{align}
    Otherwise, if $\pbar = (P, I)$,
    \begin{align}\label{eq:v-def}
        (V_k y)_{P, I} \triangleq 2 y_{P, I} + \sum_{\substack{Q, S_Q \subseteq S_P \\ Q \neq I, P}} \overline{c(PQ)}\cdot y_{Q, \pauli{}(PQ)} +   \sum_{\substack{\qbar \succeq (P, I)\\\qbar \neq (P, I)}} (-1)^{s_{\qbar}- s_{P}} y_{\qbar}-   \sum_{\substack{\qbar \succeq (I, P)\\\qbar \neq (I, P)}} (-1)^{s_{\qbar}- s_{P}} y_{\qbar}.
    \end{align}
\end{notation}

Next, we show that the $A_k$ matrix is invertible and that its inverse is equal to $V_k$.
This implies that it is possible to recover the true Lindbladian parameters if we know the local Fourier coefficients exactly.
To begin, we need the following helper lemma.

\begin{lemma}[Helper lemma] \label{lem:helper}
    Suppose that $\rbar \succeq \pbar$. Then
    \begin{equation}
         \sum_{\substack{\qbar: \qbar\succeq \pbar\\\rbar\succeq \qbar}} (-1)^{s_{\qbar}-s_{\pbar}}
         = \left\{\begin{array}{rl}
         1 & \text{if $\rbar = \pbar$},\\
         0 & \text{otherwise}.
         \end{array}\right.
    \end{equation}
\end{lemma}
\begin{proof}
    The pairs $\qbar$ which satisfy $\rbar \succeq \qbar$ and $\qbar \succeq \pbar$ are exactly those which (i) agree with $\pbar$ on $S_{\pbar}$,
    (ii) agree with $\rbar$ on some subset $T$ of $S_{\rbar} \setminus S_{\pbar}$,
    and (iii) are identity on the remaining qubits in $S_{\rbar}$.
    As $\rbar$ is non-identity on every qubit in $T$,
    we have $s_{\qbar} - s_{\pbar} = |T|$. Thus,
    \begin{equation}
        \sum_{\substack{\qbar: \qbar\succeq \pbar\\\rbar\succeq \qbar}} (-1)^{s_{\qbar}-s_{\pbar}}
        =
        \sum_{T\subseteq S_{\rbar}\setminus S_{\pbar}}(-1)^{|T|}
        =
        \left\{\begin{array}{rl}
         1 & \text{if $\rbar = \pbar$},\\
         0 & \text{otherwise}.
         \end{array}\right.
    \end{equation}
    This completes the proof.
\end{proof}

To prove that $V_k$ is the inverse of $A_k$,
we first show that it successfully recovers any Lindbladian parameter $x_{P_1, P_2}$ with $P_1, P_2 \neq I$.
We remark that this can be seen as an analogue of Theorem 8.35 in~\cite{o2021analysis}, which proves a similar statement for Boolean functions.

\begin{lemma}
    For any $P_1, P_2 \neq I$, we have $(V_k A_k x)_{\pbar} = x_{\pbar}$.
\end{lemma}
\begin{proof}
    To see this,
    \begin{align}
        (V_k A_k x)_{\pbar}
        &= \sum_{\qbar \succeq \pbar} (-1)^{s_{\qbar}-s_{\pbar}}(A_kx)_{\qbar}\\
        &= \sum_{\qbar \succeq \pbar} (-1)^{s_{\qbar}-s_{\pbar}}\sum_{\rbar \succeq \qbar}x_{\rbar}\\
        &= \sum_{\rbar\succeq \pbar}x_{\rbar}
            \sum_{\substack{\qbar: \qbar\succeq \pbar\\\rbar\succeq \qbar}} (-1)^{s_{\qbar}-s_{\pbar}}
        = x_{\pbar},
    \end{align}
    where the last step used \Cref{lem:helper}. This completes the proof.
\end{proof}

Next, we show that $V_k$ also recovers any Lindbladian parameter $x_{\pbar}$ with $\pbar = (P, I)$.

\begin{lemma}
    For any $P \neq I$, we have $(V_k A_k x)_{P, I} = x_{P, I}$.
\end{lemma}
\begin{proof}
    To see this, note that $(V_k A_k x)_{P, I}$ is equal to
    \begin{equation}\label{eq:real-big-sum}
        2 (A_k x)_{P, I} + \sum_{\substack{Q, S_Q \subseteq S_P \\ Q \neq I, P}} \overline{c(PQ)}\cdot (A_k x)_{Q, \pauli{}(PQ)} +   \sum_{\substack{\qbar \succeq (P, I)\\\qbar \neq (P, I)}} (-1)^{s_{\qbar}- s_{P}} (A_k x)_{\qbar}-   \sum_{\substack{\qbar \succeq (I, P)\\\qbar \neq (I, P)}} (-1)^{s_{\qbar}- s_{P}} (A_k x)_{\qbar}.
    \end{equation}
    Note that only the first term involves indexing $(A_k x)$ by a pair of Paulis, one of which is $I$;
    the other terms always index by two non-identity Paulis.
    Hence, the first two terms are equal to
    \begin{equation}\label{eq:first-two-combined}
        2\sum_{\qbar\succeq(P,I)}x_{\qbar}
      -\sum_{Q \neq I}
      \overline{c(P Q)}\cdot x_{Q, \pauli{}(PQ)}
      + 
        \sum_{\substack{Q, S_Q \subseteq S_P \\ Q \neq I, P}} \overline{c(PQ)}\cdot \sum_{\rbar \succeq (Q, \pauli{}(PQ))} x_{\rbar}.
    \end{equation}
    Note that in the second summation, the pair $(Q, \pauli{}(PQ))$ has support equal to $S_P$ and is identity outside of it.
    This means that $\rbar$ is equal to $(Q, \pauli{}(PQ))$ within $S_P$, and $R_1$ and $R_2$ agree outside of $S_P$.
    This means that (i) $\rbar = (R_1, \pauli{}(P R_1))$,
    (ii) $c(P R_1) = c(P Q)$,
    and (iii) $\rbar$ ranges over all possible pairs of this form, subject to $R_1|_{S_P}$ not being $I$ or $P$.
    Hence,
    \begin{align}
        \eqref{eq:first-two-combined}
        & = 2\sum_{\qbar\succeq(P,I)}x_{\qbar}
      -\sum_{Q \neq I}
      \overline{c(P Q)}\cdot x_{Q, \pauli{}(PQ)}
      + 
        \sum_{R : R|_{S_P} \neq I, P} \overline{c(P R)} \cdot x_{R, \pauli{}(PR)}\\
        &=2\sum_{\qbar\succeq(P,I)}x_{\qbar}
        -\sum_{\substack{Q \neq I,\\Q|_{S_P} = I \text{ or }P}} \overline{c(P Q)}\cdot x_{Q, \pauli{}(PQ)}\\
        &=2\sum_{\qbar\succeq(P,I)}x_{\qbar}
        -\sum_{\substack{Q \neq I,\\Q|_{S_P} = I \text{ or }P}}  x_{Q, PQ}\\&=2\sum_{\qbar\succeq(P,I)}x_{\qbar}
        - x_{P, I} - \sum_{\substack{Q \neq I,P\\Q|_{S_P} = I \text{ or }P}}  x_{Q, PQ},
    \end{align}
    where in the second-to-last line we use the fact that if $Q_{S_P} = I$ or $P$, then $PQ$ is a Pauli, and so $\pauli{}(PQ) = PQ$ and $c(PQ) = 1$.
    Now, the third term in \Cref{eq:real-big-sum} is equal to
    \begin{align}
        \sum_{\substack{\qbar \succeq (P, I)\\\qbar \neq (P, I)}} (-1)^{s_{\qbar}- s_{P}} (A_k x)_{\qbar}
        &= \sum_{\substack{\qbar \succeq (P, I)\\\qbar \neq (P, I)}} (-1)^{s_{\qbar}- s_{P}} \sum_{\rbar \succeq \qbar} x_{\rbar}\\
        &= \sum_{\qbar \succeq (P, I)} (-1)^{s_{\qbar}- s_{P}} \sum_{\rbar \succeq \qbar} x_{\rbar}- \sum_{\rbar \succeq (P, I)} x_{\rbar}.
    \end{align}
    The first of these terms is equal to
    \begin{equation}
        \sum_{\rbar \succeq (P, I)} x_{\rbar} 
        \sum_{\substack{\qbar: \qbar \succeq (P, I) \\ \rbar \succeq \qbar}}(-1)^{s_{\qbar}- s_{P}} 
        = x_{P, I},
    \end{equation}
    by \Cref{lem:helper}.
    Hence, the third term in \Cref{eq:real-big-sum} is equal to
    \begin{equation}
         x_{P, I} - \sum_{\rbar \succeq (P, I)} x_{\rbar}
        = -\sum_{\substack{R \neq I, P,\\R|_{S_P} = P}} x_{R, PR}.
    \end{equation}
    Similarly, the fourth term in \Cref{eq:real-big-sum} is equal to
    \begin{equation}
          \sum_{\substack{R \neq I, P,\\R|_{S_P} = I}} x_{R, PR}.
    \end{equation}
    Plugging everything back into \Cref{eq:real-big-sum}, we get that
\begin{align}
        (V_k A_k x)_{P, I}
        &= 2\sum_{\qbar\succeq(P,I)}x_{\qbar}-x_{P, I}
        -\sum_{\substack{Q \neq I, P,\\Q|_{S_P} = I \text{ or }P}}  x_{Q, PQ} -\sum_{\substack{R \neq I, P,\\R|_{S_P} = P}} x_{R, PR} +\sum_{\substack{R \neq I, P,\\R|_{S_P} = I}} x_{R, PR}\\ 
        &= 2\sum_{\qbar\succeq(P,I)}x_{\qbar}-x_{P, I}
        -2\sum_{\substack{Q \neq I, P,\\Q|_{S_P} = P}}  x_{Q, PQ}\\
        &=  2 x_{P, I} - x_{P, I}\\
        &= x_{P, I},
    \end{align}
    where in the third step we used that if $\qbar = (Q, PQ)$ and $Q|_{S_P} = P$, then $\qbar \succeq (P, I)$.
    This completes the proof.
\end{proof}

Combining these two lemmas, we have the following corollary.
\begin{corollary}[Inverse of $A_k$]
    $A_k$ is invertible, and its inverse is $A_k^{-1} = V_k$.
\end{corollary}

Not only do we want $A$ to be invertible,
we also want both it and its inverse to be well-behaved.
We show that they are well-behaved in a precise technical sense in the following lemma.

\begin{lemma}[$A_k$ is well-behaved]
    \label{lem:A-norm-bounds}
    Let $A = A_k$ for $0 \leq k \leq n$.
    Then
    \begin{align}
        &\|A\|_{B_1 \to \infty}
        \leq \|A\|_{B_1 \to B_1}\leq 4^\locality,\label{eq:condition-one}\\
        &\|A^{-1}\|_{B_1 \to \infty}
        \leq \|A^{-1}\|_{B_1 \to B_1} \leq 4^\locality.\label{eq:condition-two}
    \end{align}
\end{lemma}

As the rows of $A$ and $A^{-1}$ in which $P_1, P_2 \neq I$ behave much differently than the rows in which $P_2 = I$,
we will handle these two cases separately.
First, we consider the case of $P_1, P_2 \neq I$.

\begin{lemma}
    Let $A = A_k$.
    Let $\Pi$ be the projector onto the rows $(P_1, P_2)$ with $P_1, P_2 \neq I$.
    Then
    \begin{align}
        &\|\Pi A\|_{B_1 \to \infty}
        \leq \|\Pi A\|_{B_1 \to B_1}\leq 2^\locality,\\
        &\|\Pi A^{-1}\|_{B_1 \to \infty}
        \leq \|\Pi A^{-1}\|_{B_1 \to B_1} \leq 2^\locality.
    \end{align}
\end{lemma}

\begin{proof}
    For $P_1, P_2 \neq I$,
    we consider bounding the entry
    \begin{align}
        |(Ax)_{P_1, P_2}| &\leq \sum_{(Q_1, Q_2) \succeq (P_1, P_2)} |x_{Q_1, Q_2}|.
    \end{align}
    Note that $|(A^{-1}x)_{P_1, P_2}|$ is also bounded by the same quantity, so the entirety of the following argument will work for it as well.
    We will now show an equivalent way to write this expression which allows us to derive our desired bound.
    For each set $S \subseteq [\locality]$,
    we will construct a vector $x^S$ as follows:
    \begin{enumerate}
        \item Initialize $x^S_{P_1, P_2} = 0$ for all pairs of Paulis with $s_{\pbar} \leq k$.
        \item For each $\qbar = (Q_1, Q_2)$ with $s_{\qbar} \leq k$, consider the $\ell \leq k$ qubits in which the Paulis are identical, and number them from $1$ to $\ell$.
        \item For each $i \in S$, set the $i$-th identical pair in $Q_1$ and $Q_2$ to the identity $I$. Call the resulting Paulis $R_1$ and $R_2$.
        If there is some $i \in S$ for which there is no corresponding identical pair (meaning that $i > \ell$), do not update $x^S$.
        \item Otherwise, update $x^S_{R_1, R_2}\gets x^S_{R_1, R_2} + |x_{Q_1, Q_2}|$.
    \end{enumerate}
    Note that for each $(Q_1, Q_2) \succeq (P_1, P_2)$,
    there is exactly one choice of $S$ so that $(R_1, R_2) = (P_1, P_2)$.
    Thus,
    \begin{align}
        |(Ax)_{P_1, P_2}| \leq \sum_{(Q_1, Q_2) \succeq (P_1, P_2)} |x_{Q_1, Q_2}|
        = \sum_{S \subseteq [\locality]} (x^S)_{P_1, P_2}.
    \end{align}
    Moreover, $x^S$ has the same locality properties as $x$.
    In particular, $\lonorm{x^S} \leq \lonorm{x}$.
    This gives the desired bounds.
\end{proof}

\begin{lemma}
    Let $A = A_k$.
    Let $\Pi$ be the projector onto the rows $(P_1, P_2)$ with $P_1, P_2 \neq I$.
    Then
    \begin{align}
        &\|\overline{\Pi} A\|_{B_1 \to \infty}
        \leq \|\overline{\Pi} A\|_{B_1 \to B_1}\leq 2,\\
        &\|\overline{\Pi} A^{-1}\|_{B_1 \to \infty}
        \leq \|\overline{\Pi} A^{-1}\|_{B_1 \to B_1} \leq 2.
    \end{align}
\end{lemma}
\begin{proof}
    Let $x$ be a vector such that $\lonorm{x} \leq 1$.
    Let $i \in [n]$.
    Then, we can bound the $B_1$ norm of $(\overline{\Pi} A^{-1})x$ associated with site $i \in [n]$ using \Cref{eq:v-def} as
    \begin{equation}
        \sum_{P:S_P \ni i} |(A^{-1} x)_{P, I}|
        \leq \sum_{P:S_P \ni i} \Big|2 x_{P, I} + \sum_{\substack{Q, S_Q \subseteq S_P \\ Q \neq I, P}} \overline{c(PQ)}\cdot x_{Q, \pauli{}(PQ)} +   \sum_{\substack{\qbar \succeq (P, I)\\\qbar \neq (P, I)}} (-1)^{s_{\qbar}- s_{P}} x_{\qbar}-   \sum_{\substack{\qbar \succeq (I, P)\\\qbar \neq (I, P)}} (-1)^{s_{\qbar}- s_{P}} x_{\qbar}\Big|
    \end{equation}
    To understand this expression, note that the second term ranges over all $(Q, PQ)$ with $S_Q \subseteq S_P$ and $Q \neq I, P$;
    the third term ranges over all $(Q, PQ)$ where $Q$ agrees with $P$ on $S_P$ (and arbitrary outside) and $Q \neq P$;
    and the fourth term ranges over all $(Q, PQ)$ where $Q$ agrees with $I$ on $S_P$ (and arbitrary outside) and $Q \neq I$.
    Together, every term of the form $(Q, PQ)$ appears at most once, and $Q = I$ and $P$ both appear zero times. Hence,
    \begin{align}
        \sum_{P:S_P \ni i} |(A^{-1} x)_{P, I}| 
        &\leq \sum_{P:S_P \ni i} \Big(2|x_{P, I}| + \sum_{Q \neq I, P} |x_{Q, \pauli{}(PQ)}|\Big)\label{eq:gonna-cite}\\
        &\leq \sum_{P:S_P \ni i} 2 |x_{P, I}| + \sum_{\substack{\pbar:i \in S_{\pbar},\\  P_1, P_2 \neq I}} |x_{\pbar}|\\
        &\leq 2\lonorm{x}\\
        &\leq 2.
    \end{align}
    As for $(\overline{\Pi} A) x$, the $B_1$ norm associated with site $i \in[n]$ is
    \begin{align}
        \sum_{P:S_P \ni i} |(A x)_{P, I}|
        & \leq \sum_{P:S_P \ni i} \Big(\sum_{\qbar \succeq (P, I)} |x_{\qbar}| + \frac{1}{2} \sum_{Q \neq I} |x_{Q, \pauli{}(PQ)}|\Big)\\
        &\leq \sum_{\pbar: i \in S_{\pbar}} |x_{P_1, P_2}|
        + \frac{1}{2} \sum_{\substack{\pbar:i \in S_{\pbar},\\  P_1, P_2 \neq I}}|x_{\pbar}|\\
        &\leq 2 \lonorm{x}\\
        &\leq 2.
    \end{align}
    Both of these held for all $i \in [n]$, so this completes the proof.
\end{proof}

Combining the previous two lemmas with the triangle inequality and the fact that $2^{\locality} + 2 \leq 4^{\locality}$ for $k \geq 1$ yields \Cref{lem:A-norm-bounds} as a consequence.

\section{Series expansions}
\label{sec:series}

We will continue our study of the local Fourier coefficients of the time evolution operator, defined in \Cref{sec:local-lindblad-fourier} as
\begin{equation}
    (E(x))_{P_1, P_2} = \E_{R \sim \pauli{S_{\pbar}}}[\ntr(\chi_{P_1, P_2}(R)^{\dagger} \cdot e^{\calL_x t}(R))].
\end{equation}
We showed in \Cref{eq:first-order} that these coefficients, when Taylor expanded as a function of $t$, can be expressed as
\begin{equation}\label{eq:first-order-repeat}
    E_{\pbar}(x) = t \cdot (A_k x)_{\pbar} + \sum_{\ell=2}^\infty \frac{t^\ell}{\ell!}\E_{R \sim \pauli{S_{\pbar}}}[\ntr(\chi_{P_1, P_2}(R)^{\dagger} \cdot \calL_x^{\ell}(R))].
\end{equation}
In this section, we will show that this Taylor series converges and concentrates around its first-order term $t \cdot (A_k x)_{\pbar}$ for sufficiently small $t$.
Specifically, we will show that for Lindbladians with $\lonorm{\mathcal{L}} \leq \energy$, $t$ only needs to be smaller than roughly $1/\energy$ for this series to converge.
Our main goal is to show the following lemma.

\ignore{In particular, the objects which are most relevant to analyze for our learning algorithm are as follows.

\begin{definition}[Vector of expectation values]
    For a (parameterized) Lindbladian $\lind_x$ and a time $t \in \R$, we consider the following vector $E: \C^m \to \C^m$ as follows:
    \begin{align}
        (E(x))_{\pbar} \triangleq \E_{R \sim S_{\pbar}}[\ntr(e^{\calL_x^\dagger t}(P_2 R P_1) R)],
    \end{align}
    where $\pbar = (P_1, P_2)$.
\end{definition}

By expanding $e^{\lind^\dagger t}$ as above, we also obtain a series expansion for these expectation values.
The goal of this section is to prove the following lemma, which bounds the contribution of the higher order terms of the series.
}

\begin{lemma}[Operator norm bound on higher-order terms]
    \label{lem:convexity-lindblad}
    Suppose that $\lonorm{x} \leq g$.
    Suppose $t > 0$ satisfies $t < 1/(4e\locality \energy)$.
    Let $A$ be the matrix defined in \Cref{not:A}.
    Then, the Jacobian $J(x)$ of $E(x)/t$ satisfies
    \begin{equation}
        \norm{J(x) - A}_{B_1 \to B_1} \leq (23k)! gt.
    \end{equation}
    Note that $E(x) = t A x + B(x)$, where the Jacobian of $B(x)/t$ is $J(x) - A$.
\end{lemma}

We would like a bound on the Jacobian, because the analysis of our algorithm will ultimately compare the expectations of an estimate $E(x)$ to the true expectations, $E(\lambda)$.
This Jacobian bound tells us that, if $x$ is close to $\lambda$, then the difference in the corresponding expectations can be explained by a linear term, along with a (smaller) higher-order term.

To prove this statement, we will bound \Cref{eq:first-order-repeat} in a term-by-term manner. 
In particular, let us define the expression
\begin{align}
    E^{(\ell)}_{\pbar}(x)
    &\triangleq\E_{R \sim \pauli{S_{\pbar}}}[\ntr(\chi_{P_1, P_2}(R)^{\dagger} \cdot \calL_x^{\ell}(R))]\\
    &=\E_{R \sim \pauli{S_{\pbar}}}[\ntr(P_2 R P_1 \cdot \calL_x^{\ell}(R))]\\
    &= \E_{R \sim \pauli{S_{\pbar}}} [\ntr(\lind_x^{\dagger\ell}(P_2 R P_1)\cdot R)].\label{eq:term-by-term}
\end{align}
We will prove a bound on the Jacobian of this expression, which will then extend to a bound on the Jacobian of $E(x)$ itself.
To do so, we will carefully control the complexity of the superchannel $\lind_x^{\dagger\ell}(\cdot)$ as a function of the growing parameter $\ell$.
Intuitively, if $\lind_x$ (and therefore $\lind_x^{\dagger}$) is local, then $\lind_x^{\dagger\ell}(\cdot)$ should remain reasonably local provided that $\ell$ is reasonably small; this will require us to show an expression for $\lind_x^{\dagger\ell}(\cdot)$ known as a \emph{cluster expansion}, which is an expansion of $\lind_x^{\dagger\ell}(\cdot)$ in terms of local components known as \emph{clusters}.
For this, it will be crucial that we study the adjoint $\lind_x^{\dagger}$ rather than the Lindbladian $\lind_x$ itself.

For brevity, in this section, we often index terms by $a$ instead of $\pbar$, as described in \Cref{def:param-lind}.

\subsection{Cluster expansions}

\begin{definition}[Multisets and clusters]
    We refer to unordered multisets with the notation $\cluster{a} = \{a_1,\dots,a_\ell\}$, where the elements need not be distinct.
    We denote its cardinality as $\abs{\cluster{a}} = \ell$, and we denote its support as $S_{\cluster{a}}$.
    We also denote $x^{\cluster{a}} \triangleq \prod_{a \in \cluster{a}} x_a$.

    We call $\cluster{a}$ a \emph{cluster} if it is connected in the dual interaction graph (there is an edge between $a$ and $b$ if $S_{\pbar_a} \cap S_{\pbar_b} \neq \emptyset$).
    We call $\cluster{a}$ a cluster \emph{from $S$} if $\cluster{a} \cup S$ is a cluster in the modified dual interaction graph where there is an additional term for $S$.
\end{definition}

\newcommand{\localg}{\mathcal{G}}

\begin{lemma}[Cluster expansion of Lindbladians] 
\label{lem:cluster}
    Let $O$ be an operator whose support is contained in (though not necessarily equal to) a set $S \subseteq [n]$.
    Then $\lind_x^{\dagger \ell}(O)$ is a degree-$\ell$ matrix-valued polynomial which we can write in the following way:
    \begin{align*}
        \lind_x^{\dagger \ell}(O) = 2^\ell \ell! \sum_{\cluster{a} : \abs{\cluster{a}} = \ell} x^{\cluster{a}} \localg_{S, \cluster{a}}(O) \iver{\cluster{a} \cup S \text{ is a cluster}},
    \end{align*}
    where $\localg_{S,\cluster{a}}$ is a superoperator with bounded Fourier weight, $\sum_{Q_1, Q_2} \abs{\widehat{\localg_{S,\cluster{a}}}(Q_1, Q_2)} \leq 1$, and $\widehat{\localg_{S,\cluster{a}}}(P_1, P_2)$ is only nonzero provided $S_{\pbar} \subseteq S_{\cluster{a}}$.
    Note that $\localg_{S, \cluster{a}}$ depends on the subset $S$ but not $O$.
\end{lemma}
\begin{proof}
    We prove the lemma by induction on $\ell$.
    For the base case of $\ell = 0$, we have $\mathcal{L}^{\dagger \ell}(O) = O$.
    Then, consider $\localg_{S,\cluster{a}}(O) = O$, i.e., $\localg_{\cluster{a}}$ is the identity superoperator,
    for $\cluster{a} = \emptyset$.
    This satisfies the required conditions on its Fourier coefficients, because the only nonzero Fourier coefficient is $\widehat{\localg_{S,\cluster{a}}}(I, I) = 1$.
    In addition, $S_{I, I} = \emptyset = S_{\cluster{a}}$.
    Moreover, $\cluster{a} \cup S$ is vacuously a cluster.
    
    % For the base case of $\ell = 0$, the claim is clear because $\mathcal{L}^{\dagger\ell}(P) = P$ so that $\mu_{\varnothing} = 1$ and $\varnothing$ is vacuously a cluster from $S_P$.
    For the inductive step, suppose the result holds for $\ell$.
    Then,
    \begin{align}
        \mathcal{L}_x^{\dagger\,\ell+1}(O) &= \mathcal{L}_x^{\dagger}(\mathcal{L}_x^{\dagger\ell}(O))\\
        &= \sum_{\substack{P_1,P_2\\P_1 \neq I}} x_{P_1,P_2} \left(P_2 \mathcal{L}^{\dagger\ell}(O) P_1 - \frac{1}{2}\{P_2P_1, \mathcal{L}^{\dagger\ell}(O)\}\right)\\
        &= 2^\ell \ell! \sum_{\substack{P_1,P_2\\P_1 \neq I}} \sum_{\cluster{a} : \abs{\cluster{a}} = \ell} x_{P_1,P_2} \cdot x^{\cluster{a}} \left(P_2 \localg_{S,\cluster{a}}(O) P_1 - \frac{1}{2}\{P_2P_1, \localg_{S,\cluster{a}}(O)\}\right) \iver{\cluster{a} \cup S \text{ is a cluster}}.\label{eq:is-a-cluster}
    \end{align}
    In the second line, we use the definition of $\mathcal{L}^\dagger$.
    In the last line, we use the inductive hypothesis.

    Now, suppose $\cluster{a}$ satisfies that $\cluster{a} \cup S$ is a cluster,
    and let us consider the corresponding term 
    \begin{align}
    &2^\ell \ell!\cdot x_{P_1,P_2} \cdot x^{\cluster{a}} \left(P_2 \localg_{S,\cluster{a}}(O) P_1 - \frac{1}{2}\{P_2P_1, \localg_{S,\cluster{a}}(O)\}\right) \iver{\cluster{a} \cup S \text{ is a cluster}}\\
    ={}&2^\ell \ell!\cdot x_{P_1,P_2} \cdot x^{\cluster{a}} \left([\chi_{P_2, P_1} - \frac12(\chi_{P_2 P_1, I} + \chi_{I, P_2 P_1})](\localg_{S,\cluster{a}}(O))\right) \iver{\cluster{a} \cup S \text{ is a cluster}}.\label{eq:second-line}
    \end{align}
    This term corresponds to the new cluster $\cluster{b} = \cluster{a} \cup \{S_{\pbar}\}$ inside $\mathcal{L}_x^{\dagger\,\ell+1}(O)$.
    Indeed, note that $x_{P_1, P_2} \cdot x^{\cluster{a}} = x^{\cluster{b}}$.
    The induction hypothesis tells us that $\localg_{S,\cluster{a}}$ is a linear combination of terms of the form $Q_1 O Q_2$ with $\qbar \subseteq S_{\cluster{a}}$.
    Hence, \Cref{eq:second-line} is a linear combination of terms of the form
    \begin{equation}\label{eq:one-term}
        2^\ell \ell! \cdot x_{P_1,P_2} \cdot x^{\cluster{a}} \left([\chi_{P_2, P_1} - \frac12(\chi_{P_2 P_1, I} + \chi_{I, P_2 P_1})](Q_1 O Q_2)\right) \iver{\cluster{a} \cup S \text{ is a cluster}}.
    \end{equation}
    Note that this is in turn a linear combination of terms of the form
    $Q_1' O Q_2'$ with $S_{\qbar'} \subseteq S_{\qbar} \cup S_{\pbar} \subseteq S_{\cluster{a}} \cup S_{\pbar} = S_{\cluster{b}}$.
    Furthermore, note that since the total support of $Q_1 O Q_2$ is contained in $S_{\cluster{a}} \cup S$, \Cref{eq:one-term} is nonzero only if $S_{\pbar}$ overlaps with $S_{\cluster{a}} \cup S$. 
    Since we know that $\cluster{a} \cup S$ is a cluster, this is equivalent to $S_{\cluster{a}} \cup \{S_{\pbar}\} \cup S = S_{\cluster{b}} \cup S$ being a cluster.

    By the triangle inequality, the expression in 
    \Cref{eq:second-line} has Fourier weight at most $2 \cdot 2^{\ell} \ell! = 2^{\ell+1} \ell!$.
    Now, we collect all the terms in the sum associated to the monomial $x^{\cluster{b}}$.
    There are at most $\ell + 1$ of them, corresponding to the clusters formed by removing one element from $\cluster{a}$ along with the element removed.
    Hence, their total Fourier weight is at most $(\ell+1) \cdot 2^{\ell+1} \ell! = 2^{\ell + 1} (\ell+1)!$.
    This gives the desired bound by collecting the corresponding (matrix) coefficient and labeling it $\localg_{\cluster{b}}(O)$.
\end{proof}

This sum is bounded because of the following statement bounding the number of clusters in a bounded-degree (weighted) graph.

\begin{lemma}[Cluster count] \label{lem:tree-count}
    Let $\ell \geq 0$.
    Let $x$ satisfy $\lonorm{x} \leq g$, and let $i \in [n]$.
    Define
    \begin{equation}
        Z_i(x) \triangleq \sum_{\cluster{a} : \abs{\cluster{a}} = \ell} x^{\cluster{a}} \iver{\cluster{a} \text{ is a cluster from } i}.
    \end{equation}
    Then
    \begin{align} \label{eq:cluster-count}
        |Z_i(x)| \leq (egk)^{\ell}.
    \end{align}
\end{lemma}
\begin{proof}
We first show \Cref{eq:cluster-count}.
Let $r$ be an integer satisfying $1 \leq r \leq \locality$.
We begin with the following standard fact (Lemma 4 of~\cite{mann2024algorithmic}):
Let $G = (V, E)$ be a multihypergraph with maximum degree at most $g$ and rank at most $k$; then the number of connected subgraphs (sets of edges) of size $r$ containing a vertex $v \in V$ in its support is at most $(eg(k-1))^r$.
The analogous statement also holds for weighted graphs: let $w_e$ be the nonnegative weight associated to hyperedge $e$, and let $g$ be now the weighted degree, $\max_{i \in V} \sum_{e \ni i} \abs{w_e}$. Then
\begin{align}\label{eq:multisuperhypergraphs}
    \sum_{S \subseteq E} \iver{S \text{ is a connected subgraph of size } r \text{ containing } v} \prod_{e \in S} w_e \leq (eg(k-1))^{r}.
\end{align}
To prove this, let us note that it suffices to prove this when the $w_e$'s are rational, by a continuity argument.
But for rational $w_e$'s, we can multiply each $w_e$ by a scalar such that the weights become integral, and then apply the unweighted statement to the analogous hypergraph.

To apply this to our setting, let $G_x$ be the multihypergraph with the vertex set $V = [n]$ and, for each Lindbladian term $a \in [m]$, a hyperedge $S_a$ with weight $|x_a|$. Then \Cref{eq:multisuperhypergraphs} implies that for each vertex $i \in [n]$,
\begin{align}\label{eq:applied-to-Gx}
    \sum_{S \subseteq [m]} \iver{S \text{ is a connected subgraph in $G_x$ of size } r \text{ containing } i} \prod_{a \in S} |x_a| \leq (eg(k-1))^{r}.
\end{align}
From there, we now consider clusters.
For every subgraph $S = \{a_1, \ldots, a_r\} \subseteq [m]$ of size $r$, there are $\binom{\ell-1}{r-1}$ ways to assign positive integer weights to the $r$ elements of $S$ which sum up to $\ell$. Each of these corresponds to a unique cluster $\cluster{a}$ of cardinality $|\cluster{a}| = \ell$ consisting of $r$ distinct elements; we write $\cluster{a} \sim S$ for a cluster formed in this manner.
Moreover, since $\lonorm{x} \leq g$, the weight $x^{\cluster{a}}$ of the cluster $\cluster{a}$ is at most the weight of the subgraph times $g^{\ell - r}$.
Therefore, we can bound the number of clusters using the number of subgraphs, giving
\begin{align}
    \abs{Z_i(x)}
    &=\Big|\sum_{r=1}^\ell\sum_{S \subseteq [m]} \iver{S \text{ is a connected subgraph in $G_x$ of size } r \text{ containing } i} \cdot \sum_{\cluster{a} \sim S} x^{\cluster{a}}\Big|\\
    &\leq\sum_{r=1}^\ell\sum_{S \subseteq [m]} \iver{S \text{ is a connected subgraph in $G_x$ of size } r \text{ containing } i} \cdot \sum_{\cluster{a} \sim S} |x^{\cluster{a}}|\\
    &\leq\sum_{r=1}^\ell\sum_{S \subseteq [m]} \iver{S \text{ is a connected subgraph in $G_x$ of size } r \text{ containing } i} \cdot \sum_{\cluster{a} \sim S} \prod_{a \in S} |x_a| \cdot g^{\ell-r}\\
    &=\sum_{r=1}^\ell\sum_{S \subseteq [m]} \iver{S \text{ is a connected subgraph in $G_x$ of size } r \text{ containing } i} \cdot \prod_{a \in S} |x_a| \cdot g^{\ell - r} \binom{\ell-1}{r-1}\\
    &=\sum_{r=1}^{\ell} (eg(\locality-1))^r \cdot g^{\ell - r}\binom{\ell-1}{r-1}\label{eq:plugged-in-awesome-bound},
\end{align}
where we used \Cref{eq:applied-to-Gx} in the last step.
Now, using the fact that $\sum_{r=1}^{\ell} x^r \cdot \binom{\ell-1}{r-1} = x (x + 1)^{\ell-1} \leq (x + 1)^{\ell}$ for a nonnegative number $x$,
we have
\begin{equation}
    \eqref{eq:plugged-in-awesome-bound}
    = g^{\ell} \cdot \sum_{r=1}^{\ell} (e(\locality-1))^r \cdot \binom{\ell-1}{r-1}
    \leq 
     g^{\ell} (e (\locality - 1) + 1)^{\ell}
     \leq g^{\ell} (ek)^{\ell}.
\end{equation}
This completes the proof.
\end{proof}

We will also need the following consequence of \Cref{lem:tree-count}.

\begin{lemma}\label{lem:derivative-bound}
    Let $\ell \geq 2$.
    Let $x$ satisfy $\lonorm{x} \leq g$, and let $v$ satisfy $\lonorm{v} \leq 1$.
    Let $i \in [n]$.
    Then
    \begin{align} \label{eq:cluster-deriv-count}
        \abs[\Big]{\sum_a v_a \partial_a Z_i(x)} \leq e^2 k \ell (egk)^{\ell - 1}
    \end{align}
\end{lemma}
\begin{proof}
To begin, let us compute
\begin{equation}
    \partial_a Z_i(x)
    = \sum_{\substack{\cluster{a} : a \in \cluster{a},\\\abs{\cluster{a}} = \ell}} x^{\cluster{a}\setminus\{a\}} \iver{\cluster{a} \text{ is a cluster from } i},
\end{equation}
where here the notation ``$\cluster{a}\setminus\{a\}$'' refers to removing a single instance of $a$ from $\cluster{a}$.
Thus, if $v$ is a vector satisfying $\lonorm{v} \leq 1$, we have
\begin{equation}
    \sum_a v_a \partial_a Z_i(x)
    = \sum_a v_a \cdot \sum_{\substack{\cluster{a} : a \in \cluster{a},\\\abs{\cluster{a}} = \ell}} x^{\cluster{a}\setminus\{a\}} \iver{\cluster{a} \text{ is a cluster from } i}.
\end{equation}
Note that the absolute value of this quantity is largest when $x$ and $v$ are nonnegative, and so we will henceforth make this assumption without loss of generality.
Now, take $u = \frac{\ell-1}{\ell}\frac{x}{g} + \frac{1}{\ell} v$, and notice that $\lonorm{u} \leq 1$ by construction.
Then
\begin{align}
    Z_i(u)
    &= \sum_{\cluster{a} : \abs{\cluster{a}} = \ell} u^{\cluster{a}} \iver{\cluster{a} \text{ is a cluster from } i}\\
    &= \sum_{\cluster{a} : \abs{\cluster{a}} = \ell} \Big(\frac{\ell-1}{\ell}\cdot \frac{x}{g} + \frac{1}{\ell}\cdot  v\Big)^{\cluster{a}} \iver{\cluster{a} \text{ is a cluster from } i}.\label{eq:about-to-binomial}
\end{align}
Note that if $\cluster{a} = \{a_1, \ldots, a_{\ell}\}$, then we can expand
\begin{align}
    \Big(\frac{\ell-1}{\ell}\cdot \frac{x}{g} + \frac{1}{\ell}\cdot  v\Big)^{\cluster{a}}
    &= \Big(\frac{\ell-1}{\ell \cdot g}\Big)^{\ell} x^{\cluster{a}} + \Big(\frac{\ell-1}{\ell \cdot g}\Big)^{\ell-1} \cdot \frac{1}{\ell} \sum_{a \in \cluster{a}}x^{\cluster{a} \setminus\{a\}} v_a
    + \cdots \\
    &\geq\Big(\frac{\ell-1}{\ell \cdot g}\Big)^{\ell-1} \cdot \frac{1}{\ell} \sum_{a \in \cluster{a}} x^{\cluster{a} \setminus\{a\}} v_{a}.
\end{align}
Here, in the first step, we use the binomial formula to expand $(\frac{\ell-1}{\ell}\cdot \frac{x}{g} + \frac{1}{\ell}\cdot  v)^{\cluster{a}}$.
In the second step, we use the fact that all the terms in this expansion are nonnegative (which follows from the fact that $x$ is nonnegative), to lower bound the expression by only those terms which use a single coordinate of $v$.
Plugging this in to \Cref{eq:about-to-binomial}, we have that
\begin{align}
    Z_i(u)
    &\geq\sum_{\cluster{a} : \abs{\cluster{a}} = \ell} \Big(\Big(\frac{\ell-1}{\ell \cdot g}\Big)^{\ell-1} \cdot \frac{1}{\ell} \sum_{a \in \cluster{a}} x^{\cluster{a} \setminus\{a\}} v_{a}\Big)\cdot \iver{\cluster{a} \text{ is a cluster from } i}\\
    &\geq\Big(\frac{\ell-1}{\ell \cdot g}\Big)^{\ell-1} \cdot \frac{1}{\ell} \cdot 
    \sum_a v_a \sum_{\substack{\cluster{a} : a \in \cluster{a},\\\abs{\cluster{a}}=\ell}}  x^{\cluster{a} \setminus\{a\}} \iver{\cluster{a} \text{ is a cluster from } i}\\
    &= \Big(\frac{\ell-1}{\ell \cdot g}\Big)^{\ell-1} \cdot \frac{1}{\ell} \cdot 
    \sum_a v_a \partial_a Z_i(x).
\end{align}
Rearranging, we have
\begin{equation}
    \sum_a v_a \partial_a Z_i(x)
    \leq \ell \Big(\frac{\ell}{\ell-1}\Big)^{\ell-1} g^{\ell - 1} \cdot Z_i(u)
    \leq \ell \Big(\frac{\ell}{\ell-1}\Big)^{\ell-1} g^{\ell - 1} \cdot (e\locality)^{\ell}
    \leq \ell (ek)^\ell e g^{\ell - 1}.
\end{equation}
In the second step, we used \Cref{lem:tree-count} and the fact that $\lonorm{u}\leq 1$. This concludes the proof.
\end{proof}

The following corollary follows directly from combining \Cref{lem:cluster,lem:tree-count}.

\begin{corollary}[Operator norm bound]
    \label{lem:inductive-bound}
    Let $P \in \pauli{\locality}$, and let $\ell \geq 2$.
    Suppose $\mathcal{L}_\lambda$ is a $\locality$-local superoperator with $\lonorm{\mathcal{L}_\lambda} \leq \energy$.
    Then,
    \begin{equation}
        \norm{\mathcal{L}_{\lambda}^{\dagger\ell}(P)}_{\mathrm{op}} \leq \ell! (2e\locality\energy)^\ell.
    \end{equation}
\end{corollary}

\subsection{Bounds on derivatives}

We can derive the following as a corollary of \Cref{lem:cluster}.

\begin{lemma}[Cluster expansion of Fourier expectations]
    \label{lem:cluster-fourier}
    We can write the function
    \begin{align}
        \E_{R \sim \pauli{S_{\pbar}}}[\ntr(\lind_x^{\dagger \ell}(P_2 R P_1) R)]
        = 2^\ell \ell! \sum_{\cluster{a} : \abs{\cluster{a}} = \ell} \gamma_{\cluster{a}} x^{\cluster{a}} \iver{\cluster{a} \text{ is a cluster}}\iver{S_{\cluster{a}} \supseteq S_{\pbar}}.
    \end{align}
    where $\gamma_{\cluster{a}}$ are some coefficients satisfying $\abs{\gamma_{\cluster{a}}} \leq 1$.
\end{lemma}

\begin{proof}
We use \cref{lem:cluster} to expand out $\lind_x^{\dagger \ell}(P_2 R P_1)$ into a polynomial for every $R \in \locals_{S_{\pbar}}$:
\begin{align}
    \E_{R \sim \pauli{S_{\pbar}}}[\ntr(\lind_x^{\dagger \ell}(P_2 R P_1)R)]
    &= 2^\ell \ell! \E_{R \sim \pauli{S_{\pbar}}}\bracks[\Big]{\ntr\parens[\Big]{R \sum_{\cluster{a} : \abs{\cluster{a}} = \ell} x^{\cluster{a}} \localg_{S_{\pbar}, \cluster{a}}(P_2RP_1) \iver{\cluster{a} \cup S_{\pbar} \text{ is a cluster}}}} \\
    &= 2^\ell \ell! \sum_{\cluster{a} : \abs{\cluster{a}} = \ell} x^{\cluster{a}} \iver{\cluster{a} \cup S_{\pbar} \text{ is a cluster}} \underbrace{\E_{R \sim \pauli{S_{\pbar}}}\bracks[\Big]{\ntr\parens[\Big]{R \localg_{S_{\pbar}, \cluster{a}}(P_2 R P_1)}}}_{\triangleq \gamma_{\cluster{a}}},
\end{align}
where in the first equality,
we used the fact that the support of $P_2 R P_1$ is contained in $S_{\pbar}$  to apply \Cref{lem:cluster}.
The coefficients of this expansion can be bounded:
\begin{align}
    \abs{\gamma_{\cluster{a}}}
    \leq \sum_{Q_1, Q_2} \left|\widehat{\localg_{S_{\pbar}, \cluster{a}}}(Q_1, Q_2)\right| \E_{R \sim \pauli{S_{\pbar}}}\bracks[\Big]{\abs[\Big]{\ntr\parens[\Big]{R \chi_{Q_1, Q_2}(P_2 R P_1)}}}
    \leq \sum_{Q_1, Q_2} \left|\widehat{\localg_{S_{\pbar}, \cluster{a}}}(Q_1, Q_2)\right|
    \leq 1,
\end{align}
where in the final inequality we use the bound on the Fourier weight of $\localg_{S_{\pbar}, \cluster{a}}$.
Moreover, because $\localg_{S_{\pbar}, \cluster{a}}$ only acts on sites contained in $S_{\cluster{a}}$, $\gamma_{\cluster{a}}$ is only nonzero provided that $S_{\cluster{a}}$ does not merely overlap $S_{\pbar}$, but \emph{contains} it.
This concludes the proof.
\end{proof}

\begin{lemma}\label{lem:jacobian-b1-b1}
    Let $\ell \geq 2$.
    Let $J^{(\ell)}(x)$ be the Jacobian of the vector-valued function $E^{(\ell)}(x)/t$ defined in \Cref{eq:term-by-term}.
    Then $\|J^{(\ell)}(x)\|_{B_1 \to B_1} \leq \frac{1}{t} 2^\ell k e^{k+2} \ell^k 16^k (\ell+1)! (egk)^{\ell - 1}$.
\end{lemma}
\begin{proof}
    Consider some $v$ such that $\lonorm{v} \leq 1$.
    Further consider some site $i \in [n]$.
    Then, by \Cref{lem:cluster-fourier},
    \begin{align}
        \sum_{a : S_{a} \ni i} \abs{(J^{(\ell)}(x) v)_a} &= \frac{1}{t}\sum_{a : S_a \ni i} \left|\sum_b v_b \partial_b E^{(\ell)}_a(x)\right|\\
        &= 2^{\ell} \ell! \frac{1}{t}\sum_{a : S_{a} \ni i} \abs[\Big]{\sum_{b} v_b\partial_b \sum_{\cluster{a} : |\cluster{a}| = \ell} \gamma_{a, \cluster{a}} x^{\cluster{a}} \iver{\cluster{a} \text{ is a cluster}} \iver{S_{\cluster{a}} \supseteq S_{a}}} \\
        &= 2^{\ell} \ell! \frac{1}{t}\sum_{a : S_{a} \ni i} \abs[\Big]{\sum_{b} v_b\partial_b \sum_{\cluster{a} : |\cluster{a}| = \ell} \gamma_{a, \cluster{a}} x^{\cluster{a}} \iver{\cluster{a} \text{ is a cluster from $i$}} \iver{S_{\cluster{a}} \supseteq S_{a}}} \\
        &\leq 2^\ell \ell! \frac{1}{t}\sum_{a : S_{a} \ni i} \sum_{\cluster{a} : |\cluster{a}| = \ell} \sum_b \abs{v_b \partial_b x^{\cluster{a}}} \iver{\cluster{a} \text{ is a cluster from $i$}} \iver{S_{\cluster{a}} \supseteq S_{a}} \\
        &= 2^\ell \ell! \frac{1}{t}\sum_{\cluster{a} : |\cluster{a}| = \ell} \sum_b \abs{v_b \partial_b x^{\cluster{a}}} \iver{\cluster{a} \text{ is a cluster from $i$}} \sum_{a : S_{a} \ni i} \iver{S_{\cluster{a}} \supseteq S_{a}} \\
        &\leq 2^\ell \frac{1}{t} \binom{\ell(k-1)}{k-1} 16^k\ell! \sum_{\cluster{a} : |\cluster{a}| = \ell} \sum_b \abs{v_b \partial_b x^{\cluster{a}}} \iver{\cluster{a} \text{ is a cluster from $i$}} \\
        &\leq 2^\ell \frac{1}{t} \binom{\ell(k-1)}{k-1}16^k e^2 k \ell! \ell (egk)^{\ell - 1}\\
        &\leq 2^\ell \frac{1}{t}(16 e \ell)^k e^2 k \ell \cdot \ell! (egk)^{\ell-1}
    \end{align}
    In the second line, we use \Cref{lem:cluster-fourier}.
    In the third line, we use the fact that $i \in S_a \subseteq S_{\cluster{a}}$, and so $\cluster{a} \cup \{\{i\}\}$ is a cluster.
    In the fourth line, we use triangle inequality.
    In the fifth line, we move the order of the sums.
    The sixth line uses that a cluster with $\ell$ elements has support size at most $\ell(k-1)+1$, so the number of possible subsets $S_a$ with $k$ elements (but still containing $i$) is at most $\binom{\ell(k-1)}{k-1}$,
    and the number of possible Paulis $\pbar$ supported on $S_a$ is at most $16^k$.
    The seventh line uses \Cref{lem:derivative-bound}.
    The last line uses that $\binom{b}{c} \leq (be/c)^c$.
    Since this holds for all $i$, we have the desired bound.
\end{proof}

\begin{proof}[Proof of {\Cref{lem:convexity-lindblad}}]
    Let $\pbar$ and $\qbar$ be $\locality$-local.
    Recall from \Cref{eq:first-order-repeat} that $E(x)$ can be written as
    \begin{align}
        E_{\pbar}(x)
        &= t (A x)_{\pbar} + \sum_{\ell=2}^\infty \frac{t^\ell}{\ell!} \E_{R\sim\pauli{S_{\pbar}}}[\ntr(\mathcal{L}^\ell_x(R)P_2 R P_1)],
    \end{align}
    where $A = A_k$ is the matrix defined in \Cref{not:A}.
    Thus, $J(x) - A$ is given by
    \begin{equation}
        J_{\pbar,\qbar}(x) - A_{\pbar, \qbar}
        = \frac{1}{t}\sum_{\ell=2}^\infty \frac{t^\ell}{\ell!} \E_{R\sim \pauli{S_{\pbar}}}[\partial_{x_{\qbar}}\ntr(\mathcal{L}_x^\ell(R)P_2 R P_1)]
        = \sum_{\ell=2}^\infty \frac{t^\ell}{\ell!} J^{(\ell)}_{\pbar, \qbar}(x).
    \end{equation}
    Thus, using \Cref{lem:jacobian-b1-b1}, we have
    \begin{align}
        \norm{J(x) - A}_{B_1 \to B_1}
        &\leq  \sum_{\ell=2}^\infty \frac{t^\ell}{\ell!} \Vert J^{(\ell)}(x)\Vert_{B_1 \to B_1}\\
        &\leq \frac{1}{t}\sum_{\ell=2}^\infty \frac{t^\ell}{\ell!} 2^\ell k e^{k+2} \ell^k 16^k (\ell+1)! (egk)^{\ell - 1}\\
        &= \frac{1}{t}\cdot 4e^{k+3}16^k k^2 gt^2 \sum_{\ell=2}^\infty (\ell +1) \ell^k (2egkt)^{\ell-2}\\
        &\leq \frac{1}{t} \cdot 4e^{k+3}16^k k^2 gt^2 \sum_{\ell=1}^\infty (\ell +1) \ell^k 2^{-\ell+2}\\
        &\leq \frac{1}{t} \cdot 32e^{k+3}16^k k^2 gt^2 \sum_{\ell=1}^\infty \frac{\ell^{k+1}}{2^\ell}\\
        &\leq \frac{1}{t} \cdot 32e^{k+3}16^k k^2 gt^2 \cdot 2^{k+1}(k+1)!\\
        &\leq \frac{1}{t} \cdot 2^{7k+12} gt^2 \cdot (k+3)!\\
        &\leq gt \cdot (23k)!.
    \end{align}
    In the second line, we use \Cref{lem:jacobian-b1-b1}.
    In the fourth line, we use $t < 1/(4egk)$.
    In the fifth line, we use that $\ell + 1 \leq 2\ell$ and enlarge the sum to include $\ell = 1$.
\end{proof}

\subsection{Estimating the series}
\label{sec:series-est}

In this section, we show how to approximate the series in \Cref{eq:first-order-repeat} or a truncated version of it.
Namely, we consider
\begin{equation}
    E_{\pbar}(x) = \sum_{\ell=1}^\infty \frac{t^\ell}{\ell!} \E_{R\sim \pauli{S_{\pbar}}}[\ntr(R \lind_x^{\dagger \ell}(P_2 R P_1))] = \sum_{\ell=1}^\infty \frac{t^\ell}{\ell!} E_{\pbar}^{(\ell)}(x).
\end{equation}
We denote this series truncated to degree $\Gamma \geq 0$ as
\begin{equation}
    \widetilde{E}_{\pbar}(x) \triangleq \sum_{\ell=1}^\Gamma \frac{t^\ell}{\ell!} E_{\pbar}^{(\ell)}(x).
\end{equation}
The goal of this section is to construct an unbiased estimator with bounded magnitude of this truncated series.
To this end, we first show how to construct an unbiased estimator of each term $E_{\pbar}^{(\ell)}(x)$ in \Cref{alg:unbiased-est-term}.
Then, we use this to construct an unbiased estimator of the truncated series $\widetilde{E}_{\pbar}(x)$ in \Cref{alg:unbiased-est-series}.
This is similar in spirit to Section 4.1 of~\cite{bakshi2024high}.

First, we consider $E_{\pbar}^{(\ell)}(x)$.
To make notation cleaner, we rewrite our Lindbladian using the notation from earlier in this section:
\begin{equation}
    \label{eq:lind-simple}
    \mathcal{L}_x^\dagger(O) = \sum_{P_1,P_2} x_{P_1,P_2}\left(P_2 O P_1 - \frac{1}{2}\{P_2P_1, O\}\right) = \sum_{a \in [m]} x_a \mathcal{L}_a^\dagger(O).
\end{equation}
Also define the set
\begin{equation}
    \mathcal{Q}_{\pbar}^{(\ell)} \triangleq \{(a_1,\dots, a_\ell) : \{a_1,\dots, a_\ell\} \cup S_{\pbar} \text{ is a cluster.}\}.
\end{equation}
Note here, on the left-hand side, $(a_1,\dots,a_\ell)$ is an ordered tuple while $\{a_1,\dots,a_\ell\}$ is unordered.
Then, by unrolling \Cref{lem:cluster} to keep track of ordered tuples instead of unordered, then we can write
\begin{equation}
    \mathcal{L}_x^{\dagger \ell}(P_2 RP_1) = \sum_{(a_1,\dots,a_\ell) \in \mathcal{Q}_{\pbar}^{(\ell)}} x_{a_1} \cdots x_{a_\ell} \mathcal{L}_{a_\ell}^\dagger \circ \cdots \circ \lind_{a_1}^\dagger(P_2 RP_1).
\end{equation}
Plugging this into the definition of $E_{\pbar}^{(\ell)}(x)$, we have
\begin{equation}
    \label{eq:ordered-cluster}
    E_{\pbar}^{(\ell)}(x) = \sum_{(a_1,\dots, a_\ell) \in \mathcal{Q}_{\pbar}^{(\ell)}} x_{a_1}\cdots x_{a_\ell} \E_{R \sim \mathcal{P}_{S_{\pbar}}}[\ntr(R \lind_{a_\ell}^\dagger \circ \cdots \circ \lind_{a_1}^\dagger(P_2 RP_1))].
\end{equation}
The idea for this unbiased estimator in \Cref{alg:unbiased-est-term} is to sample a term in this series.
From the form of this series, we can do so iteratively: if we wish to sample a term of degree $\ell + 1$, we can look at the term sampled in the previous step for degree $\ell$, plug it into $\lind_x^\dagger$, and sample one of those resulting terms at random.
From the cluster expansion, we know that only terms which form a cluster contribute, so we only need to pick one of those terms.
The weighting of the sampling procedure takes into account the magnitude of the parameters $x$ and is chosen such that the random variable has the desired magnitude and is unbiased.
The full algorithm is given in \Cref{alg:unbiased-est-term}.

\SetKwComment{AlgComment}{$\triangleright$\ }{}
\SetCommentSty{textit}
\begin{algorithm}
    \caption{Unbiased estimator for $E_{\pbar}^{(\ell)}(x)$}
    \label{alg:unbiased-est-term}
    \KwData{Parameters $x$ satisfying $\lonorm{x}\leq g$; a $k$-local pair of Paulis $\pbar = (P_1,P_2)$; an integer $\ell \geq 0$.}
    \KwResult{A scalar $Z_{\pbar}^{(\ell)}(x)$ such that $\E[Z_{\pbar}^{(\ell)}(x)] = E_{\pbar}^{(\ell)}(x)$.}

    Sample $R \sim \pauli{S_{\pbar}}$.\\
    Initialize $c^{(0)} = c(P_2 R P_1)$, $Q^{(0)} = \mathcal{P}(P_2 RP_1)$, $a^{(0)} = (,) \in [m]^0$ (the empty list), and $A^{(0)} = S_{\pbar}$.\\
    \AlgComment{$c^{(r)}$ and $Q^{(r)}$ keep track of the Pauli term sampled, $a^{(r)}$ keeps track of a monomial in \Cref{eq:ordered-cluster}, and $A^{(r)}$ keeps track of the support associated to $a^{(r)}$.}
    \For{$r=0,\dots, \ell - 1$}{
        Compute $w_i \triangleq \sum_{a: i \in S_a} |x_a|$, $W_r \triangleq \sum_{j \in A^{(r)}} w_j$.\\
        If $W_r = 0$ or $c^{(r)} = 0$, return $Z_{\pbar}^{(\ell)}(x) = 0$. Else, continue.\\
        Let $\mathcal{R}_r \triangleq \{(i,b) : i \in A^{(r)} \cap S_b, x_b \neq 0\}$.\\
        Sample from $\mathcal{R}_r$ as follows:\\
            \hspace{\algorithmicindent} Sample $i_{r+1}$ with probability
            \begin{equation}
                \Pr[i_{r+1} = i] = \begin{cases}
                    w_i/W_r & i \in A^{(r)}\\
                    0 & \mathrm{else}
                \end{cases}.
            \end{equation}\\
            \hspace{\algorithmicindent} Sample $b_{r+1}$ with probability
            \begin{equation}
                \Pr[b_{r + 1} = b \mid i_{r+1} = i] = \frac{|x_b|}{w_i}\indicator\{i \in S_b\}.
            \end{equation}\\
            \hspace{\algorithmicindent} Call the sampled values $i,b$ for brevity.\\
        Update $Q^{(r+1)}$ via $\lind_b^\dagger(Q^{(r)}) = c_r' Q^{(r+1)}$.\\
        Update
        \begin{equation}
            c^{(r+1)} = c^{(r)} \frac{W_r}{|A^{(r)} \cap S_b|} \frac{x_b}{|x_b|} c_r'.
        \end{equation}\\
        Update $a^{(r+1)} = a^{(r)} \cup b$ and $A^{(r+1)} = A^{(r)} \cup S_b$.\\
        \AlgComment{Here, $a^{(r)} \cup b$ means appending to the list $a^{(r)}$.}
    }
    \Return{$Z_{\pbar}^{(\ell)}(x) \triangleq c^{(\ell)} \ntr(Q^{(\ell)} R)$.}
\end{algorithm}

\begin{lemma}[Unbiased estimator of $E_{\pbar}^{(\ell)}(x)$]
    \label{lem:unbiased-est-term}
    Let $\lind_x$ be a $k$-local Lindbladian with $\lonorm{x} \leq g$.
    For any integer $\ell \geq 0$ and $k$-local pair of Paulis $\pbar$, then \Cref{alg:unbiased-est-term} outputs a scalar $Z_{\pbar}^{(\ell)}(x)$ such that $\E[Z_{\pbar}^{(\ell)}(x)] = E_{\pbar}^{(\ell)}(x)$ and $|Z_{\pbar}^{(\ell)}(x)| \leq \ell! (2kg)^\ell$.
\end{lemma}

\begin{proof}
    We prove that $\E[c^{(r)} Q^{(r)} \mid R] = \lind_x^{\dagger r}(P_2 R P_1)$ and $|c^{(r)}| \leq r! (2kg)^r$ for all $0 \leq r \leq \ell$ by induction on $r$.
    This clearly implies the claim of the lemma because
    \begin{equation}
        \E[Z_{\pbar}^{(\ell)}(x)] = \E_{R \sim \pauli{S_{\pbar}}}[\E[c^{(\ell)}\ntr(Q^{(\ell)}R) \mid R]] = \E_{R \sim \pauli{S_{\pbar}}}[\ntr(\E[c^{(\ell)} Q^{(\ell)} \mid R] R)] = \E_{R \sim \pauli{S_{\pbar}}}[\ntr(\lind_x^{\dagger \ell}(P_2 R P_1) R)],
    \end{equation}
    where in the second equality, we use linearity.
    Moreover, $|Z_{\pbar}^{(\ell)}(x)| \leq |c^{(\ell)}|$.
    Thus, it suffices to prove $\E[c^{(r)} Q^{(r)} \mid R] = \lind_x^{\dagger r}(P_2 R P_1)$ and $|c^{(r)}| \leq r! (2kg)^r$.

    For the base case of $r = 0$, $\E[c^{(0)} Q^{(0)} \mid R] = P_2 R P_1$, which is clearly equal to $\lind_x^{\dagger 0}(P_2 R P_1)$.
    For the inductive step, suppose that the claim holds after $r$ iterations of the algorithm.
    Considering only the randomness in iteration $r + 1$, we have the joint probability of sampling $(i,b)$ is
    \begin{align}
        \Pr_{\mathrm{iteration}\; r + 1}[i_{r+1} = i, b_{r+1} = b] &= \Pr_{\mathrm{iteration}\; r + 1}[i_{r + 1} = i] \Pr_{\mathrm{iteration}\; r + 1}[b_{r+1} = b \mid i_{r+1} = i]\\
        &= \frac{w_i}{W_r} \frac{|x_b|}{w_i} \indicator\{i \in S_b\} \indicator\{i \in A^{(r)}\}\\
        &= \indicator\{i \in A^{(r)} \cap S_b\} \frac{|x_b|}{W_r}.
    \end{align}
    Using this, we can find the marginal probability for sampling $b$
    \begin{equation}
        \Pr_{\mathrm{iteration}\; r + 1}[b_{r+1} = b] = \sum_i \Pr[i_{r+1} =i, b_{r+1} = b] = \sum_{i \in A^{(r)} \cap S_b} \frac{|x_b|}{W_r} = \frac{|A^{(r)} \cap S_b| |x_b|}{W_r}.
    \end{equation}
    Now, we can compute the expectation
    \begin{align}
        \E_{\mathrm{iteration}\; r + 1}[c^{(r+1)} Q^{(r+1)}] &= c^{(r)} \sum_{\substack{b: S_b \cap A^{(r)} \neq \emptyset\\x_b \neq 0}} \frac{|A^{(r)} \cap S_b| |x_b|}{W_r} \frac{W_r}{|A^{(r)} \cap S_b|} \frac{x_b}{|x_b|} c_r' Q^{(r+1)}\\
        &= c^{(r)}\sum_{b : S_b \cap A^{(r)}} x_b \lind_b^\dagger(Q^{(r)})\\
        &= c^{(r)} \lind_x^\dagger(Q^{(r)}),
    \end{align}
    where in the second line, we use the definition of $c_r' Q^{(r+1)}$.
    In the last line, we use \Cref{eq:lind-simple}.
    Note that $\supp(Q^{(r)}) \subseteq A^{(r)}$, so terms excluded from this sum are zero anyways.
    By the inductive hypothesis, then this gives us
    \begin{equation}
        \E[c^{(r+1)} Q^{(r+1)} \mid R] = \lind_x^{\dagger (r+1)}(P_2 R P_1),
    \end{equation}
    as required.
    For the bound on $c^{(r+1)}$, we have
    \begin{equation}
        |c^{(r+1)}| = \left|c^{(r)}\frac{W_r}{|A^{(r)} \cap S_b|} \frac{x_b}{|x_b|} c_r'\right| \leq 2|W_r| |c^{(r)}| \leq 2|A^{(r)}| g |c^{(r)}| \leq 2k(r+1)g|c^{(r)}|,
    \end{equation}
    where in the first inequality, we use that $|c_r'| \leq 2$.
    In the second inequality, we use that $|W_r| \leq |A^{(r)}| g$ by definition.
    In the last inequality, we use that $A^{(r)} = S_{\pbar} \cup S_{a_1} \cup \cdots \cup S_{a_r}$, where each term is $k$-local, so the total size of their supports must be at most $k(r+1)$.
    By the inductive hypothesis, then $|c^{(r+1)}| \leq (r+1)!(2kg)^{r+1}$.
\end{proof}

Now, we can use this to construct an unbiased estimator for the truncated series $\widetilde{E}_{\pbar}(x)$.
The algorithm for this is given in \Cref{alg:unbiased-est-series}.

\begin{algorithm}
    \caption{Unbiased estimator for $\widetilde{E}_{\pbar}(x)$}
    \label{alg:unbiased-est-series}
    \KwData{Parameters $x$ satisfying $\lonorm{x} \leq g$; a $k$-local pair of Paulis $\pbar = (P_1, P_2)$; threshold $\Gamma \in \mathbb{N} \cup \{+\infty\}$; time $t > 0$.}
    \KwResult{A scalar $Y_{\pbar}(x)$ such that $\E[Y_{\pbar}(x)] = \widetilde{E}_{\pbar}(x)$.}
    Sample $\ell \sim \{0,1,2,\dots, \Gamma\}$ with probabilities $\{2^{-\Gamma}, 2^{-1}, 2^{-2}, \dots, 2^{-\Gamma}\}$.\\
    If $\ell = 0$, output $Y_{\pbar}(x) \triangleq 0$.\\
    Else, run \Cref{alg:unbiased-est-term} with input $\ell$ to obtain $Z_{\pbar}^{(\ell)}(x)$.\\
    \Return{$Y_{\pbar}(x) \triangleq \frac{2^\ell t^\ell}{\ell!} Z_{\pbar}^{(\ell)}(x)$.}
\end{algorithm}

\begin{lemma}[Unbiased estimator of $\widetilde{E}_{\pbar}(x)$]
    \label{lem:unbiased-est-series}
    Let $\lind_x$ be a $k$-local Lindbladian with $\lonorm{x} \leq g$.
    For any threshold $\Gamma \geq 0$ and $k$-local pair of Paulis $\pbar$, then \Cref{alg:unbiased-est-series} outputs a scalar $Y_{\pbar}(x)$ such that $\E[Y_{\pbar}(x)] = \widetilde{E}_{\pbar}(x)$ and $|Y_{\pbar}(x)| \leq (4kgt)^\ell$ for the degree $\ell$ sampled in \Cref{alg:unbiased-est-series}.
\end{lemma}

\begin{proof}
    By \Cref{lem:unbiased-est-term}, we have
    \begin{equation}
        \E[Y_{\pbar}(x)] = \sum_{\ell=1}^\Gamma 2^{-\ell} \frac{2^\ell t^\ell}{\ell!} \E[Z_{\pbar}^{(\ell)}(x)] = \sum_{\ell=1}^\Gamma \frac{t^\ell}{\ell!}E_{\pbar}^{(\ell)}(x) = \widetilde{E}_{\pbar}(x).
    \end{equation}
    Also by \Cref{lem:unbiased-est-term}, we have
    \begin{equation}
        |Y_{\pbar}(x)| \leq \frac{2^\ell t^\ell}{\ell!} |Z_{\pbar}^{(\ell)}(x)| \leq \frac{2^\ell t^\ell}{\ell!} \ell!(2kg)^\ell = (4kgt)^\ell.
    \end{equation}
\end{proof}

Notably, if $t < 1/(4kg)$, then $|Y_{\pbar}(x)| \leq 1$.

\section{Algorithm}
\label{sec:algo}

The goal of this section is to prove \cref{thm:main}.
The detailed version of this theorem is given in \cref{cor:main-detailed}.

Throughout the following section, we will assume that (1) $k > 1$, (2) $g > 0$, and (3) $\eps < g$.
If (1) fails, then the Lindbladian is easy to learn, since it decomposes into a product of Lindbladians on each qubit, which can be learned separately and in parallel.
If either (2) or (3) fails, then outputting the zero Lindbladian suffices.

\subsection{Overview of the algorithm}
We begin by giving an overview of our algorithm for learning local Lindbladians.
Let $\alpha, D$ denote the true parameters that we want to learn.
Let $\calL \triangleq \calL_{\alpha,D}$ be a Lindbladian with bounded local one-norm $\lonorm{\calL} \leq \energy$ and approximate degree $d \triangleq \deg_{\epsilon/(100\cdot 16^k)}(\mathcal{L})$.
We assume access to the time evolution operator $e^{\calL t}$ for a time $t$ satisfying
\begin{equation}
    t < t_{\max} \triangleq \frac{1}{200 \cdot 4^k \cdot (23k)! \energy}.
\end{equation}
Given a vector of coefficients $x$, we define
\begin{align}
    E_{P_1, P_2}(x)
    &\triangleq \E_{R \sim \mathcal{P}_{S_{\pbar}}} \left[\ntr(\chi_{P_1, P_2}(R)^{\dagger} e^{\mathcal{L}_xt}(R))\right] = \E_{R \sim \mathcal{P}_{S_{\pbar}}} \left[\ntr(P_2 R P_1 e^{\mathcal{L}_xt}(R))\right]
\end{align}
to be the local Fourier coefficients (as in \Cref{sec:local-fourier,sec:local-lindblad-fourier}) of the corresponding time evolution.
Note that $E_{P_1, P_2}(\lambda)$ are the local Fourier coefficients corresponding to the true Lindbladian's time evolution $e^{\calL t}$.

\paragraph{Estimating the local Fourier coefficients.}
Our algorithm begins by running \Cref{alg:local-fourier}
to produce estimates $\widehat{E}_{P_1, P_2}(\lambda)$ for all $P_1, P_2 \in \pauli{n}$ with $s_{\pbar} \leq k$ such that
    \begin{equation}\label{eq:estimates}
    |\widehat{E}_{P_1, P_2}(\lambda) - E_{P_1, P_2}(\lambda)| \leq \frac{t \eta}{2}.
    \end{equation}
Here, $\eta$ is an error parameter which we set to
\begin{equation}
    \eta \triangleq \frac{\epsilon}{96000 \cdot 256^k d}.
\end{equation}
To accomplish this, we set the ``$\epsilon$'' parameter of \Cref{alg:local-fourier} to $t \eta$ and the ``$\delta$'' parameter to $0.01$, so that \Cref{alg:local-fourier} performs
\begin{equation}
    \Theta\Big(\frac{C^{\locality}}{(t \eta)^2} \log(n)\Big)
    = \Theta\left(\frac{C_k g^2 d^2\log(n)}{\epsilon^2}\right)
\end{equation}
applications of the time evolution $e^{\calL t}$, where $C_k$ is a constant that depends only on $k$.
Since each application costs time $t$, this leads to a total time evolution of
\begin{equation}
    \tet = \Theta\left(\frac{C_k gd^2\log(n)}{\epsilon^2}\right).
\end{equation}

\paragraph{Estimating the Lindbladian coefficients.}
The main challenge the algorithm faces is to convert these estimates of the local Fourier coefficients into estimates of the actual Lindbladian parameters. 
To do so, it maintains a vector $x$ of its estimates for the Lindbladian coefficients and evaluates the quality of its estimates by comparing the local Fourier coefficients of its guessed (parameterized) Lindbladian $e^{\calL_x t}$ with those of the true Lindbladian $e^{\calL t}$.
Formally, it considers the errors
\begin{align}
    \mathcal{F}_{P_1, P_2}(x) &\triangleq \frac{1}{t}E_{P_1, P_2}(x) - \frac{1}{t}E_{P_1, P_2}(\lambda).
\end{align}
We denote the vector of these values as
\begin{equation}\label{eq:gonna-find-some-roots}
    \mathcal{F}(x) \triangleq (\mathcal{F}_{P_1,P_2}(x))_{P_1,P_2}.
\end{equation}
If all of these errors are small, then $x$ should be close to the true Lindbladian parameters, but if one of these errors is large, then the algorithm updates $x$ in the direction needed to reduce the error.
In this way, the algorithm starts with a poor estimate of the true Lindbladian parameters and iteratively improves it until the result is a good estimate.
Our algorithm is inspired by the Newton-Raphson root-finding algorithm, as a perfect solution $x = \lambda$ will cause \Cref{eq:gonna-find-some-roots} to be equal to 0 and is therefore a root of $\calF(x)$.

Our algorithm can also discover the structure of $\lind$.
Different Lindbladian terms can interact in ways which are complicated and hard to understand, e.g., the ``confusion'' Paulis in the sense of \Cref{sec:estimate-coefs}.
Moreover, the presence of large Lindbladian terms can overshadow the contribution of Lindbladian terms which are small but nevertheless still part of the structure.
However, we have no trouble extracting information about such large Lindbladian terms, unobscured by the noise of other Lindbladian terms.
Inspired by this, our algorithm proceeds in rounds:
In the $j$-th round, the algorithm maintains $\mathcal{O}(\epsilon_j)$-accurate estimates for every Lindbladian term with magnitude $\epsilon_j = 2^{-j} g$ or larger.
If an estimate is smaller in magnitude than $\epsilon_j/ (4d)$, the algorithm rounds it down to $0$.
The remaining nonzero coordinates of the current estimate then reflect the structure of $\lind$ discovered by this iteration.
By iteratively decreasing the error threshold, in a given round, we already have good enough estimates of the larger Lindbladian coefficients so that we can effectively filter out their contribution and only detect the smaller terms.

One (minor) technical wrinkle is that the algorithm is not able to access the errors in \Cref{eq:gonna-find-some-roots} exactly.
This is for two reasons.
First, given access to $e^{\lind t}$, we can only approximate the local Fourier coefficients $E_{P_1,P_2}(\lambda)$, not compute them exactly.
Second, although the algorithm has access to its own estimates $x$, it still cannot compute $E_{P_1, P_2}(x)$ exactly, as this involves taking a matrix exponential of the Lindbladian $\calL_x$.
Instead, the algorithm Taylor expands $E_{P_1, P_2}(x)$ and truncates at a sufficiently high degree.
This truncated Taylor series can be approximated efficiently.
As a result, the algorithm works with an approximation $\widehat{\calF}(x)$ to the error rather than the true error $\mathcal{F}(x)$.
We describe how the algorithm obtains such an approximation in more detail in \Cref{sec:sample-time}.
For the purposes of this section, it suffices to know that we can obtain an approximation $\widehat{\mathcal{F}}(x)$ such that the error is bounded as 
\begin{equation}
    \label{eq:f-hat}
    \eta(x) \triangleq \widehat{\mathcal{F}}(x) - \mathcal{F}(x), \qquad \norm{\eta(x)}_\infty \leq \eta \;\;\text{always}.
\end{equation}

\subsection{The algorithm and guarantee}

We now state our algorithm for learning local Lindbladians.
The full algorithm is detailed in \Cref{alg:full}.

\let\oldnl\nl
\newcommand{\nonl}{\renewcommand{\nl}{\let\nl\oldnl}}
\begin{algorithm}
    \caption{Structure learning Lindbladians}
    \label{alg:full}
    \KwData{Accuracy $\epsilon > 0$; time $t$ satisfying $t < \frac{1}{200 \cdot 4^k \cdot (23k)! \energy}$; expectation accuracy $\eta = \frac{\epsilon}{96000 \cdot 256^k d}$.}
    \KwResult{Estimates $\widehat{\lambda} = (\widehat{\alpha}, \widehat{D})$ such that $\lonorm{\lambda - \widehat{\lambda}} \leq \epsilon$.}

    \BlankLine\BlankLine\BlankLine
    \begingroup
        \DontPrintSemicolon
        \nonl{\textbf{Simplified Algorithm:}}\;
    \endgroup
    \setcounter{AlgoLine}{0}%
    Suppose we have estimates $\widehat{\mathcal{F}}(x)$ satisfying \Cref{eq:f-hat} for given inputs $x$.\\
    Initialize $x^{(0)} = 0 \in \mathbb{C}^m$ and $T = \left\lceil \log_2(g/\epsilon) \right\rceil$.\\
    \For{$j=0,\dots, T - 1$}{
        Set $\epsilon_j\triangleq 2^{-j}g$.\\
        Set $\tau_j\triangleq \epsilon_j / (800 \cdot 16^k d)$.\\
        Update
        \begin{equation}
            x^{(j+1)} = \round_{\eps_j / (4d)}\left(x^{(j)} - A^{-1} \round_{\tau_j}\left(\widehat{\mathcal{F}}(x^{(j)})\right)\right),
        \end{equation}
        where $A$ is defined in \Cref{not:A}.
    }
    Set $\widehat{\lambda} \triangleq x^{(T)}$.\\
    \Return{$\widehat{\alpha}_{P,I} \triangleq \widehat{\lambda}_{P,I}$ and $\widehat{D}_{P_1,P_2} \triangleq \widehat{\lambda}_{P_1,P_2}$.}

    \BlankLine\BlankLine\BlankLine
    \begingroup
        \DontPrintSemicolon
        \nonl{\textbf{Estimating $\widehat{\mathcal{F}}(x)$:}}\;
    \endgroup
    \setcounter{AlgoLine}{0}%
    Use \Cref{alg:local-fourier} to obtain estimates $\widehat{E}(\lambda)$ of $E(\lambda)$ such that $\norm{\widehat{E}(\lambda) - E(\lambda)}_\infty \leq \eta t/2$.\\
    Set $M \triangleq \lceil 64\log(12mT/\delta)/(t\eta)^2 \rceil$.\\
    Run \Cref{alg:unbiased-est-series} $M$ times with threshold $\Gamma = \left\lceil \frac{\log(8/(t\eta))}{\log(1/(2ekgt))} - 1 \right\rceil$, resulting in outcomes $Y_{P_1,P_2}^{(i)}(x)$ for $i \in [M]$, and define
    \begin{equation}
        \widehat{E}_{P_1,P_2}(x) \triangleq \frac{1}{M}\sum_{i=1}^M Y_{P_1,P_2}^{(i)}(x).
    \end{equation}\\
    \Return{$\widehat{\mathcal{F}}(x) \triangleq \frac{1}{t}\widehat{E}(x) - \frac{1}{t}\widehat{E}(\lambda).$}
\end{algorithm}

Notably, our algorithm only uses simple experiments of the form: prepare a Pauli eigenstate, apply the unknown evolution $e^{\mathcal{L}t}$, and measure in a Pauli eigenstate.
A schematic diagram of these simple circuits is presented in \Cref{fig:experiments}.
Our algorithm has the following guarantee.
We do not attempt to optimize the performance of our algorithm with respect to the locality $k$.

\newcommand{\cvdots}{\ensuremath{\raisebox{1ex}{\(\vdots\)}}}
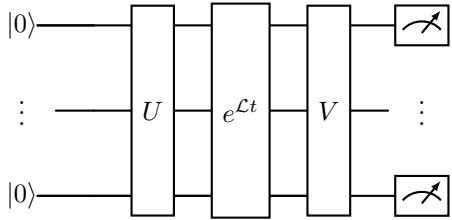
\begin{figure}
    \begin{center}
    \begin{quantikz}
    \ket{0} & \qw & \gate{U_1} & \gate[wires=3]{e^{\mathcal{L}t}} & \gate{V_1} & \meter{}\\
    \phantom{Ci}\cvdots\phantom{Ci} &  & \phantom{Ci}\cvdots\phantom{Ci} &  & \phantom{Ci}\cvdots\phantom{Ci} &\phantom{Ci}\cvdots\phantom{Ci} \\
    \ket{0} & \qw & \gate{U_n} & \qw & \gate{V_n} & \meter{}\\
    \end{quantikz}
    \caption{\textbf{Quantum experiments in our learning algorithm.} All quantum circuits used in our Lindbladian learning algorithm take this form. Here, $\mathcal{L}$ is the unknown Lindbladian, and $U_i, V_i$ are single-qubit Clifford gates.}
    \label{fig:experiments}
    \end{center}
\end{figure}

\begin{theorem} \label{cor:main-detailed}
    Let $\epsilon, \delta > 0$.
    Let $\alpha, D$ be the true parameters, and let $\mathcal{L} = \mathcal{L}_{\alpha, D}$ be a $\locality$-local Lindbladian with bounded local one-norm $\lonorm{\lind} \leq \energy$.
    Let $d = \deg_{\epsilon/(100\cdot 16^{\locality})}(\mathcal{L})$ be the approximate degree of $\mathcal{L}$.
    Let $t$ be such that
    \begin{equation}
        t < \frac{1}{200 \cdot 4^k \cdot (23k)! \energy}.
    \end{equation}
    Then, \Cref{alg:full} outputs estimates $\widehat{\lambda} = (\widehat{\alpha}, \widehat{D})$ of the Lindbladian coefficients $\lambda = (\alpha, D)$ such that $\lonorm{\widehat{\lambda} - \lambda} \leq \epsilon$ with probability at least $1-\delta$ using a total time evolution of $\tet = \mathcal{O}(C_k g d^2\log(n/\delta)/\epsilon^2)$ and classical runtime $\mathcal{O}\left(\frac{n^k g^2 d^2}{\epsilon^2}\log\left(\frac{g}{\epsilon}\right)\log\left(\frac{n}{\delta}\log\left(\frac{g}{\epsilon}\right)\right)\right)$, where $C_k$ is a constant that depends only on $k$.
    We take $k = \mathcal{O}(1)$ in the classical runtime.
    This also implies that $\norm{\widehat{\alpha} - \alpha}_\infty \leq \epsilon$ and $\norm{\widehat{D} - D}_\infty \leq \epsilon$.
\end{theorem}

\subsection{Proof of correctness}
\label{sec:correctness}

First, to prove the correctness of our algorithm, we assume that we are given estimates $\widehat{\mathcal{F}}(x)$ of $\mathcal{F}(x)$ such that
\begin{equation}
    \eta(x) = \widehat{\mathcal{F}}(x) - \mathcal{F}(x), \qquad \norm{\eta}_\infty \leq \eta.
\end{equation}
We describe how one may obtain such estimates in \Cref{sec:sample-time}.
With these estimates, our algorithm simplifies to the form under ``Simplified Algorithm'' in \Cref{alg:full}.
We analyze the algorithm in \Cref{thm:manual-nr}.

\begin{theorem}
    \label{thm:manual-nr}
    Let $\epsilon > 0$.
    Consider Lindbladian parameters $\alpha, D$, and let $\mathcal{L} = \calL_{\alpha, D}$ be a $\locality$-local Lindbladian with bounded local one-norm $\lonorm{\mathcal{L}} \leq \energy$.
    Let $d = \deg_{\eps /(100 \cdot 16^k)}(\lind)$ be the approximate degree of $\lind$.
    Let $t, \eta > 0$ be such that
    \begin{equation}
        t < \frac{1}{200 \cdot 4^k \cdot (23k)! \energy}, \qquad \eta < \frac{\epsilon}{96000\cdot 256^k d}.
    \end{equation}
    Suppose we can compute estimates $\widehat{\mathcal{F}}(x)$ for given inputs $x$ such that
    \begin{equation}
        \norm{\mathcal{F}(x) - \widehat{\mathcal{F}}(x)}_\infty \leq \eta.
    \end{equation}
    Then, the simplified algorithm \Cref{alg:full} finds estimates $\widehat{\lambda} = (\widehat{\alpha}, \widehat{D})$ such that $\lonorm{\widehat{\lambda} - \lambda} \leq \epsilon$.
    This also implies that $\norm{\widehat{\alpha} - \alpha}_\infty \leq \epsilon$ and $\norm{\widehat{D} - D}_\infty \leq \epsilon$.
\end{theorem}

For the sake of analysis, we consider writing the true unknown Lindbladian $\mathcal{L}_{\alpha, D}$ as the parameterized Lindbladian $\mathcal{L}_\lambda$, where $\lambda_{P,I} = \alpha_P$ and $\lambda_{\pbar} = D_{\pbar}$.
Clearly, these representations are equivalent, but $\mathcal{L}_\lambda$ allows us to index into the parameter vector more simply.

Before we prove our main theorem, we need to prove some properties of $\mathcal{F}(x)$.
In particular, one can show that the first order term of $A^{-1}\mathcal{F}(x)$ corresponds precisely to the ideal update we want to perform in the algorithm.
To see this, consider expanding each of the entries of $\mathcal{F}(x)$ in a Taylor series.
By \Cref{lem:convexity-lindblad}, we can write the vector of expectation values as $E(x) = tAx + B(x)$ for some higher order terms $B(x)$.
Thus, we have
\begin{equation}
    \mathcal{F}(x) = \frac{1}{t}E(x) - \frac{1}{t}E(\lambda) = A(x - \lambda) + \frac{1}{t}(B(x) - B(\lambda)).
\end{equation}
Here, we see that, if we could ignore the higher order terms denoted by $B$, we would be done.
Namely, the update $x \leftarrow x - A^{-1}\mathcal{F}(x)$ would directly reveal the unknown parameters $\lambda$.
Of course, we cannot simply throw away the higher order terms, so one key technical step is to bound the contribution of the higher order terms in the expansion of $\mathcal{F}(x)$.
We can do so by using the bounds on the higher order terms of the Jacobian $J(x) \triangleq \partial_{x_{\qbar}} \mathcal{F}_{\pbar}(x)$ of $\mathcal{F}(x)$, which we developed in \Cref{sec:series}.
In particular, we can use \Cref{lem:convexity-lindblad} to bound the higher order terms of $\mathcal{F}$ via the Fundamental Theorem of Calculus.

\begin{corollary}
    \label{coro:f-bound}
    Let $\lonorm{x} \leq \energy$, and 
    let $\Delta \triangleq x - \lambda$.
    Let $t < 1/(8e\locality\energy)$.
    Let $A = A_{\locality}$ be the matrix defined in \Cref{not:A}.
    Then
    \begin{equation}
        \lonorm{\mathcal{F}(x) - (A\Delta)} \leq ct \lonorm{\Delta} \text{ for } c \triangleq (23k)! 2g.
    \end{equation}
\end{corollary}

\begin{proof}
    Consider $f_{\pbar}:[0,1] \to \mathbb{C}$ defined by $f_{\pbar}(s) \triangleq E_{\pbar}(\lambda+s\Delta)/t$.
    Then, by the Fundamental Theorem of Calculus,
    \begin{equation}
        f_{\pbar}(1) - f_{\pbar}(0) = \int_0^1 \partial_s f_{\pbar}(s)\,ds.
    \end{equation}
    Expanding both sides and using $\partial_s = \sum_{\qbar} \Delta_{\qbar} \partial_{\qbar}$, we see that
    \begin{equation}
        \mathcal{F}_{\pbar}(x) = \int_0^1 \sum_{\qbar} \Delta_{\qbar} J_{\pbar,\qbar}(\lambda+s\Delta)\,ds = \int_0^1 (J(\lambda+s\Delta) \Delta)_{\pbar}\,ds.
    \end{equation}
    Subtracting $(A\Delta)$ from both sides, we have
    \begin{equation}
        \label{eq:ftc-f-lindblad}
        \mathcal{F}_{\pbar}(x) - (A\Delta)_{\pbar} = \int_0^1 ( (J(\lambda+s\Delta) - A)\Delta)_{\pbar}\,ds.
    \end{equation}
    Bounding the absolute value of this, we have
    \begin{equation}
        \lonorm{\mathcal{F}(x) - (A\Delta)} \leq \int_0^1 \norm{J(\lambda + s\Delta) -A}_{B_1\to B_1}\lonorm{\Delta}\,ds \leq (23k)! 2g t \lonorm{\Delta},
    \end{equation}
    where in the last inequality, we used \Cref{lem:convexity-lindblad} applied to $\lambda + s\Delta$, which has $\lonorm{\lambda + s\Delta} \leq 2g$.
\end{proof}

Now, we are ready to prove \Cref{thm:manual-nr}.

\begin{proof}[Proof of {\Cref{thm:manual-nr}}]
Let $j \in \{0,\dots, T-1\}$.
We prove this via induction on $j$, where at each iteration, we maintain the invariants
\begin{equation}
\begin{gathered}
    \label{eq:inductive-hypo}
    \lonorm{x^{(j)} - \lambda} \leq \eps_j,\\
    \deg(x^{(j)}) \leq 3 d.
\end{gathered}
\end{equation}
For the base case of $j = 0$, recall that $x^{(0)} = 0$ and $\epsilon_0 = g$.
Thus, we have $\lonorm{x^{(0)} - \lambda} = \lonorm{\lambda} \leq g$.
Moreover, it is vacuously true that $\deg(x^{(0)}) \leq 3 d$.

For the inductive step, suppose that the inductive hypotheses hold at iteration $j$.
We will prove that they still hold at iteration $j +1$.
To simplify notation, we drop the iteration index.
Let $x\triangleq x^{(j)}$ denote the current iterate, $x^+ \triangleq x^{(j+1)}$ the next iterate, $\Delta \triangleq x - \lambda$ the error vector of the current iterate, $\Delta^+ \triangleq x^+ -\lambda$ the error vector of the next iterate, $\eps\triangleq \epsilon_j$ the current error, $\epsilon^+ \triangleq \epsilon_{j+1} = \epsilon/2$ the desired error of the next iterate, and $\tau \triangleq \tau_j$ the current threshold.
(Note that setting $\eps\triangleq \epsilon_j$ creates a notational conflict with the ``$\eps$'' used as input to this algorithm. However, in this proof we will only ever use the form $\eps$ and never the latter input ``$\eps$''.)
We use $y$ to denote the next iterate before rounding:
\begin{equation}
    x^+ = \round_{\epsilon/(4d)}(y)
    \quad \text{ where }
    y \triangleq x - A^{-1}\round_{\tau}\left(\widehat{\mathcal{F}}(x)\right).
\end{equation}
By the inductive hypothesis, $x$ satisfies \Cref{eq:inductive-hypo}.
We will show that $x^+$ satisfies the inductive hypotheses with error parameter $\epsilon^+ = \epsilon / 2$.
It will suffice to analyze the unrounded vector $y$ and show that
\begin{equation}
    \label{eq:y-dist}
    \lonorm{y - \lambda} \leq \frac{\epsilon}{10}.
\end{equation}
To see why, we will first show that this implies all of the inductive hypotheses for $x^+$ for error parameter $\epsilon^+$.

Recall that the definition of approximate degree splits the superoperator $\mathcal{L}_\lambda$ into two parts $\lind_\lambda = \Lbig_\lambda + \Lsmall_\lambda$, where $\lonorm{\Lsmall_\lambda} < \epsilon/(100 \cdot 16^k)$, and minimizes $\deg(\Lbig_\lambda)$.
From here on, let $\Lbig_\lambda, \Lsmall_\lambda$ be the parts attained in this minimization, i.e.,
\begin{equation}
    \Lbig_\lambda = \mathop{\mathrm{argmin}}_{\Lbig : \lonorm{\lind_\lambda - \Lbig} < \epsilon/(100 \cdot 16^k)} \deg(\Lbig),\qquad \Lsmall_\lambda = \lind_\lambda - \Lbig_\lambda.
\end{equation}
As discussed in \Cref{sec:lindblad},
without loss of generality, splitting the Lindbladian in this way simply selects a subset of the coefficients to include in either $\Lbig_\lambda$ or $\Lsmall_\lambda$.
In other words, if $\lambdabig$ are the coefficients of $\Lbig_\lambda$ and $\lambdasmall$ are the coefficients of $\Lsmall_\lambda$, then we may assume that the supports of $\lambdabig$ and $\lambdasmall$ are disjoint. Let $\Wbig$ denote the set of pairs of Paulis which are nonzero in $\lambdabig$, and let $\Wsmall$ denote the set of remaining pairs of Paulis.
\ignore{
i.e., let $\Wbig$ denote a set of pairs of Paulis such that
\begin{equation}
    \lambdabig_{\pbar} = \lambda_{\pbar} \iver{\pbar \in \Wbig},
\end{equation}
where $\lambdabig$ are the coefficients of $\Lbig_\lambda$.
We can similarly define $\Wsmall$ and $\Lsmall_\lambda$.}
Let $\Pi_{\lambdabig}$ denote the coordinate projection onto $\Wbig$, i.e., the indices of coefficients included in $\Lbig_\lambda$, and let $\Pi_{\lambdasmall}$ denote the projection onto $\Wsmall$.
Note that $\Pi_{\lambdasmall} = I - \Pi_{\lambdabig}$.
For the first hypothesis in \Cref{eq:inductive-hypo}, we can bound the contributions after projecting onto $\Pi_{\lambdabig}$ and $\Pi_{\lambdasmall}$ separately:
\begin{align}
    \lonorm{\Pi_{\lambdabig}(x^+ - \lambda)}
    &\leq \lonorm{\Pi_{\lambdabig}(x^+ - y)} + \lonorm{\Pi_{\lambdabig}(y - \lambda)} \\
    &\leq \norm{\Pi_{\lambdabig}}_{\infty \to B_1} \infnorm{x^+ - y} + \lonorm{y - \lambda} \\
    &\leq d \frac{\eps}{4 d} + \frac{\eps}{10} \\
    &= \frac{7\epsilon}{20}.
\end{align}
In the second line, we use that, since $\Pi_{\lambdabig}$ is a coordinate projection, then $\norm{\Pi_{\lambdabig}}_{B_1 \to B_1} \leq 1$.
In the third line, we use that $d = \deg(\Lbig_\lambda)$ so that $\norm{\Pi_{\lambdabig}}_{\infty\to B_1} \leq d$.
We also use that $x^+ = \round_{\epsilon/(4d)}(y)$ and \Cref{eq:y-dist}.
For the $\Pi_{\lambdasmall}$ part, let $U_y \triangleq \{\pbar : |y_{\pbar}| \leq \epsilon/(4d) \}$, and define $\Pi_{U_y}$ to be the coordinate projection onto this set.
In this way, then $x^+ = (I-\Pi_{U_y})y$ so that
\begin{align}
    \Pi_{\lambdasmall}(x^+ - \lambda) = \Pi_{\lambdasmall}(I-\Pi_{U_y})y - \Pi_{\lambdasmall}\lambda
    &= \Pi_{\lambdasmall}(I-\Pi_{U_y})(y-\lambda) - \Pi_{\lambdasmall}\Pi_{U_y}\lambda\\
    &= \Pi_{\lambdasmall}(I-\Pi_{U_y})(y-\lambda) - \Pi_{U_y}\Pi_{\lambdasmall}\lambda,
\end{align}
where in the last step we used the fact that $\Pi_{U_y}$ and $\Pi_{\lambdasmall}$ are coordinate projections and hence commute.
Bounding the $B_1$-norm of this, we have
\begin{align}
    \lonorm{\Pi_{\lambdasmall}(x^+ - \lambda)}
    &\leq \lonorm{\Pi_{\lambdasmall}(I - \Pi_{U_y})(y-\lambda)} + \lonorm{\Pi_{U_y}\Pi_{\lambdasmall}\lambda}\\
    &\leq \norm{\Pi_{\lambdasmall}}_{B_1\to B_1}\norm{I-\Pi_{U_y}}_{B_1\to B_1}\lonorm{y-\lambda} + \norm{\Pi_{U_y}}_{B_1\to B_1}\lonorm{\Pi_{\lambdasmall} \lambda}\\
    &\leq \frac{\eps}{10} + \frac{\eps}{100 \cdot 16^k}\label{eq:small-bound} \\
    &\leq \frac{41\eps}{400}.
\end{align}
Here, we use that $\norm{\Pi_{\lambdasmall}}_{B_1\to B_1}, \norm{\Pi_{U_y}}_{B_1\to B_1}, \norm{I - \Pi_{U_y}}_{B_1\to B_1} \leq 1$ (since they're coordinate projections), \Cref{eq:y-dist}, and $\lonorm{\Lsmall_\lambda} \leq \epsilon/(100 \cdot 16^k)$.
Combining these, we have
\begin{align}
    \lonorm{x^+ - \lambda}
    &\leq \lonorm{\Pi_{\lambdabig}(x^+ - \lambda)} + \lonorm{\Pi_{\lambdasmall}(x^+ - \lambda)} \leq \eps/2.
\end{align}
For the second hypothesis in \Cref{eq:inductive-hypo}, note that
\begin{equation}
    \label{eq:deg-bound}
    \deg(x^+) \leq \deg(\Pi_{\lambdabig} x^+) + \deg(\Pi_{\lambdasmall} x^+) \leq d + \deg(\Pi_{\lambdasmall} x^+),
\end{equation}
where the first inequality uses that $I = \Pi_{\lambdabig} + \Pi_{\lambdasmall}$.
The second inequality uses that $d = \deg(\Lbig)$.
To bound the second term, consider a fixed qubit $i \in [n]$.
Notice that
\begin{equation}
    \lonorm{\Pi_{\lambdasmall} x^+} \leq \lonorm{\Pi_{\lambdasmall} (x^+ - \lambda)} + \lonorm{\Pi_{\lambdasmall} \lambda} \leq \frac{\eps}{10} + \frac{2\eps}{100\cdot 16^k} \leq \frac{\epsilon}{2},
\end{equation}
where the last inequality uses \Cref{eq:small-bound} and $\lonorm{\Lsmall_\lambda} \leq \epsilon/(100\cdot 16^k)$.
Then, because the overall $B_1$-norm of $\Pi_{\lambdasmall} x^+$ is less than $\epsilon/2$, the number of entries of $\Pi_{\lambdasmall} x^+$ of magnitude at least $\eps/(4d)$ (which, due to the rounding in $x^+$, are the \emph{only} nonzero entries in $\Pi_{\lambdasmall} x^+$) whose support contains a fixed qubit $i$ is at most $2d$.
The degree of $\Pi_{\lambdasmall} x^+$ is bounded by the number of elements of $x^+$.
Thus, together with \Cref{eq:deg-bound}, $\deg(x^+) \leq 3d$, so the second hypothesis in \Cref{eq:inductive-hypo} is satisfied.

Now, it remains to prove \Cref{eq:y-dist}.
We define two sets of coordinates.
Define
\begin{align}
    U_{\tau} &\triangleq \left\{(Q_1,Q_2): \left|\widehat{\mathcal{F}}_{Q_1,Q_2}(x)\right| \geq \tau\right\}, \\
    V &\triangleq \left\{(Q_1,Q_2): (Ax)_{Q_1, Q_2} \neq 0 \text{ or } (A\lambdabig)_{Q_1, Q_2} \neq 0\right\},
\end{align}
and let $\Pi_{U_\tau}$ and $\Pi_V$ be the coordinate projections onto $U_\tau$ and $V$, respectively.
In particular, $\Pi_{U_\tau} (\widehat{\func}(x)) = \round_\tau(\widehat{\func}(x))$.
\begin{claim}
    \label{claim:round-set}
    Let $\eta > 0$ be such that $\eta \leq \tau/(120 \cdot 16^k)$.
    Then,
    \begin{align}
        \norm{\Pi_{{U_\tau}}}_{\infty \to B_1} &\leq \frac{4^{k+1}\eps}{\tau} \leq \frac{\eps}{(30 \cdot 4^k \eta)},\\
        \norm{\Pi_V}_{\infty \to B_1} &\leq 4^k (4d).
    \end{align}
\end{claim}
\begin{proof}
Since $\norm{\Pi_{U_\tau}}_{\infty\to B_1} = \max_{i \in [n]} \abs{\{\overline{P} \in U_\tau \mid i \in S_{\pbar}\}}$, we aim to bound the number of elements of $U_{\tau}$ whose supports contain a fixed qubit $i \in [n]$.
Call this set $U_{i, \tau}$, i.e.,
\begin{equation}
    U_{i,\tau} \triangleq \{\pbar \in U_\tau : i \in S_{\pbar}\}.
\end{equation}
By \Cref{eq:f-hat}, we know that
\begin{equation}
    \widehat{\func}(x) = \func(x) + \eta(x),
\end{equation}
where $\|\eta(x)\|_\infty \leq \eta$, so
\begin{equation}
    \label{eq:u-i-tau}
    U_{i, \tau} \subseteq \braces[\Big]{\overline{Q} : \left|\func_{\overline{Q}}(x)\right| \geq \tau - \eta \geq \tau / 2},
\end{equation}
where we use that $\eta \leq \tau/2$.
By \Cref{coro:f-bound}, for $c = (23k)! 2g$,
\begin{equation}
    \lonorm{\func(x)} \leq \lonorm{A \Delta} + c t \lonorm{\Delta} \leq (4^k + c t)\eps \leq 2\cdot 4^{k}\epsilon,
\end{equation}
where in the second to last inequality, we use \Cref{lem:A-norm-bounds} and the inductive hypothesis that $\lonorm{\Delta} \leq \epsilon$.
In the last inequality, we use our choice of $t$.
Thus, because the $B_1$-norm of $\mathcal{F}(x)$ is bounded by $2 \cdot 4^k \epsilon$, the number of elements of $\mathcal{F}(x)$ that are larger than $\tau/2$ in magnitude and which contain qubit $i$ in their support must be at most $4^{k+1} (\epsilon/\tau)$.
In particular, the size of $U_{i, \tau}$ is bounded by $4^{k+1}(\eps / \tau)$.
This can also be seen via \Cref{rmk:deg-bound}.

Now, we prove the bound on $\norm{\Pi_V}_{\infty\to B_1}$.
By the inductive hypothesis, $\deg(x) \leq 3d$, and we know $\deg(\lambdabig) \leq d$ by definition.
Then, by \Cref{lem:A-norm-bounds}, $\deg(Ax) \leq 4^k 3d$ and $\deg(A\lambdabig) \leq 4^k d$.
The norm of $\Pi_V$ is bounded by the sum of these two degree bounds.
\end{proof}

We can write $y-\lambda$ in terms of these projectors:
\begin{align}
    y - \lambda &= x - A^{-1} \Pi_{U_\tau} \left(\widehat{\mathcal{F}}(x)\right) - \lambda\\
    &=\Delta - A^{-1} \Pi_{U_\tau} \left(\widehat{\mathcal{F}}(x)\right)\\
    &=\Delta - A^{-1} \Pi_{U_\tau}\left(\mathcal{F}(x)\right) \underbrace{- A^{-1} \Pi_{U_\tau} \left(\eta(x)\right)}_{\triangleq \err_1} \\
    &= \Delta - A^{-1} \Pi_{U_\tau} (A \Delta) + \underbrace{A^{-1} \Pi_{U_\tau}\left(A \Delta - \mathcal{F}(x)\right)}_{\triangleq \err_2} + \err_1 \\
    &= \Delta - A^{-1} A \Delta + \underbrace{A^{-1}(I - \Pi_{U_\tau}) A \Delta}_{\triangleq \err_3} + \err_2 + \err_1 \\
    &= \err_3 + \err_2 + \err_1.\label{eq:errs}
\end{align}
Thus, there are three forms of error we need to bound.
The first error $\err_1$ comes from only having approximate access to $\widehat{\mathcal{F}}$ (as in \Cref{eq:f-hat}).
The second error $\err_2$ comes from the impact of the higher-order terms of $\mathcal{F}(x)$, which we bounded in \Cref{coro:f-bound}.
The third error $\err_3$ comes from the rounding of each $\widehat{\mathcal{F}}(x)$.
In the following, we bound each error separately.

Now, we can bound each of the error terms in \Cref{eq:errs}.
First, consider $\err_1$, the error from the approximation to $\widehat{\mathcal{F}}(x)$, which we can bound as follows.
\begin{equation}
    \lonorm{\err_1} = \left\|A^{-1} \Pi_{U_\tau} \eta(x)\right\|_{B_1} \leq \norm{A^{-1}}_{B_1 \to B_1}\norm{\Pi_{U_\tau}}_{\infty\to B_1} \norm{\eta(x)}_\infty \leq 4^k \frac{\epsilon}{30 \cdot 4^k \eta} \eta = \frac{\epsilon}{30},
\end{equation}
where we use \Cref{lem:A-norm-bounds}, \Cref{claim:round-set}, and \Cref{eq:f-hat}.
In the last inequality, we also use our choice of $\eta$.

Next, consider $\err_2$, the error from the higher-order terms of $\mathcal{F}(x)$.
\begin{align}
    \lonorm{\err_2} &= \lonorm{A^{-1} \Pi_{U_\tau}(A \Delta - \mathcal{F}(x))} \\
    &\leq \norm{A^{-1}}_{B_1 \to B_1}\norm{\Pi_{U_\tau}}_{B_1 \to B_1} \lonorm{(A \Delta - \mathcal{F}(x))} \\
    &\leq 4^{\locality} c t \lonorm{\Delta}\\
    &\leq 4^{\locality} c t \epsilon\\
    &\leq \frac{\epsilon}{30}.
\end{align}
where $c = (23k)! 2g$.
In the third line, we use \Cref{lem:A-norm-bounds}, $\norm{\Pi_{U_\tau}}_{B_1\to B_1} \leq 1$ since $\Pi_{U_\tau}$ is a coordinate projection, and \Cref{coro:f-bound}.
In the second to last line, we also use the inductive hypothesis that $\lonorm{\Delta} \leq \epsilon$.
In the last line, we use our choice of $t$.

Finally, consider $\err_3$, the error from rounding.
We further break this error up into two parts, corresponding to whether the terms are in $V$ or not.
\begin{align}
    \lonorm{\err_3} &= \lonorm{A^{-1}(I - \Pi_{U_\tau})A \Delta}\\
    &\leq \norm{A^{-1}}_{B_1 \to B_1}\lonorm{(I-\Pi_{U_\tau}) A \Delta}\\
    &\leq 4^k\parens[\Big]{\lonorm{(I-\Pi_{U_\tau})(I - \Pi_V) A \Delta} + \lonorm{(I-\Pi_{U_\tau})\Pi_V  A \Delta}}\\
    &= 4^k\parens[\Big]{\lonorm{\underbrace{(I-\Pi_{U_\tau})(I - \Pi_V) A \Delta}_{\triangleq \err_4}} + \lonorm{\underbrace{\Pi_V (I-\Pi_{U_\tau}) A \Delta}_{\triangleq \err_5}}},
\end{align}
where in the third line, we used \Cref{lem:A-norm-bounds},
and in the fourth line, we used the fact that $\Pi_V$ and $(I-\Pi_{U_\tau})$ are coordinate projections and hence commute.

To bound $\err_4$, note that
\begin{equation}
    (I - \Pi_V) A \Delta = (I - \Pi_V) A (x - \lambdabig) - (I - \Pi_V) A \lambdasmall = -(I - \Pi_V) A \lambdasmall
\end{equation}
by definition of $V$.
Then, we can bound
\begin{align}
    \lonorm{\err_4}
    \leq \lonorm{(I - \Pi_V) A \lambdasmall} \leq 4^k \cdot \frac{\epsilon}{100\cdot 16^k} = \frac{\epsilon}{100\cdot 4^k},
\end{align}
where in the first inequality, we use $\norm{I-\Pi_{U_\tau}}_{B_1 \to B_1} \leq 1$.
In the second inequality, we use $\norm{I-\Pi_V}_{B_1\to B_1} \leq 1$, \Cref{lem:A-norm-bounds}, and $\lonorm{\Lsmall_\lambda} \leq \epsilon/(100\cdot 16^k)$.
To bound $\err_5$, recall that for coordinates not in $U_\tau$,
\begin{align}
    \abs{\func_{Q_1, Q_2}(x)}
    \leq \abs{\widehat{\func}_{Q_1, Q_2}(x)} + \eta \leq \tau + \eta \leq 2\tau,
\end{align}
where we use that $\eta \leq \tau$.
Consequently, we can conclude that
\begin{align}
    \lonorm{\err_5}
    &= \lonorm{\Pi_V (I - \Pi_{U_\tau}) A \Delta} \\
    &\leq \lonorm{\Pi_V (I - \Pi_{U_\tau})  \func(x) } + c t \eps \\
    &\leq \norm{\Pi_V}_{\infty \to B_1} \infnorm{(I - \Pi_{U_\tau})  \func(x)} + c t \eps\\
    &\leq 4^k (4d) 2\tau  + c t \eps\\
    &\leq \frac{\epsilon}{50 \cdot 4^k},
\end{align}
where $c = (23k)! 2g$.
In the second line, we use \Cref{coro:f-bound}.
In the fourth line, we use \Cref{claim:round-set} and \Cref{eq:u-i-tau}.
In the last line, we use our choice of $t$ and $\tau$.

Overall, plugging the bounds on the three errors back into \Cref{eq:errs}, we have
\begin{align}
    \lonorm{y - \lambda}
    &\leq \lonorm{\err_1} + \lonorm{\err_2} + \lonorm{\err_3} \\
    &\leq \lonorm{\err_1} + \lonorm{\err_2} + 4^k(\lonorm{\err_4} + \lonorm{\err_5})\\
    &\leq \frac{\epsilon}{30} + \frac{\epsilon}{30} + \frac{\epsilon}{100} + \frac{\epsilon}{50}\\
    &\leq \frac{\epsilon}{10},
\end{align}
as required.
As discussed after \Cref{eq:y-dist}, this completes the proof.
\end{proof}

\subsection{Sample and time complexity analysis}
\label{sec:sample-time}

We analyze the total time evolution, time resolution, and classical runtime of our algorithm from \Cref{alg:full}.
In the previous section, we proved \Cref{thm:manual-nr}, which states that, as long as we can produce estimates $\widehat{\mathcal{F}}(x)$ of $\mathcal{F}(x)$ to error $\eta$ in $\infty$-norm, then we can learn the Lindbladian parameters well.
Thus, it suffices to analyze the resources required to obtain such an approximation of $\mathcal{F}(x)$.

Recall that, for $\pbar = (P_1,P_2)$, $\mathcal{F}(x)$ is defined as
\begin{equation}
    \mathcal{F}_{\pbar}(x) = \frac{1}{t}E_{\pbar}(x) - \frac{1}{t}E_{\pbar}(\lambda) = \frac{1}{t}\E_{R \sim \pauli{S_{\pbar}}}[\ntr(e^{\mathcal{L}_xt}(R)P_2 R P_1)] - \frac{1}{t}\E_{R \sim \pauli{S_{\pbar}}}[\ntr(e^{\mathcal{L}t}(R)P_2 R P_1)],
\end{equation}
where $\lambda$ are the true parameters of the unknown Lindbladian.
We can estimate $E_{\pbar}(\lambda)$ using our access to $e^{\mathcal{L}t}$ for the unknown Lindbladian $\mathcal{L}$, as in \Cref{sec:estimate-coefs}.
Moreover, we can approximate $E_{\pbar}(x)$ by approximating $e^{\mathcal{L}_xt}$ via a truncated series expansion and constructing an unbiased estimator for this truncated series, similarly to~\cite{haah2024learning,bakshi2024high}.
The full algorithm is given in \Cref{alg:full}.
Using this approach, we have the following guarantee.

\begin{theorem}
    \label{thm:complexity-general}
    Let $\lambda = (\alpha, D)$, and let $\mathcal{L} = \mathcal{L}_{\lambda}$ be a $\locality$-local Lindbladian with bounded local one-norm $\lonorm{\lind} \leq \energy$.
    Let $d = \deg_{\epsilon/(100\cdot 16^{\locality})}(\lind)$.
    Let $t,\eta > 0$, where $t < 1/(8e\locality\energy)$ and $0 < t\eta \leq 1$.
    Then, for iterates $x$ in \Cref{alg:full}, which satisfy $\lonorm{x} \leq 2g$, there exists an algorithm for computing estimates $\widehat{\mathcal{F}}(x)$ of $\mathcal{F}(x)$ for all $T \geq 0$ iterations such that
    \begin{equation}
        \norm{\mathcal{F}(x) - \widehat{\mathcal{F}}(x)}_\infty \leq \eta
    \end{equation}
    with probability at least $1-\delta$ which uses a total evolution time of $\tet = \mathcal{O}(C^{\locality}\log(n/\delta)/(t\eta^2))$ for some absolute constant $C > 0$ and classical runtime $\mathcal{O}\left(\frac{n^kT}{(t\eta)^2}\log\left(\frac{nT}{\delta}\right)\right)$.
    In the classical runtime, we take $k = \mathcal{O}(1)$.
\end{theorem}

With our choices of parameters from \Cref{thm:manual-nr}, this gives us \Cref{cor:main-detailed}.

\begin{proof}[Proof of {\Cref{cor:main-detailed}}]
    This follows by instantiating \Cref{thm:complexity-general} with our choice of $t,\eta, T$.
    Then, taking $k = \mathcal{O}(1)$, the time complexity simplifies to the claimed quantity.
\end{proof}

Now, we prove \Cref{thm:complexity-general}.

\begin{proof}[Proof of {\Cref{thm:complexity-general}}]
    First, we can estimate the expectation values $E_{\pbar}(\lambda)$ using our query access to $e^{\mathcal{L}t}$.
    In particular, by \Cref{lem:estimate-coefs}, we can obtain estimates $\widehat{E}_{P_1,P_2}(\lambda)$ such that
    \begin{equation}
        |\widehat{E}_{P_1,P_2}(\lambda) - E_{P_1,P_2}(\lambda)| \leq \frac{t\eta}{2}
    \end{equation}
    with probability at least $1-\delta/3$ using $\Theta(C^k\log(n/\delta)/(t\eta)^2)$ queries to the evolution operator $e^{\lind t}$, for some absolute constant $C > 0$.
    Because each query evolves for time $t$, this corresponds to the stated total evolution time.
    Note that we only need to compute these expectation values with a quantum computer once, not every iteration.
    It remains to show that we can obtain estimates $\widehat{E}_{P_1,P_2}(x)$ such that
    \begin{equation}
        \label{eq:e-hat}
        |\widehat{E}_{P_1,P_2}(x) - E_{P_1,P_2}(x)| \leq \frac{t\eta}{2}.
    \end{equation}
    To do so, consider Taylor expanding out $e^{\mathcal{L}_xt}$ in the expression for $E_{P_1,P_2}(x)$ up to degree
    \begin{equation}
        \Gamma \triangleq \left\lceil \frac{\log(8/(t\eta))}{\log(1/(2e\locality\energy t))} - 1\right\rceil.
    \end{equation}
    In other words, we approximate
    \begin{equation}
        \label{eq:e-tilde-series}
        \widetilde{E}_{P_1,P_2}(x) \triangleq \sum_{\ell=1}^\Gamma \frac{t^\ell}{\ell!}\E_{R\sim \pauli{S_{\pbar}}}[\ntr(R \mathcal{L}_x^{\dagger \ell}(P_2 R P_1))] = \sum_{\ell=1}^\Gamma \frac{t^\ell}{\ell!} E_{P_1,P_2}^{(\ell)}(x).
    \end{equation}
    We show that
    \begin{equation}
        \label{eq:e-tilde}
        |\widetilde{E}_{P_1,P_2}(x) - E_{P_1,P_2}(x)| \leq \frac{t\eta}{4}
    \end{equation}
    and then show how to efficiently estimate the truncated series $\widetilde{E}_{P_1,P_2}(x)$ to error $t\eta/4$.
    To show \Cref{eq:e-tilde}, by triangle inequality, we have
    \begin{equation}
        |\widetilde{E}_{P_1,P_2}(x) - E_{P_1,P_2}(x)| \leq \sum_{\ell=\Gamma + 1}^\infty \frac{t^\ell}{\ell!}\E_{R\sim\pauli{S_{\pbar}}}[|\ntr(R \mathcal{L}_x^{\dagger \ell}(P_2 R P_1))|].
    \end{equation}
    We can bound each of these terms using \Cref{lem:inductive-bound} and $\lonorm{x} \leq 2g$:
    \begin{equation}
        |\ntr(R \mathcal{L}^{\dagger \ell}(P_2 R P_1))| \leq \frac{1}{N}\norm{R}_{\tr} \norm{\mathcal{L}_x^{\dagger \ell}(P_2 R P_1)}_{\mathrm{op}} \leq \ell! (4e\locality\energy)^\ell.
    \end{equation}
    Plugging back into the above, we get
    \begin{equation}
        |\widetilde{E}_{P_1,P_2}(x) - E_{P_1,P_2}(x)| \leq \sum_{\ell=\Gamma + 1}^\infty (4e\locality\energy t)^\ell = \frac{(4e\locality \energy t)^{\Gamma +1}}{1-4e\locality\energy t} \leq 2(4e\locality \energy t)^{\Gamma + 1} \leq \frac{t\eta}{4},
    \end{equation}
    where in the equality, we use the sum of a geometric series when $4e\locality \energy t \leq 1$.
    In the second to last inequality, we use $4e \locality \energy t < 1/2$.
    In the last inequality, we use our choice of $\Gamma$.

    Now, we show how to construct an estimate $\widehat{E}_{P_1,P_2}(x)$ of $\widetilde{E}_{P_1,P_2}(x)$ up to $\eta t/4$ error with probability at least $1-\delta/3$.
    By \Cref{eq:e-tilde}, then this would show \Cref{eq:e-hat}, as required.
    To do so, we use \Cref{alg:unbiased-est-series}, which has the guarantee given in \Cref{lem:unbiased-est-series}.
    Namely, \Cref{lem:unbiased-est-series} produces a random variable $Y_{\pbar}(x)$ that is an unbiased estimator of $\widetilde{E}_{\pbar}(x)$ and satisfies $|Y_{\pbar}(x)| \leq (8kgt)^\ell$ for $\ell \geq 0$.
    Because $t < 1/(8ekg)$, then $|Y_{\pbar}(x)| \leq 1$.
    Then, we can define
    \begin{equation}
        \widehat{E}_{\pbar}(x) \triangleq \frac{1}{M}\sum_{i=1}^M Y_{\pbar}^{(i)}(x),
    \end{equation}
    where $Y_{\pbar}^{(i)}(x)$ are independent realizations of the random variable $Y_{\pbar}(x)$\footnote{Note that in our application of this, the inputs $x$ are not independent of each other because $x^{(j)}$ depends on the previous iterate $x^{(j)}$. However, the independence claim still holds conditioned on all previous iterations, which is all that is needed.}.
    We can ensure that this is close to $\widetilde{E}_{\pbar}(x)$ with high probability by Hoeffding's inequality:
    \begin{equation}
        \Pr\left[|\widehat{E}_{\pbar}(x) - \widetilde{E}_{\pbar}(x)| > \frac{\eta t}{4}\right] \leq 4 \exp\left(-\frac{M(\eta t/4)^2}{4}\right) \leq \frac{\delta}{3mT},
    \end{equation}
    where in the last inequality we take
    \begin{equation}
        M \triangleq \left\lceil \frac{64}{(t\eta)^2} \log\frac{12mT}{\delta} \right\rceil.
    \end{equation}
    Union bounding over the iterates of the algorithm and all terms $\pbar$, then we have
    \begin{equation}
        |\widehat{E}_{\pbar}(x) - \widetilde{E}_{\pbar}(x)| \leq \frac{\eta t}{4}
    \end{equation}
    with probability at least $1-\delta/3$, as required.

    \paragraph{Time complexity.}
    Now, we analyze the time complexity for this algorithm.
    Throughout the discussion of the time complexity, we consider $k = \mathcal{O}(1)$.
    First, computing $\widehat{E}_{\pbar}(\lambda)$ via \Cref{alg:local-fourier} costs $\mathcal{O}(n^k \log(n/\delta)/(t\eta)^2)$ time.

    We also need to analyze the cost of estimating $\widetilde{E}_{\pbar}$ using \Cref{alg:unbiased-est-series}.
    First, consider the subroutine in \Cref{alg:unbiased-est-term}.
    Notice that for each of the $T$ iterations of the overall algorithm, we can precompute for each qubit $i \in[n]$, its list of incident terms $\{b : i \in S_b, x_b \neq 0\}$ and weights $w_i$ from \Cref{alg:unbiased-est-term} in $\mathcal{O}(m)$ time per iteration.
    This costs an additive $\mathcal{O}(mT)$ in total.
    Then, to sample a single $Z_{\pbar}^{(\ell)}(x)$ in \Cref{alg:unbiased-est-term} with these precomputed quantities, the dominant cost is computing $W_r$, which costs $\mathcal{O}(kr) = \mathcal{O}(r)$.
    We need to repeat this for all iterations $r \in \{0,\dots, \ell-1\}$, so this gives an $\mathcal{O}(\ell^2)$ cost in total.

    Now, consider the cost of sampling a single $Y_{\pbar}(x)$ in \Cref{alg:unbiased-est-series}.
    This is clearly dominated by the cost of sampling $Z_{\pbar}^{(\ell)}(x)$, so this also costs $\mathcal{O}(\ell^2)$.
    In total, we need to sample $mTM$ samples of $Y_{\pbar}(x)$ (for each of the $T$ iterations of the overall algorithm, for each of the terms $\pbar$ in the Lindbladian, and for each realization in the empirical average for $\widehat{E}_{\pbar}(x)$), using fresh randomness for each call.
    To help with efficiency, one can precompute each of the sampled degrees $\ell$ in \Cref{alg:unbiased-est-series}, call these $\ell_j$ for $j \in [mTM]$.
    Then, the total cost from sampling all of these $Y_{\pbar}(x)$ is
    \begin{equation}
        \mathcal{O}\left(\sum_{j=1}^{mTM} \ell_j^2\right).
    \end{equation}
    We want to bound this with high probability.
    Note that, because $0 \leq \ell_j^2 \leq \Gamma^2$ and $0 \leq \E[\ell_j^2] \leq \Gamma^2$, then
    \begin{equation}
        |\ell_j^2 - \E[\ell_j^2]| \leq \Gamma^2
    \end{equation}
    \begin{equation}
        \mathrm{Var}(\ell_j^2) \leq \E[\ell_j^4] = \sum_{\ell=1}^\Gamma \ell^4 2^{-\ell} \leq \sum_{\ell=1}^\infty \ell^4 2^{-\ell} = 150.
    \end{equation}
    Thus, by Bernstein's inequality, for $a \geq 0$
    \begin{equation}
        \Pr\left(\sum_{j=1}^{mTM}(\ell_j^2 - \E[\ell_j^2]) \geq a\right) \leq \exp\left(-\frac{a^2}{2\sum_{j=1}^{mTM} \mathrm{Var}(\ell_j^2) + 2\Gamma^2 a/3}\right) \leq \exp\left(-\frac{a^2}{300mTM + 2\Gamma^2a/3}\right).
    \end{equation}
    Similarly to the variance, note that
    \begin{equation}
        \E[\ell_j^2] = \sum_{\ell=1}^\Gamma \ell^2 2^{-\ell} \leq \sum_{\ell=1}^\infty \ell^2 2^{-\ell} = 6.
    \end{equation}
    Thus, we have
    \begin{equation}
        \Pr\left(\sum_{j=1}^{mTM}\ell_j^2 \geq 6mTM + a\right) \leq \exp\left(-\frac{a^2}{300mTM + 2\Gamma^2a/3}\right).
    \end{equation}
    Choose $a = \sqrt{300mTM \log(3/\delta)} + 2\Gamma^2\log(3/\delta)/3$.
    Then, we have that
    \begin{equation}
        \sum_{j=1}^{mTM} \ell_j^2 \leq 6mTM + \sqrt{300mTM \log(3/\delta)} + \frac{2}{3}\Gamma^2 \log(3/\delta)
    \end{equation}
    with probability at least $1-\delta/3$.
    In other words, the total cost for sampling all of the needed $Y_{\pbar}(x)$ is
    \begin{equation}
        \mathcal{O}\left(\sum_{j=1}^{mTM}\ell_j^2\right) = \mathcal{O}\left(mTM\right) = \mathcal{O}\left(\frac{mT}{(t\eta)^2}\log\left(\frac{mT}{\delta}\right)\right) = \mathcal{O}\left(\frac{n^k T}{(t\eta)^2}\log\left(\frac{nT}{\delta}\right)\right)
    \end{equation}
    with probability at least $1-\delta/3$, where in the last equality, we plug in $m = \mathcal{O}(n^k)$.
    Out of the other costs, this one clearly dominates.
\end{proof}

\subsection{Applications to specific settings}
\label{sec:apps}

Our main result in \Cref{cor:main-detailed} is stated in terms of general parameters, namely the approximate degree $d$ and local one-norm bound $g$.
To exemplify the applicability of our result, we instantiate our theorem for various well-studied cases of Lindbladians.
In particular, we consider the cases of geometrically local, general $k$-local, quasi-local, and power-law Lindbladians.
We show explicit choices of $g,d$ for each of these settings, resulting in a performance guarantee for our algorithm.

Moreover, in \Cref{sec:ham}, we show that a simplification of our algorithm can be applied to structure learning Hamiltonians with bounded local one-norm.
Surprisingly, our result demonstrates that dependence on the effective sparsity parameter of~\cite{bakshi2024structure} is not necessary for structure learning Hamiltonians.
The analysis of the algorithm is greatly simplified in this case compared with Lindbladian learning.

First, we consider the setting of geometrically local Lindbladians, which is arguably the most well-studied setting in both Hamiltonian and Lindbladian learning.

\begin{corollary}[Learning of strictly local Lindbladians]
    \label{coro:geo-local}
    Let $\lind_\lambda$ be a Lindbladian whose coefficients are bounded, $\abs{\lambda_{\pbar}} \leq 1$.
    Further suppose that $\lind$ is strictly local with respect to some geometry: every term has support size $k=\mathcal{O}(1)$, and every qubit interacts with at most $d$ nonzero terms.
    Then, \Cref{alg:full} uses a total evolution time of $\tet = \mathcal{O}(d^3 \log(n)/\epsilon^2)$ to learn an estimate $\widehat{\lambda}$ such that $\lonorm{\widehat{\lambda} - \lambda} < \eps$ with probability at least $0.99$.
    Moreover, the time resolution is $\tres = \Theta(1/d)$.
\end{corollary}

\begin{proof}
    In this case,
    \begin{equation}
        \lonorm{\lambda} = \max_{i \in [n]} \sum_{\pbar : S_{\pbar} \ni i} |\lambda_{\pbar}| \leq \max_{i \in [n]} |\{\pbar : S_{\pbar} \ni i\}| \leq d, 
    \end{equation}
    where we use that every qubit interacts with at most $d$ nonzero terms.
    Moreover, note that $\deg_\epsilon(\lambda) \leq d$.
    We can decompose $\lind_\lambda = \Lbig_\lambda + \Lsmall_\lambda$, where $\Lbig_\lambda = \lind_\lambda$ and $\Lsmall_\lambda = 0$.
    In this way, $\lonorm{\Lsmall_\lambda} = 0 < \epsilon$, and $\deg_\epsilon(\lind_\lambda) \leq \deg(\lind_\lambda) = d$.
    Taking $g = d$ in \Cref{thm:main} gives the result.
\end{proof}

For the following two applications, we need the following fact.

\begin{fact}[Volume bounds on lattices]
    \label{fact:volume-bounds}
    Consider a $p$-dimensional lattice for $p \geq 2$, and let $\ell \geq 1$.
    The number of $j$ such that $\dist(i, j) = \ell$ is $\leq 2^p \binom{p + \ell - 1}{p - 1} \leq 2^p (e(\ell + 1))^{p-1}$.
    Consequently, there are at most $(2e\ell)^{pk}$ sets $S$ of size $\leq k$ where $i \in S$ and $\dist(i, j) < \ell$ for all $j \in S$.
\end{fact}

Next, we consider learning quasi-local Lindbladians.
Such Lindbladians are especially interesting due to the recent surge of interest in quantum Gibbs samplers, which are constructed using quasi-local Lindbladians~\cite{chen2025efficient,chen2023efficient,ding2024efficient,scandi2026thermalization,rouze2026optimal,bakshi2026dobrushin,bakshi2024high,bergamaschi2024quantum,bergamaschi2026fast}.

\begin{corollary}[Learning of quasi-local Lindbladians]
    \label{coro:quasi-local}
    Let $p \geq 2$.
    Consider a system of $n$ qubits on a $p$-dimensional lattice, and let $\lind_\lambda$ be a $k$-local Lindbladian.
    Suppose $\lind_\lambda$ is also quasi-local with respect to the lattice, i.e., there is a parameter $\gamma > 0$ such that, for all $i,j \in [n]$,
    \begin{align}
        \label{eq:quasi-local}
        \sum_{\substack{\pbar \\ \{i, j\} \subseteq S_{\pbar}} }\abs{\lambda_{\pbar}} \leq e^{-\dist(i, j) / \gamma}.
    \end{align}
    Then, \Cref{alg:full} uses a total evolution time of $\tet = \mathcal{O}(d^2 \log(n)/\epsilon^2)$ to learn an estimate $\widehat{\lambda}$ such that $\lonorm{\widehat{\lambda} - \lambda} < \eps$ with probability at least $0.99$, where
    \begin{equation}
        d = C_k(C \gamma (p \log(1+\gamma) + \log(1/\epsilon)))^{pk},
    \end{equation}
    where $C_k$ is a constant that depends only on $k$, and $C$ is an absolute constant.
    Moreover, the time resolution is $\tres = \Theta(1)$.
\end{corollary}

\begin{proof}
    First, we can easily bound the local one-norm.
    Let $i \in [n]$.
    Then,
    \begin{equation}
        \sum_{\pbar : S_{\pbar} \ni i} |\lambda_{\pbar}| \leq 1,
    \end{equation}
    where in the inequality, we use \Cref{eq:quasi-local} for $j = i$.
    This holds for all $i \in [n]$, so we see that $\lonorm{\lambda} \leq 1$.

    To bound the approximate degree, consider decomposing $\lind_\lambda = \Lbig_\lambda + \Lsmall_\lambda$, where $\Lsmall$ is the part of $\lind_\lambda$ with coefficients indexed by $\pbar$ such that $\diam(S_{\pbar}) > r$, where $r = 4\gamma(p\log(2(1+2\gamma)) + \log(100\cdot 16^k/\epsilon))$, i.e.,
    \begin{equation}
        \lambdasmall_{\pbar} \triangleq \lambda_{\pbar}\iver{\diam(S_{\pbar}) > r},\qquad \lambdabig \triangleq \lambda - \lambdasmall.
    \end{equation}
    Then, if $\lonorm{\lambdasmall} \leq \epsilon/(100\cdot 16^k)$, we have $\deg_{\epsilon/(100\cdot 16^k)}(\lind_\lambda) \leq \deg(\Lbig)$.
    First, we show that $\lonorm{\lambdasmall} \leq \epsilon/(100\cdot 16^k)$.
    Let $i \in [n]$. Then,
    \begin{align}
        \sum_{\pbar : S_{\pbar} \ni i} |\lambdasmall_{\pbar}| &= \sum_{\substack{\pbar : S_{\pbar} \ni i\\\diam(S_{\pbar}) > r}} |\lambda_{\pbar}|\\
        &\leq \sum_{\substack{j \in [n]\\\dist(i,j) \geq r/2}} \sum_{\substack{\pbar\\\{i,j\} \subseteq S_{\pbar}}}|\lambda_{\pbar}|\\
        &\leq \sum_{\substack{j\in [n]\\\dist(i,j) \geq r/2}} e^{-\dist(i,j)/\gamma}\\
        &= \sum_{\ell=r/2}^{\infty} \sum_{\substack{j\in[n]\\\dist(i,j)= \ell}} e^{-\ell/\gamma}\\
        &\leq 2^p \sum_{\ell=r/2}^\infty \binom{p+\ell-1}{p-1} e^{-\ell/\gamma}\\
        &\leq 2^p e^{-r/(4\gamma)} \sum_{\ell=0}^\infty \binom{p+\ell-1}{p-1} e^{-\ell/(2\gamma)}\\
        &= 2^p e^{-r/(4\gamma)} (1-e^{-1/(2\gamma)})^{-p}\\
        &\leq e^{-r/(4\gamma)}(2(1 +2\gamma))^p\\
        &\leq \frac{\epsilon}{100\cdot 16^k}.
    \end{align}
    In the first line, we use the definition of $\lambdasmall$.
    In the second line, we use that if $i \in S_{\pbar}$ and $\diam(S_{\pbar}) > r$, there must exist two points $j, a \in S_{\pbar}$ such that $\dist(j,a) > r$, by definition of diameter.
    Then, by triangle inequality, one of the two points $j$ or $a$, say $j$, must satisfy $\dist(i,j) \geq r/2$.
    In the third line, we use \Cref{eq:quasi-local}.
    In the fifth line, we use \Cref{fact:volume-bounds}.
    In the sixth line, we use that $\ell \geq r/2$ and expand the number of terms in the sum.
    In the seventh line, we use the identity
    \begin{equation}
        \sum_{\ell=0}^\infty \binom{p+\ell-1}{p-1}x^\ell = (1-x)^{-p}
    \end{equation}
    for $x = e^{-1/(2\gamma)}$.
    In the eighth line, we use $(1-e^{-1/(2\gamma)})^{-1} \leq 1+2\gamma$ and that $p > 0$.
    In the last line, we use our choice of $r$.
    Since this holds for any $i \in [n]$, this completes the proof that $\lonorm{\lambdasmall} \leq \epsilon/(100\cdot 16^k)$.
    
    Now, we need to bound $\deg_{\epsilon/(100\cdot 16^k)}(\lind_\lambda) \leq \deg(\Lbig)$.
    By definition of $\lambdabig$, for a given site $i \in [n]$, if $i \in S_{\pbar}$ and $\lambdabig_{\pbar} \neq 0$, then $\diam(S_{\pbar}) \leq r$.
    By \Cref{fact:volume-bounds}, there are at most $(2e r)^{pk}$ possible supports that satisfy this.
    This corresponds to at most $16^k (2er)^{pk}$ possible Pauli terms $\pbar$.
    Thus,
    \begin{equation}
        d=\deg_{\epsilon/(100\cdot 16^k)}(\lind_\lambda) \leq 16^k(2er)^{pk} \leq C_k(C \gamma p \log(1+\gamma) + k + \log(1/\epsilon))^{pk},
    \end{equation}
    where $C > 0$ is an absolute constant, and $C_k$ is a constant that depends only on $k$.
    Plugging in these parameters to \Cref{thm:main} gives the result.
\end{proof}

In addition to Lindbladians with exponentially decaying correlations, we can also consider Lindbladians with weaker long-range correlations, where the decay is inverse polynomial rather than inverse exponential.

\begin{corollary}[Learning Lindbladians satisfying power-law decay]
    \label{coro:algebraic}
    Let $p \geq 2$.
    Consider a system of $n$ qubits on a $p$-dimensional lattice, and let $\lind_\lambda$ be a $k$-local Lindbladian.
    Suppose $\lind_\lambda$ obeys power-law decay with respect to the lattice, i.e., there is a parameter $\gamma > 0$ such that, for all $i,j \in [n]$,
    \begin{align}
        \label{eq:power-law}
        \sum_{\substack{\pbar \\ \{i, j\} \subseteq S_{\pbar}} }\abs{\lambda_{\pbar}} \leq \frac{1}{\max(1, \dist(i, j))^\gamma}.
    \end{align}
    Suppose that $p, k = \mathcal{O}(1)$ and $\gamma > p$ with $\gamma - p = \Omega(1)$.
    Let
    \begin{equation}
        \kappa \triangleq \frac{2pk}{\gamma - p}.
    \end{equation}
    Then, \Cref{alg:full} uses a total time evolution of $\tet = \mathcal{O}(2^{\gamma \kappa} \log(n)/\epsilon^{2 + \kappa})$
    to learn a $\widehat{\lambda}$ such that $\lonorm{\widehat{\lambda} - \lambda} < \eps$ with probability at least $0.99$.
    Moreover, the time resolution is $\tres = \Theta(1)$.
\end{corollary}

\begin{proof}
    As in the quasi-local case, $\lonorm{\lambda} \leq 1$.
    To bound the approximate degree, consider decomposing $\lind_\lambda = \Lbig_\lambda + \Lsmall_\lambda$, where $\Lsmall$ is the part of $\lind_\lambda$ with coefficients indexed by $\pbar$ such that $\diam(S_{\pbar}) > r$, where $r = \left(\frac{\epsilon(\gamma-p)}{100 \cdot 16^k \cdot 2^\gamma(2e)^p}\right)^{-1/(\gamma-p)}$, i.e.,
    \begin{equation}
        \lambdasmall_{\pbar} \triangleq \lambda_{\pbar}\iver{\diam(S_{\pbar}) > r},\qquad \lambdabig \triangleq \lambda - \lambdasmall.
    \end{equation}
    We first show that $\lonorm{\lambdasmall} < \epsilon/(100 \cdot 16^k)$.
    The calculation proceeds similarly to the quasi-local case, so we combine a few steps.
    Let $i \in [n]$. Then,
    \begin{align}
        \sum_{\pbar : S_{\pbar} \ni i} |\lambdasmall_{\pbar}| &\leq \sum_{\substack{j\in[n]\\\dist(i,j) \geq r/2}} \sum_{\substack{\pbar\\\{i,j\} \subseteq S_{\pbar}}} |\lambda_{\pbar}|\\
        &\leq \sum_{\substack{j\in[n]\\\dist(i,j) \geq r/2}} \frac{1}{\dist(i,j)^\gamma}\\
        &= \sum_{\ell=r/2}^\infty \sum_{\substack{j\in[n]\\\dist(i,j) = \ell}} \frac{1}{\ell^\gamma}\\
        &\leq 2^p e^{p-1}\sum_{\ell=r/2}^\infty \frac{(\ell+1)^{p-1}}{\ell^\gamma}\\
        &\leq (4e)^p \sum_{\ell=r/2}^\infty \ell^{p-1-\gamma}\\
        &\leq (4e)^p \int_{r/2-1}^\infty x^{p-1-\gamma}\,dx\\
        &\leq (4e)^p \frac{r^{p-\gamma}}{2^{p-\gamma}(\gamma - p)}\\
        &\leq \frac{\epsilon}{100\cdot 16^k}.
    \end{align}
    In the second line, we use \Cref{eq:power-law}.
    In the fourth line, we use \Cref{fact:volume-bounds}.
    In the fifth line, we use that $\ell + 1 \leq 2\ell$ for $\ell \geq 1$.
    In the last line, we use our choice of $r$.

    Now, we need to bound $\deg_{\epsilon/(100\cdot 16^k)}(\lind_\lambda) \leq \deg(\Lbig)$.
    By definition of $\lambdabig$, for a given site $i \in [n]$, if $i \in S_{\pbar}$ and $\lambdabig_{\pbar} \neq 0$, then $\diam(S_{\pbar}) \leq r$.
    By \Cref{fact:volume-bounds}, there are at most $(2e r)^{pk}$ possible supports that satisfy this.
    This corresponds to at most $16^k (2er)^{pk}$ possible Pauli terms $\pbar$.
    Thus, for $\gamma - p = \Omega(1)$ and $k,p = \mathcal{O}(1)$, we have
    \begin{equation}
        d=\deg_{\epsilon/(100\cdot 16^k)}(\lind_\lambda) \leq 16^k(2er)^{pk} = \mathcal{O}\left(\frac{2^\gamma}{\epsilon}\right)^{pk/(\gamma-p)}.
    \end{equation}
    Plugging into \Cref{thm:main} gives the result.
\end{proof}

Finally, we can also instantiate our theorem for general $k$-local Lindbladians with no geometric constraints.

\begin{corollary}[Learning local Lindbladians]
    \label{coro:local}
    Let $\lind_\lambda$ be a $k$-local Lindbladian on $n$ qubits such that $\lonorm{\lambda} \leq g$.
    Then, \Cref{alg:full} uses a total time evolution of $\tet = \mathcal{O}(gn^{2k - 2}\log(n)/\eps^2)$ to learn an estimate $\widehat{\lambda}$ such that $\lonorm{\widehat{\lambda} - \lambda} < \epsilon$ with probability at least $0.99$.
    Moreover, the time resolution is $\tres = \Theta(1/g)$.
\end{corollary}

\begin{proof}
    In this case, we can naively bound $d$ as
    \begin{equation}
        d \leq \deg(\lind_\lambda) \leq 16^k\binom{n-1}{k-1} \leq 16^k n^{k-1}.
    \end{equation}
    Plugging this into \Cref{thm:main} gives the result.
\end{proof}

\subsection{Structure learning Hamiltonians}
\label{sec:ham}

In this section, we show that our framework can be applied to the problem of structure learning Hamiltonians from both real-time evolution and high-temperature Gibbs states.
While we cannot simply instantiate our theorem for new choices of $g$ and $d$ in this case, we find that analyzing our algorithm for Hamiltonians is significantly simpler than that for Lindbladians.
The reason for this simplification is the absence of ``confusion'' (in the sense of \Cref{sec:local-fourier}).
Thus, the matrix $A$ defined in \Cref{not:A} is the identity matrix, and we can keep track of the error of our iterates in $\infty\to\infty$ norm instead.

Let $H = \ii\sum_P \lambda_P P$ be a $k$-local Hamiltonian.
The factor of $\ii$ ensures that the coefficients $\lambda_P$ are purely imaginary, which is consistent with our definition of the coherent part of a Lindbladian in \Cref{sec:lindblad}.
We consider two different access models to this Hamiltonian: the ability to evolve under $e^{-\ii Ht}$ for a chosen time $t > 0$ or access to many copies of high-temperature Gibbs states $\rho_\beta = e^{-\beta H}/\tr(e^{-\beta H})$ for a small inverse temperature $\beta > 0$.
Similarly to the Lindbladian case, we define a vector of expectation values with entries
\begin{equation}
    E_P(x) \triangleq \E_{R\sim \pauli{S_P}}[\ntr(e^{-\ii H(x)t}R e^{\ii H(x)t} R P)], \quad \text{or}\quad E_P(x) \triangleq \tr(P\rho_\beta(x)),
\end{equation}
for the real-time evolution and Gibbs state cases, respectively\footnote{For the real-time evolution case, we could also use the canonical choice of observables typically used in the Hamiltonian learning literature. However, we use this choice for a closer analogy with the Lindbladian case.}.
Moreover, as before, we define
\begin{equation}
    \mathcal{F}_P(x) \triangleq \frac{1}{t}E_P(x) - \frac{1}{t}E_P(\lambda),\quad \text{or}\quad \mathcal{F}_P(x) \triangleq -\frac{1}{\beta}E_P(x) + \frac{1}{\beta}E_P(\lambda),
\end{equation}
respectively\footnote{Note that the sign is switched for the Gibbs state version to make the first order terms match. Namely, if the signs are not flipped, the first order term of $\mathcal{F}_P$ for the dynamics version is $x_P$ while for the Gibbs state version it is $-x_P$.}.
Define the Jacobian $J(x)$ of $\mathcal{F}(x)$ entrywise as $J_{P,Q}(x) \triangleq \partial_{x_Q} \mathcal{F}_P(x)$.
As in \Cref{sec:algo}, we consider an approximation $\widehat{\mathcal{F}}(x)$ of $\mathcal{F}(x)$ such that
\begin{equation}
    \label{eq:f-hat-ham}
    \eta(x) = \widehat{\mathcal{F}}(x) - \mathcal{F}(x),\qquad \norm{\eta(x)}_\infty \leq \eta.
\end{equation}
One can obtain such an approximation via the same approach as in \Cref{sec:sample-time}.
Our algorithm is detailed in \Cref{alg:ham}.

\begin{algorithm}
    \caption{Structure learning algorithm for Hamiltonians}
    \label{alg:ham}
    \KwData{Accuracy $\epsilon > 0$; time $t$ satisfying $t < t_{\max}$ or inverse temperature $\beta$ satisfying $\beta < \beta_{\max}$; estimates $\widehat{\mathcal{F}}(x)$ satisfying \Cref{eq:f-hat-ham} for given inputs $x$.}
    \KwResult{Estimates $\widehat{\lambda}$ such that $\norm{\lambda - \widehat{\lambda}}_\infty \leq \epsilon$.}
    Initialize $x^{(0)} = 0 \in \mathbb{R}^M$ and $T = \left\lceil \log_2(1/\epsilon) \right\rceil$.\\
    \For{$j=0,\dots, T - 1$}{
        Set $\epsilon_j\triangleq 2^{-j}$.\\
        Update
        \begin{equation}
            \label{eq:ham-update}
            x^{(j+1)} = \round_{\eps_j /4}\left(x^{(j)} - \widehat{\mathcal{F}}(x^{(j)})\right).
        \end{equation}
    }
    \Return{$\widehat{\lambda} \triangleq x^{(T)}$.}
\end{algorithm}

First, we prove a general theorem stating that, if we have an approximation for $\widehat{\mathcal{F}}(x)$ and a bound on the Jacobian of $\mathcal{F}(x)$, then the algorithm in \Cref{alg:ham} obtains good estimates of the Hamiltonian parameters.
In the following sections, we prove both of these hypotheses.
In particular, we can obtain an approximation for $\widehat{\mathcal{F}}(x)$ similarly to \Cref{sec:sample-time}.

\begin{theorem}[Structure learning of Hamiltonians]
    \label{thm:ham}
    Let $\epsilon > 0$.
    Let $H = \ii\sum_P \lambda_P P$ be a $k$-local Hamiltonian with $\lonorm{\lambda} \leq g$ and $\norm{\lambda}_\infty \leq 1$.
    Let $S_H$ be the support of the Hamiltonian, which is unknown to the algorithm.
    Let $c_g > 0$ be a constant depending on $g$.
    Let $0 < t,\beta < 1/(20c_{2g})$.
    Suppose we can compute estimates $\widehat{\mathcal{F}}(x)$ for given inputs $x$ such that
    \begin{equation}
        \norm{\mathcal{F}(x) - \widehat{\mathcal{F}}(x)}_\infty \leq \frac{\epsilon}{20}.
    \end{equation}
    Also suppose that, for any $x$ such that $\lonorm{x} \leq g$ and for any $P \in \pauli{k}$,
    \begin{equation}
        \sum_{Q \in S_H}|J_{P,Q}(x) - \delta_{P,Q}| \leq c_g t,\quad \text{or}\quad \sum_{Q \in S_H}|J_{P,Q}(x) - \delta_{P,Q}| \leq c_g \beta,
    \end{equation}
    for the real-time and Gibbs state settings, respectively.
    Then, \Cref{alg:ham} finds estimates $\widehat{\lambda}$ such that $\norm{\widehat{\lambda} - \lambda}_\infty \leq \epsilon$.
\end{theorem}

First, under the hypothesis of \Cref{thm:ham}, we can also bound the higher order terms of $\mathcal{F}$ via the Fundamental Theorem of Calculus.
This is the analogue of \Cref{coro:f-bound}.

\begin{lemma}
    \label{lem:f-bound-ham}
    Let $\lonorm{x} \leq g$, and let $\Delta = x - \lambda$.
    Suppose that $\Delta_P = 0$ unless $P \in S_H$, and suppose for any $P \in \pauli{k}$ that
    \begin{equation}
        \sum_{Q \in S_H}|J_{P,Q}(x) - \delta_{P,Q}| \leq c_g t.
    \end{equation}
    Then,
    \begin{equation}
        \norm{\mathcal{F}(x) - \Delta}_\infty \leq c_{2g}t \norm{\Delta}_\infty,\quad \text{or}\quad  \norm{\mathcal{F}(x) - \Delta}_\infty \leq c_{2g}\beta \norm{\Delta}_\infty,
    \end{equation}
    for the real-time and Gibbs state cases, respectively.
\end{lemma}

\begin{proof}
    As in the proof of \Cref{coro:f-bound}, consider $f_P: [0,1] \to \C$ defined by $f_P(s) \triangleq E_P(\lambda + s\Delta)$.
    Then, by the Fundamental Theorem of Calculus,
    \begin{equation}
        f_P(1) - f_P(0) = \int_0^1 \partial_s f_P(s)\,ds.
    \end{equation}
    Expanding both sides and using $\partial_s = \sum_Q \Delta_Q \partial_Q$, we see that
    \begin{equation}
        \mathcal{F}_P(x) = \int_0^1 \sum_Q \Delta_Q J_{P,Q}(\lambda + s\Delta)\,ds.
    \end{equation}
    Subtracting $\Delta$ from both sides, we have
    \begin{equation}
        \mathcal{F}_P(x) - \Delta_P = \int_0^1 \sum_Q \Delta_Q (J_{P,Q}(\lambda + s\Delta) - \delta_{P,Q})\,ds = \int_0^1 \sum_{Q \in S_H} \Delta_Q (J_{P,Q}(\lambda + s\Delta) -\delta_{P,Q})\,ds,
    \end{equation}
    where in the last equality, we use that $\Delta_Q = 0$ unless $Q \in S_H$.
    Taking the absolute value of both sides, we have
    \begin{align}
        |\mathcal{F}_P(x) - \Delta_P| &\leq \int_0^1 \sum_{Q \in S_H}|J_{P,Q}(\lambda + s\Delta) - \delta_{P,Q}| |\Delta_Q| \,ds\\
        &\leq \int_0^1 \sum_{Q \in S_H}|J_{P,Q}(\lambda + s\Delta) - \delta_{P,Q}|\,ds \cdot \norm{\Delta}_\infty\\
        &\leq c_{2g}t\norm{\Delta}_\infty.
    \end{align}
    In the last line, we use the bound on the Jacobian for $\lambda + s\Delta$, where $\lonorm{\lambda + s\Delta} \leq 2g$.
    This holds for all $P \in \pauli{k}$, so this gives the claim.
    The proof is the same for the Gibbs state case.
\end{proof}

With this, we can prove \Cref{thm:ham}.

\begin{proof}[Proof of~{\Cref{thm:ham}}]
    In this case, the proof is short.
    Let $j \in \{0,\dots, T-1\}$.
    We prove this via induction on $j$, where at each iteration, we maintain the invariants
    \begin{equation}
        \label{eq:inductive-hypo-ham}
        \begin{gathered}
            \norm{x^{(j)} - \lambda}_\infty \leq \epsilon_j,\\
            |x^{(j)}_P| \leq 2|\lambda_P|
        \end{gathered}
    \end{equation}
    For the base case of $j = 0$, recall that $x^{(0)} = 0$ and $\epsilon_0 = 1$.
    Thus, we have $\norm{x^{(0)} - \lambda}_\infty = \norm{\lambda}_\infty \leq 1$, as required.
    Moreover, $|x_P^{(0)}| = 0 \leq 2|\lambda_P|$ is trivially satisfied.

    For the inductive step, suppose that $\norm{x^{(j)} - \lambda}_\infty \leq \epsilon_j$ and $|x_P^{(j)}| \leq 2|\lambda_P|$.
    We prove that this holds for iteration $j + 1$.
    For brevity, we drop the iteration index.
    Let $x\triangleq x^{(j)}$ denote the current iterate, $x^+ \triangleq x^{(j+1)}$ the next iterate, $\Delta \triangleq x - \lambda$ the error vector of the current iterate, $\Delta^+ \triangleq x^+ -\lambda$ the error vector of the next iterate, $\eps\triangleq \epsilon_j$ the current error, and $\epsilon^+ \triangleq \epsilon_{j+1} = \epsilon/2$ the desired error of the next iterate.
    We use $y$ to denote the next iterate before rounding:
    \begin{equation}
        x^+ \triangleq \round_{\epsilon/4}(y),\qquad y \triangleq x - \widehat{\mathcal{F}}(x).
    \end{equation}
    Note that it suffices to show that
    \begin{equation}
        \norm{y - \lambda}_\infty \leq \frac{\epsilon}{10}.
    \end{equation}
    This implies \Cref{eq:inductive-hypo-ham} with error parameter $\epsilon^+$:
    \begin{equation}
        \norm{x^+ - \lambda}_\infty \leq \norm{x^+ - y}_\infty + \norm{y - \lambda}_\infty \leq \frac{\epsilon}{4} + \frac{\epsilon}{10} \leq \frac{\epsilon}{2} = \epsilon^+.
    \end{equation}
    For the second hypothesis in \Cref{eq:inductive-hypo-ham}, consider a parameter indexed by a Pauli $P$.
    Then, either the rounding kicks in so that $x_P^+ = 0$, in which case the bound is immediate, or $|y_P| > \epsilon/4$.
    In the latter case, by the reverse triangle inequality, then
    \begin{equation}
        |\lambda_P| \geq |y_P| - \frac{\epsilon}{10} > \frac{3\epsilon}{20},
    \end{equation}
    so
    \begin{equation}
        |x_P^+| = |y_P| \leq |\lambda_P| + \frac{\epsilon}{10} \leq 2|\lambda_P|,
    \end{equation}
    so that \Cref{eq:inductive-hypo-ham} holds.
    Thus, it suffices to show that $\norm{y - \lambda}_\infty \leq \epsilon/10$.
    We show this as follows:
    \begin{equation}
        y - \lambda = x - \widehat{\mathcal{F}}(x) - \lambda = \Delta - \mathcal{F}(x) \underbrace{-\eta(x)}_{\triangleq\mathrm{err}_1} = \Delta - \Delta + \underbrace{\left(\Delta - \mathcal{F}(x)\right)}_{\triangleq \mathrm{err}_2} + \mathrm{err}_1 = \mathrm{err}_2 + \mathrm{err}_1.
    \end{equation}
    Then, we can bound each of the errors as follows:
    \begin{equation}
        \norm{\err_1}_\infty \leq \frac{\epsilon}{20}
    \end{equation}
    by the hypothesis of the theorem.
    Also,
    \begin{equation}
        \norm{\err_2}_\infty = \norm{\Delta - \mathcal{F}(x)}_\infty \leq c_{2g}t\norm{\Delta}_\infty \leq \frac{\epsilon}{20},
    \end{equation}
    where we use the inductive hypothesis and \Cref{lem:f-bound-ham}.
    Note that the hypothesis of \Cref{lem:f-bound-ham}, that $\Delta_P = 0$ unless $P \in S_H$, holds because we maintain the invariant $|x_P| \leq 2|\lambda_P|$.
    Putting everything together, $\norm{y-\lambda}_\infty \leq \epsilon/10$, as required.
    The proof is the same in the Gibbs state case via \Cref{lem:f-bound-ham}.
\end{proof}

\subsubsection{Real-time evolution}

We can instantiate \Cref{thm:ham} when we are given access to real-time evolution under the unknown Hamiltonian $H$.
We do so by proving the hypotheses of \Cref{thm:ham} hold in this setting.

\begin{theorem}[Structure learning of Hamiltonians from real-time evolution]
    \label{thm:ham-time}
    Let $\epsilon,\delta > 0$, and let $0 < t < 1/(162kg)$.
    Let $H = \ii\sum_P \lambda_P P$ be a $k$-local Hamiltonian with $\lonorm{\lambda} \leq g$ and $\norm{\lambda}_\infty \leq 1$.
    Then, there exists an algorithm that finds estimates $\widehat{\lambda}$ such that $\norm{\widehat{\lambda} - \lambda}_\infty \leq \epsilon$ with probability at least $1-\delta$ using $\tet = \mathcal{O}(g\log(n/\delta)/\epsilon^2)$ and classical runtime $\mathcal{O}\left(\frac{n^k g^2}{\epsilon^2}\log\left(\frac{1}{\epsilon}\right)\log\left(\frac{n}{\delta \epsilon}\right)\right)$.
\end{theorem}

In order to prove this theorem, we require the series expansion properties we proved in \Cref{sec:series}.
The classical runtime guarantee follows by the same analysis as \Cref{sec:sample-time}.
Define the Jacobian of $\mathcal{F}$ as $J_{P,Q}(x) \triangleq \partial_{x_Q} \mathcal{F}_P(x)$.
We want to prove a bound on the higher order terms of this Jacobian.
The reason the bounds in \Cref{sec:series} do not apply is because we will want to keep track of the error of our estimates in \Cref{alg:ham} in $\infty$-norm rather than $B_1$-norm.
Thus, the bounds in \Cref{sec:series} do not apply, as there we only bound the higher order terms of the Jacobian in $B_1\to B_1$ norm (\Cref{lem:convexity-lindblad}), not $\infty\to\infty$ norm.
In fact, in the Lindbladian case, the $\infty\to\infty$ norm of the Jacobian can scale with the system size.
Luckily, in the Hamiltonian case, the $\infty\to\infty$ norm of the Jacobian is bounded.

\begin{lemma}
    \label{lem:convexity}
    Let $H(\lambda) = \ii\sum_P \lambda_P P$ be a $k$-local Hamiltonian with bounded local one-norm $\lonorm{\lambda} \leq g$.
    Suppose $t > 0$ satisfies $t < 1/(162kg)$.
    Then, for any $x$ such that $\lonorm{x} \leq g$, $\norm{J(x) - I}_{\infty\to\infty} \leq c_gt$, where $c_g = 81kg/20$.
\end{lemma}

First, we show that, with this lemma, we can obtain \Cref{thm:ham-time}.

\begin{proof}[Proof of {\Cref{thm:ham-time}}]
    We proceed as in the proof of \Cref{thm:complexity-general}.
    First, by \Cref{lem:estimate-coefs}, we can obtain estimates $\widehat{E}_P(\lambda)$ such that
    \begin{equation}
        |\widehat{E}_P(\lambda) - E_P(\lambda)| \leq \frac{t\epsilon}{40}
    \end{equation}
    for all $k$-local Paulis $P$ with probability at least $1-\delta$ using $\Theta(C^k \log(n/\delta)/(t\epsilon)^2)$ queries to the time evolution operator $e^{-\ii Ht}$, for some absolute constant $C > 0$.
    This corresponds to a total time evolution of $\Theta(C^k \log(n/\delta)/(t\epsilon^2)) = \Theta(C_k g\log(n/\delta)/\epsilon^2)$ for some constant $C_k > 0$ that depends only on the locality $k$.
    As in \Cref{thm:complexity-general}, we can obtain estimates $\widehat{E}_P(x)$ such that
    \begin{equation}
        |\widehat{E}_P(x) - E_P(x)| \leq \frac{t\epsilon}{40}
    \end{equation}
    also by Taylor expanding $e^{-\ii Ht}$ up to degree
    \begin{equation}
        \Gamma \triangleq \left\lceil \frac{\log(160/(t\epsilon))}{\log(1/(2ekgt))} - 1 \right\rceil,
    \end{equation}
    by the same analysis as before.
    This gives us an approximation of $\mathcal{F}(x)$ up to $\epsilon/20$ error.
    Moreover, the Jacobian bound needed in \Cref{thm:ham} is clearly implied by \Cref{lem:convexity}:
    \begin{equation}
        \sum_{Q \in S_H}|J_{P,Q}(x) -\delta_{P,Q}| \leq \sum_Q |J_{P,Q}(x) - \delta_{P,Q}| = \norm{J(x) - I}_{\infty\to\infty} \leq c_g t.
    \end{equation}
\end{proof}

It remains to prove \Cref{lem:convexity}, and we spend the rest of this section proving it.
First, we require an inductive bound on the operator norm of a nested commutator.
This is an analogue of \Cref{lem:inductive-bound}.

\begin{lemma}
    \label{lem:inductive-local-norm}
    Let $Q \in \pauli{k}$. Suppose $H = \ii\sum_P \lambda_P P$ is a $k$-local Hamiltonian with $|\lambda_P| \leq 1$ and $\lonorm{\lambda} \leq g$. Then, $[H,Q]_\ell = \sum_T d_{T,\ell} T$, where $\mathrm{supp}(T) \leq k(\ell+1)$ and
    \begin{equation}
        \norm{[H,Q]_\ell}_{P,1} \triangleq \sum_T |d_{T,\ell}| \leq \ell! (2kg)^\ell.
    \end{equation}
\end{lemma}

\begin{proof}
We prove this by induction on $\ell$.
For the base case of $\ell=0$, the claim is clear because $[H,Q]_0 = Q$ so that $\sum_T |d_T| = 1$ and $\supp(Q) \leq k$.
For the inductive step, suppose the result holds for $\ell$.
Then,
\begin{equation}
    [H,Q]_{\ell+1} = [H,[H,Q]_\ell] = \ii\sum_P \lambda_P [P, [H, Q]_\ell] = \ii\sum_P \lambda_P \sum_T d_{T,\ell}[P, T],
\end{equation}
where in the last equality we use the inductive hypothesis.
By the inductive hypothesis, $\supp(T) \leq k(\ell+1)$.
Since $|S_P| \leq k$, $|\supp([P, T])| \leq k(\ell+2)$.
We can bound the 1-norm of the Pauli coefficients.
Note that $\lambda_P d_{T,\ell}$ is only included when $[P, T] \neq 0$.
\begin{equation}
    \sum_T |d_{T,\ell+1}| \leq 2 \sum_P \sum_T |\lambda_P||d_{T,\ell}| \iver{T \text{ overlaps with } P}
\end{equation}
Let $\supp(T) = \{i_1,\dots, i_{k(\ell+1)}\}$ (the argument also works for fewer qubits).
Then, we can break up the sum over $P$ according to whether $P$ is supported on some qubit $i_j$.
\begin{align}
    \sum_T |d_{T,\ell+1}| &\leq 2 \sum_T |d_{T,\ell}| \left(\sum_{P : S_P \ni i_1} |\lambda_P| + \cdots + \sum_{P : S_P \ni i_{k(\ell+1)}}|\lambda_P|\right)\\
    &\leq 2g k(\ell+1) \sum_T |d_{T,\ell}|\\
    &\leq 2g k(\ell+1) \ell! (2kg)^\ell\\
    &= (\ell+1)!(2kg)^{\ell+1}.
\end{align}
In the second line, we use that $\lonorm{\lambda} \leq g$.
In the third line, we use the inductive hypothesis.
\end{proof}

We also find the following two claims useful.
The first claim provides an explicit expression for the derivative of a nested commutator.

\begin{claim}
    \label{claim:deriv-commutator}
    Let $Q, T \in \pauli{n}$.
    For all $\ell \geq 1$ and $x \in [-1,1]^m$,
    \begin{equation}
        \partial_{x_Q} [H(x), T]_\ell = \ii \sum_{j=0}^{\ell-1} [H(x), [Q, [H(x), T]_{\ell-1-j}]]_j.
    \end{equation}
\end{claim}

\begin{proof}
This is essentially the Leibniz rule, but we prove it for completeness.
We proceed via induction.
For the base case of $\ell = 1$, $\partial_{x_Q} [H(x), T] = [\partial_{x_Q} H(x), T] = [Q, T]$, which is clearly the same as the right-hand side (only the $j = 0$ term survives and $\ell - 1- j = 0$).
For the inductive case, suppose the claim holds for $\ell$.
\begin{align}
    \partial_{x_Q} [H(x), T]_{\ell+1} &= \partial_{x_Q} [H(x), [H(x), T]_\ell]\\
    &= [\partial_{x_Q} H(x), [H(x), T]_\ell] + [H(x), \partial_{x_Q}[H(x), T]_\ell]\\
    &= \ii [Q, [H(x), T]_\ell] + \ii \sum_{j=0}^{\ell-1} [H(x), [H(x), [Q, [H(x), T]_{\ell-1-j}]]_j]\\
    &= \ii [Q, [H(x), T]_\ell] + \ii \sum_{j=1}^{\ell} [H(x), [Q, [H(x), T]_{\ell-j}]]_j\\
    &= \ii \sum_{j=0}^\ell [H(x), [Q, [H(x), T]_{\ell-j}]]_j.
\end{align}
Here, the second line follows from the Leibniz rule.
The third line follows by the inductive hypothesis.
The fourth line follows by shifting the index $j \to j+1$.
\end{proof}

\begin{claim}
    \label{claim:exp-sum}
    Let $E_1,\dots, E_M \in \mathcal{P}_n$ be pairwise distinct Paulis.
    Then, for any observables $A,B$,
    \begin{equation}
        \sum_{b=1}^M \left|\ntr(A[E_b, B])\right| \leq 2\norm{A}_{P,1} \norm{B}_{P,1},
    \end{equation}
    where, for $A = \sum_{R\in \mathcal{P}_n} \alpha_R R$, $\norm{A}_{P,1} = \sum_R |\alpha_R|$.
\end{claim}

\begin{proof}
Write $A = \sum_{R \in \mathcal{P}_n} \alpha_R R$ and $B = \sum_{S \in \mathcal{P}_n} \beta_S S$.
Then, we have
\begin{equation}
    \sum_{b=1}^M \left|\ntr(A[E_b, B])\right| \leq \sum_{R,S} |\alpha_R| |\beta_S| \sum_{b=1}^M \left|\ntr(R[E_b,S])\right|.
\end{equation}
    Thus, it suffices to show that, for any Paulis $R,S$,
\begin{equation}
    \sum_{b=1}^M \left|\ntr(R[E_b,S])\right| \leq 2.
\end{equation}
Notice that
\begin{equation}
    [E_b, S] = \begin{cases}
        0 & \text{if } [E_b, S] = 0\\
        2 \omega T & \text{if } [E_b, S] \neq 0
    \end{cases},
\end{equation}
where $T \in \mathcal{P}_n$ such that $\omega T = E_bS$, where $\omega \in \{\pm 1, \pm \ii\}$.
Then,
\begin{equation}
    \ntr(R[E_b, S]) = \begin{cases}
        0 & \text{if } [E_b, S] = 0\\
        2\omega \indicator\{T = R\} & \text{if } [E_b, S] \neq 0
    \end{cases}.
\end{equation}
We claim that there can be only one $b \in [M]$ such that $[E_b, S] \neq 0$ and $T = R$.
This means that only one $b$ contributes to the sum over $b$, so the result follows.
This is true because for $[E_b, S] \neq 0$ and $T = R$ to hold, $E_b = \omega R S$ (since $\omega T = E_b S$ when $[E_b,S] \neq 0$).
Thus, given $R$ and $S$, $E_b$ is uniquely determined amongst the set of pairwise distinct (phaseless) Paulis.
\end{proof}

With all of these lemmas, we can now prove \Cref{lem:convexity}.

\begin{proof}[Proof of~{\Cref{lem:convexity}}]
Let $x$ be such that $\lonorm{x} \leq g$.
Recall by the definition of $\mathcal{F}_P$,
\begin{align}
    \mathcal{F}_P(x) &= \frac{1}{t}E_P(x) - \frac{1}{t}E_P(\lambda)\\
    &= \frac{1}{t}\E_{R\sim\pauli{S_P}} [\ntr(e^{-\ii Ht} R e^{\ii Ht} RP)] - \frac{1}{t}E_P(\lambda)\\
    &= \frac{1}{t}\E_{R\sim\pauli{S_P}} [\ntr(R e^{\ii Ht} RP e^{-\ii Ht})] - \frac{1}{t}E_P(\lambda)\\
    &= -\E_{R\sim \pauli{S_P}}[\ntr(R[H(x),RP])] + \frac{1}{t}\sum_{\ell=2}^\infty \frac{(\ii t)^\ell}{\ell!} \E_{R\sim \pauli{S_P}}[\ntr(R[H(x), RP]_\ell)] - \frac{1}{t}E_P(\lambda)\\
    &= -\sum_Q x_Q \E_{R\sim\pauli{S_P}}[\ntr(R[Q, RP])] + \frac{1}{t}\sum_{\ell=2}^\infty \frac{(\ii t)^\ell}{\ell!} \E_{R\sim \pauli{S_P}}[\ntr(R[H(x), RP]_\ell)] - \frac{1}{t}E_P(\lambda)\\
    &= -\sum_Q x_Q \left(\E_{R\sim\pauli{S_P}}[\ntr(RQRP)] - \E_{R\sim\pauli{S_P}}[\ntr(PQ)]\right) + \frac{1}{t}\sum_{\ell=2}^\infty \frac{(\ii t)^\ell}{\ell!} \E_{R\sim \pauli{S_P}}[\ntr(R[H(x), RP]_\ell)] - \frac{1}{t}E_P(\lambda) \\
    &= x_P + \frac{1}{t}\sum_{\ell=2}^\infty \frac{(\ii t)^\ell}{\ell!} \E_{R\sim \pauli{S_P}}[\ntr(R[H(x), RP]_\ell)] - \frac{1}{t}E_P(\lambda),
\end{align}
where in the last line, we use \Cref{lem:restricted-orthogonality}.
We highlight that, unlike in the Lindbladian setting, there is no confusion (in the sense of \Cref{sec:local-fourier,sec:local-lindblad-fourier}).
This is what makes the analysis of the algorithm much simpler because the $A$ matrix from \Cref{not:A} is now just the identity matrix.
Then, the Jacobian is given by
\begin{equation}
    \label{eq:jacobian}
    J_{P,Q}(x) = \partial_{x_Q} \mathcal{F}_P(x) = \delta_{P,Q} + \frac{1}{t}\sum_{\ell=2}^{\infty} \frac{(\ii t)^\ell}{\ell!} \E_{R\sim\pauli{S_P}} [\partial_{x_Q} \ntr(R[H(x),RP]_\ell)].
\end{equation}
Thus, the leading order term of the Jacobian is $I$.
It suffices to show that the higher order terms are small, i.e., $\norm{J(x) - I}_{\infty \to \infty} \leq c_g t$.
Recall that $\norm{M}_{\infty\to\infty} = \sup_{\norm{x}_\infty \leq 1} \norm{Mx}_\infty$.
Consider $u = (u_1,\dots, u_M)$ such that $\norm{u}_\infty \leq 1$.
We will bound $\norm{(J(x) - I)u}_\infty$.
\begin{align}
    |((J(x) - I)u)_P| &= \left|\sum_Q (J(x) - I)_{P,Q} u_Q\right|\\
    &= \frac{1}{t}\left|\sum_Q u_Q\left(\sum_{\ell=2}^\infty \frac{(\ii t)^\ell}{\ell!} \E_{R\sim\pauli{S_P}}[\partial_{x_Q} \ntr(R[H(x),RP]_\ell)]\right)\right|\\
    &\leq \frac{1}{t}\sum_{\ell=2}^\infty \frac{t^\ell}{\ell!} \left(\sum_Q \E_{R\sim \pauli{S_P}}[\left|\partial_{x_Q} \ntr(R[H(x), RP]_\ell)\right|]\right),
\end{align}
where in the second line we use \Cref{eq:jacobian}, and in the third line, we use $|u_Q| \leq 1$ and the triangle inequality.
Now, using \Cref{claim:deriv-commutator}, we can expand the derivative of the nested commutator:
\begin{equation}
    |((J(x) - I)u)_P| \leq \frac{1}{t}\sum_{\ell=2}^\infty \frac{t^\ell}{\ell!} \sum_{j=0}^{\ell-1} \sum_Q \E_{R\sim \pauli{S_P}}\left[\left|\ntr(R[H(x), [Q, [H(x), RP]_{\ell-1-j}]]_j)\right|\right]
\end{equation}
To get this into the correct form to apply \Cref{claim:exp-sum}, we can repeatedly use $\tr(A[H,B]) = \tr([A,H]B)$.
Thus, we have
\begin{align}
    |((J(x) - I)u)_P| &\leq \frac{1}{t}\sum_{\ell=2}^\infty \frac{t^\ell}{\ell!} \sum_{j=0}^{\ell-1}\E_{R\sim\pauli{S_P}}\left[\sum_Q \left|\ntr([R,H(x)]_j [Q, [H(x), RP]_{\ell-1-j}])\right|\right]\\
    &\leq \frac{1}{t}\sum_{\ell=2}^\infty \frac{t^\ell}{\ell!} \sum_{j=0}^{\ell-1} \E_{R\sim \pauli{S_P}}[2\norm{[R, H(x)]_j}_{P,1} \norm{[H(x), RP]_{\ell-j-1}}_{P,1}]\\
    &\leq \frac{2}{t}\sum_{\ell=2}^\infty \frac{t^\ell}{\ell!} \sum_{j=0}^{\ell-1} j! (2kg)^j (\ell-1-j)! (2kg)^{\ell-1-j}\\
    &\leq \frac{2}{t} \sum_{\ell=2}^\infty t^\ell (2kg)^{\ell-1}\\
    &= 2\sum_{\ell=1}^\infty (2kgt)^\ell\\
    &= 2 \cdot \frac{2kgt}{1-2kgt}\\
    &\leq t\cdot \frac{81kg}{20}.
\end{align}
In the second line, we use Claim~\ref{claim:exp-sum}.
In the third line, we use \Cref{lem:inductive-local-norm}, which can be applied because $\lonorm{x} \leq g$.
In the fourth line, we use that
\begin{equation}
    \sum_{j=0}^{\ell-1} j! (\ell-1-j)! = \sum_{j=0}^{\ell-1}(\ell-1)! \frac{1}{\binom{\ell-1}{j}} \leq (\ell-1)! \sum_{j=0}^{\ell-1} 1 = \ell!.
\end{equation}
In the fifth line, we shift the index $\ell \to \ell-1$.
In the sixth line, we use the sum of a geometric series when $2kg t < 1$.
Finally, in the last line, we use $t < 1/(162kg)$.
Thus, this implies that $\norm{J(x) - I}_{\infty\to\infty} \leq c_gt$.
\end{proof}

\subsubsection{High-temperature Gibbs states}

We can also apply \Cref{thm:ham} to the problem of structure learning a bounded degree Hamiltonian from copies of its Gibbs state.
This is the first algorithm for structure learning Hamiltonians from any Gibbs state access model.
Again, we prove this by showing that the hypotheses of \Cref{thm:ham} hold in this setting.

\begin{theorem}[Structure learning of Hamiltonians from high-temperature Gibbs states]
    \label{thm:ham-gibbs}
    Let $\epsilon,\delta > 0$, and let $k = \mathcal{O}(1)$.
    Let $H = \ii\sum_P \lambda_P P$ be a $k$-local Hamiltonian with $\norm{\lambda}_\infty \leq 1$ and bounded-degree interactions, i.e., every qubit interacts with at most $d$ nonzero terms.
    Let $\beta > 0$ satisfy
    \begin{equation}
        \beta \leq \frac{1}{1000e^6(2kd+1)^8}.
    \end{equation}
    Then, there exists an algorithm that finds estimates $\widehat{\lambda}$ such that $\norm{\widehat{\lambda} - \lambda}_\infty \leq \epsilon$ with probability at least $1-\delta$ using $\mathcal{O}(\log(n/\delta)/(\beta\epsilon)^2)$ copies of the Gibbs state $\rho_\beta = e^{-\beta H}/\tr(e^{-\beta H})$.
    The classical runtime of this algorithm is $\mathcal{O}(n^k\poly(d)\log(n/\delta)/(\beta\epsilon)^2)$.
\end{theorem}

Before proving this theorem, we recall a result from~\cite{haah2024learning}, which bounds the $\infty\to\infty$ norm of the Jacobian for high-temperature Gibbs states.

\begin{lemma}[Lemma 4.3 in~\cite{haah2024learning}]
    \label{lem:hkt-jacobian-bound}
    Let $J_{P,Q}(x) = \partial_{x_Q}\mathcal{F}_P(x)$ be the Jacobian, where $P,Q$ range over $k$-local Paulis such that the graph with these Paulis as vertices, and edges between $P,Q$ when $S_P \cap S_Q \neq \emptyset$, has degree at most $D$.
    Then,
    \begin{equation}
        \norm{J(x) - I}_{\infty\to\infty} \leq c_D \beta,
    \end{equation}
    where $c_D = 50e^6(D+1)^8$.
\end{lemma}

We use this to instantiate the hypothesis on the Jacobian bound in \Cref{thm:ham}.
Thus, with this, we can prove \Cref{thm:ham-gibbs}.

\begin{proof}[Proof of {\Cref{thm:ham-gibbs}}]
    As with \Cref{thm:ham-time}, we first need to obtain estimates $\widehat{E}_P(\lambda)$ such that
    \begin{equation}
        |\widehat{E}_P(\lambda) - E_P(\lambda)| \leq \frac{\beta \epsilon}{40}
    \end{equation}
    for all $k$-local Paulis $P$.
    We can do so with probability at least $1-\delta$ using $\mathcal{O}(4^k\log(n/\delta)/(\beta\epsilon)^2)$ copies of the Gibbs state $\rho_\beta$ via classical shadows~\cite{HKP20}.
    Moreover, for $k=\mathcal{O}(1)$, this has a classical runtime of $\mathcal{O}(n^k\log(n/\delta)/(\beta\epsilon)^2)$.
    For our choice of $\beta$, this corresponds to $\mathcal{O}(d^{16}\log(n/\delta)/\epsilon^2)$ copies.

    Also, we can obtain estimates $\widehat{E}_P(x)$ such that
    \begin{equation}
        |\widehat{E}_P(x) - E_P(x)| \leq \frac{\beta \epsilon}{40}
    \end{equation}
    by computing a truncated cluster expansion.
    The proof of Theorem 4.6 in~\cite{haah2024learning} shows that this takes classical runtime $\mathcal{O}(n^k\poly(d, \log(1/(\beta\epsilon))/\epsilon)$.
    This applies to our setting because we maintain the invariant in our algorithm that $|x_P^{(j)}| \leq 2|\lambda_P|$.
    This means that $\deg(x) \leq \deg(\lambda) \leq d$, which fits in the setting of~\cite{haah2024learning}.
    This gives us an approximation of $\mathcal{F}(x)$ up to $\epsilon/20$ error.
    It remains to argue that the Jacobian bound holds.
    Consider the statement we want to prove:
    \begin{equation}
        \sum_{Q \in S_H} |J_{P,Q}(x) - \delta_{P,Q}| \leq c_d \beta,
    \end{equation}
    for $\deg(x) \leq d$.
    This closely resembles \Cref{lem:hkt-jacobian-bound}.
    Notice that $Q \in S_H$, so since $H$ has bounded degree, the graph with $Q \in S_H$ as its vertices has degree at most $d$.
    However, the key difference is that $P$ can be any $k$-local Pauli and does not necessarily lie in this graph.
    Nevertheless, the graph with vertices consisting of $S_H \cup \{P\}$ still has degree at most $D = kd$.
    Thus, we can apply \Cref{lem:hkt-jacobian-bound} with $D = kd$ to obtain our claim with $c_d = 50e^6(kd+1)^8$.
    This completes the proof.
\end{proof}

\section{Lower bound}
\label{sec:lower-bound}

In this section, we prove our lower bound on the total time evolution required to learn a Lindbladian.
This lower bound holds against learning algorithms which can query the dynamics adaptively, have access to ancillary systems, and use entangled operations/measurements.
Meanwhile, our algorithm is non-adaptive, ancilla-free, and uses only product state preparation and measurements.
Nevertheless, the lower bound establishes that our algorithm is optimal up to logarithmic factors in all parameters.

\begin{theorem}
    \label{thm:lower-bound}
    Let $m \geq 10^8$ be an integer.
    Let $g > 0$, $0 < \epsilon \leq g/512$, $n = m^3 + 1$, $d = 6m + 1$.
    Then, there exists a finite family $\mathcal{C}_{m,g,\epsilon}$ of two-local Lindbladians on $n$ qubits such that every $\mathcal{L} \in \mathcal{C}_{m,g,\epsilon}$ satisfies
    \begin{equation}
        \frac{g}{2} \leq \lonorm{\mathcal{L}} \leq g, \quad \frac{m}{2} \leq \deg_{\epsilon/25600}(\lind) \leq \deg(\lind) \leq d.
    \end{equation}
    Moreover, any learning algorithm which learns the coefficients of $\lind \in \mathcal{C}_{m,g,\epsilon}$ to $\epsilon$-error in $\infty$-norm with probability at least $0.99$ must use a total evolution time of
    \begin{equation}
        \tet = \Omega\left(\frac{gm^2}{\epsilon^2 \log m}\right) = \Omega\left(\frac{gd^2}{\epsilon^2 \log d}\right).
    \end{equation}
\end{theorem}

The remainder of this section is dedicated to proving this theorem.
Throughout the following subsections, we consider the parameters in \Cref{thm:lower-bound}.

\subsection{Hard instance}

In this section, we describe the family of Lindbladians $\mathcal{C}_{m,g,\epsilon}$.
The idea is to reduce learning a Lindbladian in this family to distinguishing between two cases, where the Lindbladian coefficients are drawn from one of two distributions.
First, we define these two distributions based on~\cite{wu2019chebyshev}, then we define the family of Lindbladians.

The two distributions are constructed to have matching higher moments (for orders $2,\dots, 2q+1$ for some integer $q \geq 1$).
The idea is that if the higher moments of two distributions agree, then any low-degree distinguisher can only use the information about the mean.
Our lower bound is not restricted to only low-degree algorithms, but we effectively show that algorithms with access to the real-time dynamics also can only use information about the mean.
Moreover, the signal coming from the mean is small.

In the following lemma, we first state a result from~\cite{wu2019chebyshev}.
This can be found by taking $L = q, \lambda = (q+1)^2, \nu = 0$ in the statement following Lemma 3 in~\cite{wu2019chebyshev}.

\begin{lemma}[\cite{wu2019chebyshev}]
    \label{lem:wu-yang}
    Let $q \geq 1$ be an integer.
    There exist two finitely supported random variables $Y_0, Y_1 \in \{0\} \cup [1,(q+1)^2]$ such that $\E[Y_0] = \E[Y_1] = 1$ and $\E[Y_0^r] = \E[Y_1^r]$ for all $r \in [q]$.
    Moreover, 
    \begin{equation}
        \Pr[Y_0 > 0] - \Pr[Y_1 > 0] = \left(1 + \frac{1}{q+1}\right)^2\left(\frac{q}{q+2}\right)^q > \frac{1}{e^2},
    \end{equation}
    where the last inequality follows because $(1+2/q)^q < e^2$.
\end{lemma}

We use this to construct the two distributions used in our proof.

\begin{corollary}
    \label{lem:matching-moments}
    Let $q \geq 1$ be an integer.
    There exist distributions $\mathcal{D}_+, \mathcal{D}_-$ supported on $[-1, -1/(q+1)] \cup [1/(q+1), 1]$ and a number $W \in [2(q+1), 2e^2(q+1))$ such that $\mathcal{D}_-$ is a reflection of $\mathcal{D}_+$ and for $X \sim \mathcal{D}_+$, then $\E[X] = 1/W$ and $\E[X^{2r+1}] = 0$ for all $r \in [q]$.
    Moreover, $\E[|X|] \leq \sqrt{2}/(q+1)$ and $\E[X^2] \leq 2/(q+1)^2$.
    In particular, $\E_{X \sim \mathcal{D}_{\pm}}[X] = \pm 1/W$ and the moments of order $2,\dots, 2q+1$ agree.
\end{corollary}

\begin{proof}
    We define a random variable by modifying those from \Cref{lem:wu-yang}.
    Choose $B$ uniformly from $\{0,1\}$, and sample $Y$ from the distribution of $Y_B$ from \Cref{lem:wu-yang}.
    If $Y = 0$, then choose $X$ uniformly from $\{-1/(q+1), 1/(q+1)\}$.
    Otherwise, conditional on $B = b$ and $Y = y > 0$, choose $X$ according to the probability distribution
    \begin{equation}
        \Pr\left[X = \pm \frac{\sqrt{y}}{q+1} \,\Big|\, B = b, Y = y\right] = \frac{1 \pm (-1)^b/\sqrt{y}}{2}.
    \end{equation}
    Note that these are valid probabilities because $y \geq 1$.
    We take $\mathcal{D}_+$ to be the distribution of $X$ and $\mathcal{D}_-$ to be the distribution of $-X$.
    We show that it satisfies the claim properties.

    Note that $|X| \in [1/(q+1), 1]$.
    With this, we can compute the conditional expectations for the case of $y > 0$
    \begin{equation}
        \E[X \mid B = b, Y = y] = \frac{1 + (-1)^b}{2(q+1)} + \frac{-1 + (-1)^b}{2(q+1)} = \frac{(-1)^b}{q+1}.
    \end{equation}
    Also, for an integer $r \geq 1$, then
    \begin{equation}
        \E[X^{2r+1} \mid B = b, Y = y] = \frac{y^{r+1/2} + y^{r}(-1)^b}{2(q+1)^{2r+1}} + \frac{-y^{r+1/2} + y^r(-1)^b}{2(q+1)^{2r+1}} = \frac{(-1)^b y^r}{(q+1)^{2r+1}}.
    \end{equation}
    Then, we can compute the expectations via the law of total expectation
    \begin{align}
        \E[X] &= \Pr[B = 0]\E[X \mid B = 0] + \Pr[B = 1]\E[X \mid B = 1]\\
        &= \frac{1}{2}\Pr[Y = 0 \mid B = 0]\E[X \mid B = 0, Y = 0] + \frac{1}{2}\Pr[Y > 0 \mid B = 0]\E[X \mid B = 0, Y > 0]\\
        &+ \frac{1}{2}\Pr[Y = 0 \mid B = 1]\E[X \mid B = 1, Y = 0] + \frac{1}{2}\Pr[Y > 0 \mid B = 1]\E[X \mid B = 1, Y > 0]\\
        &= \frac{1}{2}\Pr[Y_0 > 0]\cdot \frac{1}{q+1} + \frac{1}{2}\Pr[Y_1 > 0]\cdot\frac{-1}{q+1}\\
        &= \frac{1}{2(q+1)}(\Pr[Y_0 > 0] - \Pr[Y_1 > 0]).
    \end{align}
    The claim then follows by \Cref{lem:wu-yang}.
    Similarly, for $r \in [q]$, then
    \begin{equation}
        \E[X^{2r+1}] = \frac{1}{2(q+1)^{2r+1}}(\E[Y_0^r] - \E[Y_1^r]) = 0,
    \end{equation}
    where the last equality follows by \Cref{lem:wu-yang}.
    For $X^2$, notice that
    \begin{equation}
        X^2 = \begin{cases}
            \frac{1}{(q+1)^2} & Y = 0\\
            \frac{Y}{(q+1)^2} & Y > 0
        \end{cases}.
    \end{equation}
    We can simplify this to write $X^2 = (Y + \indicator\{Y = 0\})/(q+1)^2$.
    Then, note that
    \begin{equation}
        \E[Y] = \frac{1}{2}\E[Y_0] + \frac{1}{2}\E[Y_1] = 1
    \end{equation}
    by \Cref{lem:wu-yang}.
    Thus, we have
    \begin{equation}
        \E[X^2] = \frac{1}{(q+1)^2}\left(\E[Y] + \E[\indicator\{Y = 0\}]\right) = \frac{1 + \Pr[Y = 0]}{(q+1)^2} \leq \frac{2}{(q+1)^2}.
    \end{equation}
    Finally, by Cauchy-Schwarz,
    \begin{equation}
        \E[|X|] \leq \sqrt{\E[X^2]} \leq \frac{\sqrt{2}}{q+1}.
    \end{equation}
\end{proof}

With this definition of our two distributions, we can define our family of Lindbladians.
Recall that a Lindbladian takes the form
\begin{equation}
    \lind(\rho) = \frac{1}{2}\sum_{P \neq I} \alpha_P [P, \rho] + \sum_{P_1, P_2 \neq I} D_{P_1,P_2}\left(P_1 \rho P_2 - \frac{1}{2}\{P_2 P_1, \rho\}\right).
\end{equation}
Consider the parameters in the statement of \Cref{thm:lower-bound}.
Let $W$ be the constant from \Cref{lem:matching-moments} applied to $q = \lceil 10 \log m\rceil$.
Define the domain $\mathcal{U} = \{0\} \cup \frac{\epsilon W}{m} \supp(\mathcal{D}_+) \cup \frac{\epsilon W}{m}\supp(\mathcal{D}_-)$, and let $u \in \mathcal{U}^{n-1}$ be a parameter vector.
Then, we consider Lindbladians $\mathcal{L}_u$ with parameters defined as
\begin{equation}
    \label{eq:hard-instance}
    \alpha_{Z_0} = -2\ii\sum_{i=1}^{n-1}u_i,\quad D_{Z_i, Z_i} = \frac{g}{2}, \quad D_{Z_0Z_i, Z_0Z_i} = \frac{2u_i^2}{g},\quad D_{Z_i, Z_0Z_i} = -\ii u_i, \quad D_{Z_0Z_i, Z_i} = \ii u_i
\end{equation}
for $i \in [n-1]$.
One can easily check that this defines a valid physical Lindbladian.
The family of Lindbladians $\mathcal{C}_{m,g,\epsilon}$ is then defined as Lindbladians with the parameters above that also satisfy the following constraints:
\begin{equation}
    \label{eq:K}
    \frac{m}{2} \leq |\{i \in [n-1] : u_i \neq 0\}| \leq 2m,
\end{equation}
\begin{equation}
    \label{eq:V}
    \sum_{i=1}^{n-1} |u_i| \leq 32\epsilon,
\end{equation}
\begin{equation}
    \label{eq:threshold}
    \sum_{i=1}^{n-1} u_i \geq \frac{3\epsilon}{4} \quad\mathrm{or}\quad \sum_{i=1}^{n-1} u_i \leq -\frac{3\epsilon}{4}.
\end{equation}
In \Cref{eq:threshold}, if the Lindbladian satisfies the first condition, call this family $\mathcal{C}_{m,g,\epsilon}^+$; otherwise, call the family $\mathcal{C}_{m,g,\epsilon}^-$.
Thus, $\mathcal{C}_{m,g,\epsilon} = \mathcal{C}_{m,g,\epsilon}^+ \cup \mathcal{C}_{m,g,\epsilon}^-$.
Such Lindbladians are two-local, and we show that they satisfy the conditions in \Cref{thm:lower-bound}.

\begin{lemma}
    Every Lindbladian $\lind_u \in \mathcal{C}_{m,g,\epsilon}$ satisfies $g/2 \leq \lonorm{\mathcal{L}} \leq g$ and
    \begin{equation}
        \frac{m}{2} \leq \deg_{\epsilon/25600}(\lind) \leq \deg(\lind) \leq d.
    \end{equation}
\end{lemma}

\begin{proof}
    First, we show the bounds on the $B_1$-norm.
    The lower bound of $g/2$ is clear because of the coefficient $D_{Z_i, Z_i} = g/2$ for all $i \in [n-1]$.
    For the upper bound, we consider the interaction strengths incident to each qubit.
    For qubit $0$, then the interaction strength is
    \begin{equation}
        2\left|\sum_{i=1}^{n-1} u_i\right| + 2\sum_{i=1}^{n-1}|u_i| + \frac{2}{g}\sum_{i=1}^{n-1}u_i^2 \leq 128\epsilon + \frac{2\epsilon W}{gm}\cdot 32\epsilon \leq \frac{g}{4} + \frac{Wg}{4096m} \leq \frac{g}{4} + \frac{e^2(q+1)g}{2048m} < g.
    \end{equation}
    In the first inequality, we use \Cref{eq:V} and the definition of $\mathcal{U}$ so that $|u_i| \leq \epsilon W/m$.
    In the second inequality, we use that $\epsilon < g/512$.
    In the third inequality, we use that $W < 2e^2(q+1)$ by \Cref{lem:matching-moments}.
    In the last inequality, we use that $q + 1 \leq m$ for $m \geq 10^8$ because $q = \lceil 10\log m \rceil$.
    Similarly, for qubit $i$, for any $i \in [n-1]$, the interaction strength incident to that qubit is
    \begin{equation}
        \frac{g}{2} + \frac{2u_i^2}{g} + 2|u_i| \leq \frac{g}{2} + \frac{2\epsilon^2 W^2}{2m^2 g} + \frac{2\epsilon W}{m} \leq \frac{g}{2} + \frac{2gW^2}{512^2 m^2} + \frac{gW}{256m} \leq \frac{g}{2} + \frac{8e^4 g}{512^2} + \frac{e^2g}{128} < g.
    \end{equation}
    In the first inequality, we use that $|u_i| \leq \epsilon W/m$.
    In the second inequality, we use that $\epsilon < g/512$.
    In the third inequality, we use that $W < 2e^2(q+1)$ by \Cref{lem:matching-moments}, and $q + 1 \leq m$ for our parameter choices.
    Thus, we can conclude that the total interaction strength is bounded by $g$.

    Now for the degree bound, we first consider the upper bound.
    The number of terms incident to qubit $0$ is $3|\{i \in [n-1] : u_i \neq 0\}| + 1 \leq 6m + 1 = d$.
    This is because we pick up one term from each of the dissipative terms except for the $Z_i, Z_i$ term, and we also pick up the coherent term on $Z_0$.
    Meanwhile, the number of terms incident to qubit $i$ is at most $4$.
    Hence, it is clear that $\deg(\mathcal{L}) \leq d$.

    It remains to show the lower bound on the approximate degree.
    Notice that, for $u_i \neq 0$, we can lower bound $|u_i|$ by
    \begin{equation}
        |u_i| \geq \frac{\epsilon W}{m(q+1)} \geq \frac{2\epsilon}{m},
    \end{equation}
    where we use the definition of $\mathcal{U}$ and $W \geq 2(q+1)$ by \Cref{lem:matching-moments}.
    There are $2|\{i \in [n-1] : u_i \neq 0\}| \geq m$ terms incident to qubit $0$, corresponding to the terms on $Z_i, Z_0Z_i$ and $Z_0Z_{-i}, Z_i$.
    Consider a superoperator $\Lbig$ which is the same as $\lind$ but sets $s \geq 1$ of the coefficients corresponding to these terms to zero.
    Then, $\lonorm{\lind -\Lbig} \geq 2s\epsilon/m$.
    Thus, in order to achieve $\lonorm{\lind - \Lbig} < \epsilon/25600$, we must have $s < m/51200$.
    Hence, any such $\Lbig$ must set at most $s < m/51200$ of these coefficients to zero.
    There are still $2|\{i \in [n-1] : u_i \neq 0\}| - s$ terms incident to qubit $0$ which are nonzero, and notice
    \begin{equation}
        2|\{i \in [n-1] : u_i \neq 0\}| - s > m - \frac{1}{51200}m > \frac{m}{2}.
    \end{equation}
    We have shown this for any $\Lbig$ satisfying $\lonorm{\lind - \Lbig} < \epsilon/25600$, so it follows that $\deg_{\epsilon/25600}(\lind) > m/2$.
    This completes the proof.
\end{proof}

In particular, we define the parameters as
\begin{equation}
    \label{eq:ui}
    u_i \triangleq \frac{\epsilon W}{2m} B_i X_i
\end{equation}
for $i \in [n-1]$, where $B_i$ are i.i.d. sampled from $B_i \sim \mathrm{Bernoulli}(1/m^2)$ and $X_i$ are i.i.d. sampled from $X_i \sim \mathcal{D}_{\pm}$.
Moreover, $B_i$ and $X_i$ are independent.
Denote this (product) distribution for the whole vector as $u \sim \Pi_{\pm}$.
The point is to reduce learning $\mathcal{L}$ to distinguishing between $u \sim \Pi_+$ or $u \sim \Pi_-$.
Namely, we show that these two cases correspond to distinguishing the sign of the coherent coefficient, i.e., distinguishing between $\mathcal{C}_{m,g,\epsilon}^+$ and $\mathcal{C}_{m,g,\epsilon}^-$.

\begin{lemma}[Learning implies distinguishing]
    \label{lem:learn-to-distinguish}
    Any algorithm which learns the coefficients of $\mathcal{L} \in \mathcal{C}_{m,g,\epsilon}$ to $\epsilon$ error in $\infty$-norm with probability at least $0.99$ distinguishes between the cases of $\Pi_+$ and $\Pi_-$ with probability at least $0.98$.
\end{lemma}

\begin{proof}
    Suppose a learning algorithm outputs estimates $\widehat{\alpha}, \widehat{D}$ of the coefficients.
    Then, we define a distinguishing algorithm which computes $\widehat{b} \triangleq -\mathrm{Im}\,\widehat{\alpha}_{Z_0}/2$ and outputs $+$ when $\widehat{b} \geq 0$ and $-$ otherwise.
    By the guarantees of the learning algorithm, we know that
    \begin{equation}
        \left|\widehat{b} - \sum_{i=1}^{n-1} u_i \right| \leq \frac{\epsilon}{2}
    \end{equation}
    with probability at least $0.99$.
    If $\mathcal{L}_u \in \mathcal{C}_{m,g,\epsilon}^+$, then recall that, by definition, this implies $\sum_{i=1}^{n-1}u_i \geq 3\epsilon/4$.
    Then, in this case, if the learner succeeds,
    \begin{equation}
        \widehat{b} \geq \frac{3\epsilon}{4} - \frac{\epsilon}{2} = \frac{\epsilon}{4} > 0.
    \end{equation}
    Similarly, if $\lind_u \in \mathcal{C}_{m,g,\epsilon}^-$, then this implies that $\sum_{i=1}^{n-1} u_i \leq -3\epsilon/4$ so that $\widehat{b} = -\epsilon/4 < 0$.
    Thus, the sign of $\widehat{b}$ correctly distinguishes the two cases whenever the learning algorithm succeeds, i.e.,
    \begin{equation}
        \Pr_{\mathcal{A}}[\text{distinguish correctly} \mid \lind_u \in \mathcal{C}_{m,g,\epsilon}] \geq 0.99,
    \end{equation}
    where the probability is over the randomness in the learning algorithm.
    Then, overall,
    \begin{equation}
        \label{eq:distinguish}
        \Pr_{u\sim \Pi_{\pm}, \mathcal{A}}[\text{distinguish correctly}] \geq 0.99 \Pr_{u \sim \Pi_{\pm}}[\lind_u \in \mathcal{C}_{m,g,\epsilon}^{\pm}].
    \end{equation}
    It remains to show that this is at least $0.98$.
    In particular, we show that $\Pr_{u \sim \Pi_{\pm}}[\lind_u \notin \mathcal{C}_{m,g,\epsilon}^{\pm}]$ is very small.
    We bound the probability that each of the conditions in \Cref{eq:V,eq:K,eq:threshold} are not satisfied.
    First, for the condition in \Cref{eq:K}, notice that $|\{i \in [n-1] : u_i \neq 0\}| \sim \mathrm{Binomial}(n-1, 1/m^2)$ because of the Bernoulli random variables in the definition of $u_i$ and $0 \notin \supp(\mathcal{D}_{\pm})$.
    Moreover, for binomial random variables $X \sim \mathrm{Binomial}(N, p)$ and $s \in \mathbb{R}$,
    \begin{equation}
        \E[e^{s X}] = \left(1 + p(e^s - 1)\right)^N \leq \exp(Np(e^s-1)).
    \end{equation}
    Using this for $s = -\log 2$ and Markov's inequality,
    \begin{equation}
        \label{eq:bound1}
        \Pr_{u \sim \Pi_{\pm}}\left[|\{i \in [n-1] : u_i \neq 0\}| < \frac{m}{2}\right] \leq \exp\left(-\frac{m}{2}(1-\log 2)\right) \leq e^{-m/8}.
    \end{equation}
    Here, we also use that $n - 1 = m^3$ so that $(n-1)/m^2 = m$.
    Similarly, for $s = \log 2$ and Markov's inequality,
    \begin{equation}
        \label{eq:bound2}
        \Pr_{u \sim \Pi_{\pm}}\left[|\{i \in [n-1] : u_i \neq 0\}| > 2m\right] \leq \exp(-m(2\log 2 - 1)) \leq e^{-m/3}.
    \end{equation}
    Now, for the condition in \Cref{eq:V}, by independence, we know that 
    \begin{equation}
        \E_{u \sim \Pi_{\pm}}\left[\sum_{i=1}^{n-1} |u_i|\right] = (n-1)\frac{\epsilon W}{m} \frac{1}{m^2}\E_{X_i\sim \mathcal{D}_{\pm}}[|X_i|] \leq (n-1)\frac{\epsilon W}{m^3}\frac{\sqrt{2}}{q+1} = \frac{\sqrt{2}\epsilon W}{q+1} < 21\epsilon,
    \end{equation}
    \begin{equation}
        \label{eq:exp-u-squared}
        \E_{u \sim \Pi_{\pm}}[u_i^2] = \frac{\epsilon^2 W^2}{4m^2} \frac{1}{m^2} \E_{X_i \sim \mathcal{D}_{\pm}}[X_i^2] \leq \frac{\epsilon^2 W^2}{4m^4} \frac{2}{(q+1)^2} = \frac{\epsilon^2 W^2}{2m^4(q+1)^2}.
    \end{equation}
    In both expressions, we use \Cref{lem:matching-moments}, the definition of $u$, and $n - 1 = m^3$.
    Moreover, this implies that
    \begin{equation}
        \sum_{i=1}^{n-1} \mathrm{Var}(|u_i|) \leq \frac{2\epsilon^2 W^2}{m(q+1)^2} \leq \frac{8e^4\epsilon^2}{m},
    \end{equation}
    where we use that $W \leq 2e^2(q+1)$ from \Cref{lem:matching-moments}.
    Also, we know that $||u_i| - \E_{u \sim \Pi_{\pm}}[|u_i|]| \leq \epsilon W/m$.
    Then, by Bernstein's inequality,
    \begin{align}
        \Pr_{u\sim\Pi_{\pm}}\left[\sum_{i=1}^{n-1}|u_i| > 32\epsilon\right] &\leq \exp\left(-\frac{(11\epsilon)^2}{2(16e^4 \epsilon^2/m) + 22\epsilon^2 W/(3m))}\right)\\
        &= \exp\left(-\frac{363m}{48e^4 + 22W}\right)\\
        &\leq \exp\left(-\frac{363m}{48e^4 + 44e^2(q+1)}\right).\label{eq:bound3}
    \end{align}
    In the last inequality, we use $W < 2e^2(q+1)$.
    Finally, for the condition in \Cref{eq:threshold}, we know that
    \begin{equation}
        \E_{u \sim \Pi_{\pm}}\left[\sum_{i=1}^{n-1} u_i\right] = (n-1)\frac{\epsilon W}{2m} \frac{1}{m^2} \cdot \pm \frac{1}{W} = \pm \frac{\epsilon}{2},
    \end{equation}
    where we use that $n - 1 = m^3$.
    By \Cref{eq:exp-u-squared}, we also have
    \begin{equation}
        \sum_{i=1}^{n-1} \mathrm{Var}(u_i) \leq \frac{2\epsilon^2W^2}{m(q+1)^2} \leq \frac{8e^4\epsilon^2}{m}.
    \end{equation}
    Moreover, $|u_i - \E_{u \sim \Pi_{\pm}}[u_i]| \leq 2\epsilon W/m$.
    Then, by Bernstein's inequality again,
    \begin{align}
        \Pr_{u \sim \Pi_-}\left[\sum_{i=1}^{n-1} u_i > -\frac{3\epsilon}{4}\right] = \Pr_{u \sim \Pi_+}\left[\sum_{i=1}^{n-1} u_i < \frac{3\epsilon}{4}\right] &\leq \exp\left(-\frac{(\epsilon/4)^2}{2(16e^4\epsilon^2/m + \epsilon^2 W/(3m))}\right)\\
        &= \exp\left(-\frac{3m}{768e^4 + 16W}\right)\\
        &\leq \exp\left(-\frac{3m}{768e^4 + 32e^2(q+1)}\right).\label{eq:bound4}
    \end{align}
    In the last inequality, we use $W < 2e^2(q+1)$.
    Putting everything together (\Cref{eq:bound1,eq:bound2,eq:bound3,eq:bound4}), we have
    \begin{equation}
        \Pr_{u \sim \Pi_{\pm}}[\lind_u \notin \mathcal{C}_{m,g,\epsilon}^{\pm}] \leq e^{-m/8} + e^{-m/3} + \exp\left(-\frac{363m}{48e^4 + 44e^2(q+1)}\right) + \exp\left(-\frac{3m}{768e^4 + 32e^2(q+1)}\right) < 10^{-6},
    \end{equation}
    where the last inequality uses our choice of $m \geq 10^8$ and $q = \lceil 10 \log m \rceil$.
    Plugging this into \Cref{eq:distinguish} completes the proof.
\end{proof}

It remains to show that distinguishing these two cases requires a large total time evolution.

\subsection{Simulating queries to Lindbladian evolution via coherent states}

We reduce the distinguishing task from the previous section to distinguishing between two states.
Then, the lower bound can be completed by bounding the trace distance between these two states.
To this end, we define a program state, which can be used to simulate a distinguishing algorithm's queries to the Lindbladian time evolution.
We consider bosonic coherent states, which, when written in the Fock basis, take the form
\begin{equation}
    \label{eq:coherent-state}
    \ket{z} \triangleq e^{-|z|^2/2} \sum_{\ell=0}^\infty \frac{z^\ell}{\sqrt{\ell!}}\ket{\ell}
\end{equation}
for $z \in \mathbb{C}$.
The program states are then defined as follows.

\begin{definition}[Program state]
    \label{def:program-state}
    Let $u \in \mathbb{R}^{n-1}$, and let $T > 0$.
    Let $\beta_i \triangleq \ii u_i\sqrt{2/g}$.
    Then, the \emph{program state} for time $T$ is defined as
    \begin{equation}
        \Psi_u(T) \triangleq \bigotimes_{i=1}^{n-1}\ket{\beta_i\sqrt{T}}\bra{\beta_i\sqrt{T}}.
    \end{equation}
\end{definition}

We prove that one can simulate any distinguishing algorithm's queries to the Lindbladian dynamics with a total evolution time of $T$ by using one copy of the program state.
This tells us that interacting with the program state is at least as powerful as interacting with the time dynamics.
Thus, if the program state cannot be used to distinguish between $\Pi_{\pm}$, then neither can the dynamics.

\begin{proposition}[Simulating Lindbladian queries with program states]
    \label{prop:simulation}
    Let $\eta > 0$.
    Let $u \in \mathcal{U}^{n-1}$.
    Consider a distinguishing algorithm, which has access to a Lindbladian $\mathcal{L}_u$ of the form of \Cref{eq:hard-instance} and uses a total evolution time of $T$.
    Let $\omega_u$ be a diagonal density matrix representing the output distribution of this distinguishing algorithm.
    Then, there exists a quantum channel $\Phi$, which is independent of $u$, such that
    \begin{equation}
        \dtr(\Phi(\Psi_u(T)), \omega_u) \leq \eta.
    \end{equation}
\end{proposition}

\begin{corollary}
    \label{coro:simulation}
    For any two distributions $\pi_0, \pi_1$ supported on $\mathcal{U}^{n-1}$, then
    \begin{equation}
        \dtr\left(\E_{u \sim \pi_0}[\omega_u], \E_{u\sim \pi_1}[\omega_u]\right) \leq \dtr\left(\E_{u\sim \pi_0}[\Psi_u(T)], \E_{u\sim\pi_1}[\Psi_u(T)]\right).
    \end{equation}
\end{corollary}

\begin{proof}
    This follows by contractivity of the trace distance:
    \begin{align}
        \dtr\left(\E_{u \sim \pi_0}[\omega_u], \E_{u\sim \pi_1}[\omega_u]\right) &\leq 2\eta + \dtr\left(\Phi\left(\E_{u \sim\pi_0}[\Psi_u(T)]\right), \Phi\left(\E_{u\sim\pi_1}[\Psi_u(T)]\right)\right)\\
        &\leq 2\eta + \dtr\left(\E_{u\sim\pi_0}[\Psi_u(T)], \E_{u\sim\pi_1}[\Psi_u(T)]\right).
    \end{align}
    In the first inequality, we use \Cref{prop:simulation}, and in the second inequality, we use contractivity.
    Because this holds for all $\eta > 0$, then, the claim follows.
\end{proof}

Now, we prove \Cref{prop:simulation}.
Similar ideas for simulating Hamiltonian or Lindbladian evolution from program states has been explored in prior work~\cite{patel2023wave1,patel2023wave,combes2017slh}. 
However, to our knowledge, our application of this to lower bounds for learning is novel, and prior work did not allow one to simulate \emph{adaptive} queries to the time evolution.
In particular, \cite{patel2023wave1,patel2023wave} show how to simulate a single query to a Lindbladian time evolution given access to several copies of a qubit quantum state.
We discuss the relationship between the states considered in \cite{patel2023wave1,patel2023wave} and \Cref{def:program-state} in the technical overview.
Similarly to our work, \cite{combes2017slh} uses a single copy of a coherent state to simulate one query to the Lindbladian time evolution.
However, to describe the joint operation applied to the qubit quantum system and the ancillary program state, \cite{combes2017slh} relies on heavy machinery from stochastic calculus.
Moreover, \cite{combes2017slh} does not consider adaptive queries to the evolution.
Thus, for completeness, we prove \Cref{prop:simulation} in full here rather than adapt existing results.

First, we show the following important property of coherent states.
One can view this as a property of rotated Poisson distributions, as the amplitudes of coherent states follow a Poisson distribution.

\begin{lemma}[Splitting coherent states]
    \label{lem:splitting}
    Let $\beta$ be purely imaginary.
    Let $R > 0$, and let $t \in [0,R]$.
    Then, there exists a unitary $B_{t/R}$ that is independent of $\beta$ such that
    \begin{equation}
        B_{t/R}(\ket{\beta \sqrt{R}} \ket{0}) = \ket{\beta \sqrt{t}} \ket{\beta\sqrt{R-t}}.
    \end{equation}
\end{lemma}

\begin{proof}
    We consider two-mode bosonic states $\ket{\ell}\ket{r}$ for $\ell, r \in \mathbb{Z}_{\geq 0}$.
    Let $e_1, e_2$ be the standard basis of $\mathbb{C}^2$.
    Here, $e_1$ represents that one particle is present in the first mode, and similarly for $e_2$.
    Then, the subspace of two-mode bosonic states with total occupation number $s$ (i.e., $\ell + r = s$) $\mathcal{H}_s$ can be identified with the symmetric subspace $\mathrm{Sym}_s(\mathbb{C}^2)$ of $(\mathbb{C}^2)^{\otimes s}$.
    Then, written in the $\{e_1, e_2\}$ basis, we define the unitary
    \begin{equation}
        V_{t/R} \triangleq \begin{pmatrix}
            \sqrt{t/R} & -\sqrt{1-t/R}\\
            \sqrt{1-t/R} & \sqrt{t/R}
        \end{pmatrix}.
    \end{equation}
    Then, we define the unitary $B_{t/R}$ on $\mathcal{H}_s$ as
    \begin{equation}
        B_{t/R}\Big|_{\mathcal{H}_s} \triangleq V_{t/R}^{\otimes s}\Big|_{\mathrm{Sym}_s(\mathbb{C}^2)}.
    \end{equation}
    One can extend this definition to all two-mode states via
    \begin{equation}
        B_{t/R} = \bigoplus_{s=0}^\infty \left(V_{t/R}^{\otimes s}\Big|_{\mathrm{Sym}_s(\mathbb{C}^2)}\right).
    \end{equation}
    Concretely, for states $\ket{\ell}\ket{0}$, then we have
    \begin{align}
        B_{t/R}(\ket{\ell}\ket{0}) &= V_{t/R}^{\otimes \ell} e_1^{\otimes \ell}\\
        &= \left(\sqrt{\frac{t}{R}}e_1 + \sqrt{1 - \frac{t}{R}}e_2\right)^{\otimes \ell}\\
        &= \sum_{a=0}^\ell \sqrt{\binom{\ell}{a}} \left(\frac{t}{R}\right)^{a/2} \left(1- \frac{t}{R}\right)^{(\ell-a)/2}\ket{a} \ket{\ell - a}.\\
        &= \sum_{a + b = \ell} \sqrt{\frac{\ell!}{a!b!}} \left(\frac{t}{R}\right)^{a/2} \left(1- \frac{t}{R}\right)^{b/2} \ket{a}\ket{b}.
    \end{align}
    where the third line follows by the binomial formula.
    The last line follows by taking $b = \ell - a$.
    We can apply this to our coherent state $\ket{\beta\sqrt{R}}\ket{0}$ term by term by \Cref{eq:coherent-state} to obtain
    \begin{align}
        B_{t/R}(\ket{\beta\sqrt{R}}\ket{0}) &= e^{-|\beta|^2R/2} \sum_{\ell=0}^\infty \sum_{a + b = \ell} \frac{(\beta\sqrt{R})^\ell}{\sqrt{a!b!}} \left(\frac{t}{R}\right)^{a/2} \left(1- \frac{t}{R}\right)^{b/2}\ket{a}\ket{b}\\
        &= e^{-|\beta|^2R/2} \sum_{a,b=0}^\infty \frac{(\beta\sqrt{R})^{a+b}}{\sqrt{a!b!}} \left(\frac{t}{R}\right)^{a/2} \left(1- \frac{t}{R}\right)^{b/2}\ket{a}\ket{b}\\
        &= e^{-|\beta|^2R/2} \sum_{a,b=0}^\infty \frac{1}{\sqrt{a!b!}} (\beta\sqrt{t})^a (\beta\sqrt{R-t})^b \ket{a}\ket{b}\\
        &= \ket{\beta\sqrt{t}}\ket{\beta\sqrt{R-t}}.
    \end{align}
    In the second line, we notice that each pair $(a,b)$ appears for exactly one $\ell$.
    In the last line, we use the definition of coherent states in \Cref{eq:coherent-state}.
\end{proof}

For simplicity, we first consider simulating a single non-adaptive query to the Lindbladian dynamics and then generalize this to adaptive queries.
To do so, write a Lindbladian $\lind_u$ of the form in \Cref{eq:hard-instance} as $\mathcal{L}_u = \sum_{i=1}^{n-1} \mathcal{L}_{u,i}$.
We show how to simulate short-time evolution under each individual $\mathcal{L}_{u,i}$ using the program state.
Then, we show how to evolve under $\mathcal{L}_u$ for time $t > 0$ via Trotterization.
Importantly, the error is independent of $u$.

\begin{lemma}[Simulation of one term]
    \label{lem:one-term-sim}
    Let $0 < \delta < 1$, and let $i \in [n-1]$.
    Consider a Lindbladian term $\mathcal{L}_{u,i}$.
    There exists a quantum channel $\Upsilon_{i,\delta}$ independent of $u$ such that
    \begin{equation}
        \Phi_{i,\delta} = \Upsilon_{i,\delta}(\ket{\beta_i\sqrt{\delta}}\bra{\beta_i \sqrt{\delta}} \otimes \rho)
    \end{equation}
    satisfies
    \begin{equation}
        \norm{\Phi_{i,\delta} - e^{\delta \lind_{u,i}}}_{\diamond} \leq C\delta^2,
    \end{equation}
    for some $C > 0$, which may depend on $g,W,m$.
    Here, $\Phi_{i,\delta}$ depends on $u$ (through $\beta_i$), but $\Upsilon_{i,\delta}$ does not.
\end{lemma}

\begin{proof}
    Let $\rho$ be an $n$-qubit quantum state which we wish to evolve by $\lind_{u,i}$ for time $\delta > 0$.
    The channel $\Phi_{i,\delta}$ is defined by appending an ancilla state which is a bosonic coherent state, applying a joint unitary on the system and ancilla, and tracing out the ancilla.
    In particular, let $\beta_i = \ii u_i\sqrt{2/g}$ and define the unitary
    \begin{align}
        \begin{split}
        U_{i,\delta} &\triangleq \cos\sqrt{g\delta/2} \ket{0}\bra{0} \otimes I - \sin\sqrt{g\delta/2}\ket{0}\bra{1} \otimes Z_0 + \sin\sqrt{g\delta/2}\ket{1}\bra{0}\otimes Z_i\\
        &+ \cos\sqrt{g\delta/2}\ket{1}\bra{1}\otimes Z_0 Z_i + (I_A -\ket{0}\bra{0} - \ket{1}\bra{1}) \otimes I,
        \end{split}
    \end{align}
    where $A$ denotes the ancilla.
    In other words, when the ancilla is in a state in $\mathrm{span}\{\ket{0},\ket{1}\}$,
    \begin{equation}
        U_{i,\delta} = \begin{pmatrix}
            \cos\sqrt{g\delta/2} I & -\sin\sqrt{g\delta/2} Z_0\\
            \sin\sqrt{g\delta/2} Z_i & \cos\sqrt{g\delta/2} Z_0 Z_i
        \end{pmatrix}.
    \end{equation}
    Meanwhile, for an occupation number at least $2$, it acts as identity.
    Then, we define the channel as
    \begin{equation}
        \Phi_{i,\delta}(\rho) = \tr_A\left(U_{i,\delta}\left(\ket{\beta_i\sqrt{\delta}}\bra{\beta_i\sqrt{\delta}} \otimes \rho\right) U_{i,\delta}^\dagger\right).
    \end{equation}
    We want to show that this achieves a small error in simulating $e^{\lind_{u,i} \delta}$.
    First, we can write $\Phi_{i,\delta}$ in terms of its Kraus operators.
    Expanding $\Phi_{i,\delta}$, we have
    \begin{align}
        \Phi_{i,\delta}(\rho) &= \sum_{r=0}^\infty (\bra{r} \otimes I) U_{i,\delta}(\ket{\beta_i\sqrt{\delta}}\bra{\beta_i\sqrt{\delta}}\otimes \rho) U_{i,\delta}^\dagger (\ket{r} \otimes I)\\
        &= \sum_{r=0}^\infty (\bra{r} \otimes I) U_{i,\delta}(\ket{\beta_i\sqrt{\delta}}\otimes I)\rho(\bra{\beta_i\sqrt{\delta}} \otimes I) U_{i,\delta}^\dagger (\ket{r} \otimes I)\\
        &= \sum_{r=0}^\infty K_r \rho K_r^\dagger,
    \end{align}
    where
    \begin{equation}
        K_r \triangleq (\bra{r} \otimes I) U_{i,\delta}(\ket{\beta_i\sqrt{\delta}} \otimes I) = e^{-|\beta_i|^2\delta/2} \sum_{\ell=0}^\infty \frac{(\beta_i\sqrt{\delta})^\ell}{\sqrt{\ell!}}(\bra{r} \otimes I) U_{i,\delta} (\ket{\ell} \otimes I).
    \end{equation}
    Here, we expand the coherent state using \Cref{eq:coherent-state}.
    Using the form of $U_{i,\delta}$, we can see that
    \begin{align}
        \begin{split}
        K_0 &= e^{-|\beta_i|^2\delta/2}\left(\cos\sqrt{\frac{g\delta}{2}}I - \beta_i\sqrt{\delta}\sin\sqrt{\frac{g\delta}{2}} Z_0\right),\\
        K_1 &= e^{-|\beta_i|^2\delta/2}\left(\sin\sqrt{\frac{g\delta}{2}} Z_i + \beta_i\sqrt{\delta}\cos\sqrt{\frac{g\delta}{2}} Z_0 Z_i\right),\\
        K_r &= e^{-|\beta_i|^2\delta/2}\frac{(\beta_i\sqrt{\delta})^r}{\sqrt{r!}}I,
        \end{split}
    \end{align}
    for $r \geq 2$.
    By definition of the Lindbladian, notice that we can write $J_i \triangleq \sqrt{\frac{g}{2}} Z_i + \ii u_i\sqrt{\frac{2}{g}} Z_0 Z_i$ so that
    \begin{equation}
        \lind_{u,i}(\rho) = -\ii [u_i Z_0, \rho] + J_i\rho J_i^\dagger -\frac{1}{2}\{J_i^\dagger J_i, \rho\}.
    \end{equation}
    Note that, by a Taylor expansion,
    \begin{equation}
        \norm{e^{\delta \lind_{u,i}} - \mathrm{id} - \delta \lind_{u,i}}_\diamond \leq \int_0^\delta (\delta - x)\norm{\lind_{u,i}^2 e^{x\lind_{u,i}}}_\diamond\,dx \leq \frac{1}{2}\delta^2\norm{\lind_{u,i}}_{\diamond}^2,
    \end{equation}
    where we use that the diamond norm of a quantum channel is $1$.
    Note that
    \begin{equation}
        \label{eq:lind-diamond}
        \norm{\lind_{u,i}}_\diamond \leq 2|u_i| + 2\norm{J_i}^2_{\mathrm{op}} \leq 2|u_i| + g + 4u_i^2/g,
    \end{equation}
    where the upper bound only depends on $g,W,m$.
    Here, we used that
    \begin{equation}
        \label{eq:j-square}
        J_i^\dagger J_i = \left(\sqrt{\frac{g}{2}}Z_i - \ii u_i \sqrt{\frac{2}{g}} Z_0 Z_i\right)\left(\sqrt{\frac{g}{2}}Z_i + \ii u_i \sqrt{\frac{2}{g}} Z_0 Z_i\right) = \left(\frac{g}{2} + u_i^2 \frac{2}{g}\right)I = \left(\frac{g}{2} + |\beta_i|^2\right)I.
    \end{equation}
    Thus, it suffices to bound $\norm{\Phi_{i,\delta} - \mathrm{id} - \delta \lind_{u,i}}_{\diamond}$ since we can absorb the terms above into $C$.
    Now, we expand out $\Phi_{i,\delta} - \mathrm{id} - \delta \lind_{u,i}$:
    \begin{align}
        (\Phi_{i,\delta} - \mathrm{id} - \delta\lind_{u,i})(\rho) &= \sum_{r=0}^\infty K_r \rho K_r^\dagger - \rho + \ii \delta [u_i Z_0, \rho] - \delta J_i \rho J_i^\dagger + \frac{\delta}{2}\{J_i^\dagger J_i, \rho\}\\
        \begin{split}
        &=\left(K_0 - I + \ii \delta u_i Z_0 + \frac{\delta}{2}J_i^\dagger J_i\right)\rho + \rho\left(K_0 - I + \ii\delta u_i Z_0 + \frac{\delta}{2}J_i^\dagger J_i\right)^\dagger\\
        &+ (K_0 - I)\rho(K_0 - I)^\dagger + K_1\rho K_1^\dagger - (\sqrt{\delta} J_i)\rho(\sqrt{\delta} J_i)^\dagger + \sum_{r=2}^\infty K_r \rho K_r^\dagger,
        \end{split}
    \end{align}
    where in the second line, we use that
    \begin{equation}
        K_0\rho K_0^\dagger - \rho = (K_0 - I)\rho + \rho(K_0 - I)^\dagger + (K_0 - I)\rho(K_0 - I)^\dagger.
    \end{equation}
    We want to bound the diamond norm of this superoperator:
    \begin{align}
        \begin{split}
        \norm{\Phi_{i,\delta} - \mathrm{id} - \delta\lind_{u,i}}_\diamond &\leq 2\norm{K_0 - I +\ii \delta u_i Z_0 + \frac{\delta}{2}J_i^\dagger J_i}_{\mathrm{op}} + \norm{K_0 - I}^2_{\mathrm{op}}\\
        &+\left(\norm{K_1}_{\mathrm{op}} + \sqrt{\delta}\norm{J_i}_{\mathrm{op}}\right)\norm{K_1 - \sqrt{\delta}J_i}_{\mathrm{op}} + \sum_{r=2}^\infty \norm{K_r}^2_{\mathrm{op}},\label{eq:diamond-bound}
        \end{split}
    \end{align}
    where we use $\norm{X (\cdot) Y}_\diamond \leq \norm{X}_{\mathrm{op}} \norm{Y}_{\mathrm{op}}$.
    In the first term in the second line, we use
    \begin{equation}
        \norm{X(\cdot)X^\dagger - Y(\cdot)Y^\dagger}_\diamond \leq (\norm{X}_{\mathrm{op}} + \norm{Y}_{\mathrm{op}})\norm{X - Y}_{\mathrm{op}}.
    \end{equation}
    We bound each of these terms individually.
    Using \Cref{eq:j-square}, then we can bound the first term in \Cref{eq:diamond-bound} as
    \begin{align}
        &\norm{K_0 - I +\ii \delta u_i Z_0 + \frac{\delta}{2}J_i^\dagger J_i}_{\mathrm{op}}\\
        &= \Bigg\|\left(e^{-|\beta_i|^2\delta/2} \cos\sqrt{\frac{g\delta}{2}} - 1 + \frac{|\beta_i|^2 + g/2}{2}\delta\right)I - \beta_i\left(e^{-|\beta_i|^2\delta/2}\sqrt{\delta} \sin\sqrt{\frac{g\delta}{2}} - \sqrt{\frac{g}{2}}\delta\right)Z_0\Bigg\|_{\mathrm{op}}\\
        &\leq \left|e^{-|\beta_i|^2\delta/2} \cos\sqrt{\frac{g\delta}{2}} - 1 + \frac{|\beta_i|^2 + g/2}{2}\delta\right| + |\beta_i|\left|e^{-|\beta_i|^2\delta/2}\sqrt{\delta} \sin\sqrt{\frac{g\delta}{2}} - \sqrt{\frac{g}{2}}\delta\right|,
    \end{align}
    where in the second line, we use the definitions of $\beta_i$ and $K_0$.
    Then, we have
    \begin{align}
        \left|e^{-|\beta_i|^2\delta/2} \cos\sqrt{\frac{g\delta}{2}} - 1 + \frac{|\beta_i|^2 + g/2}{2}\delta\right| &\leq \left|e^{-|\beta_i|^2 \delta/2} - 1 + \frac{|\beta_i|^2\delta}{2}\right| + \left|\cos\sqrt{\frac{g\delta}{2}} - 1 + \frac{g \delta}{4}\right|\\
        &+ \left|e^{-|\beta_i|^2\delta/2} - 1\right|\left|\cos\sqrt{\frac{g\delta}{2}}-1\right|\\
        &\leq \left(\frac{|\beta_i|^4}{8} + \frac{g^2}{96} + \frac{|\beta_i|^2 g}{8}\right)\delta^2.
    \end{align}
    In the first inequality, we use triangle inequality.
    In the last line, we use several basic inequalities: $|e^{-x} - 1 + x| \leq x^2/2$, $|\cos x - 1 + x^2/2| \leq x^4/24$, $|e^{-x}-1| \leq x$, and $|\cos x - 1| \leq x^2/2$.
    Similarly,
    \begin{align}
        \left|e^{-|\beta_i|^2\delta/2}\sqrt{\delta} \sin\sqrt{\frac{g\delta}{2}} - \sqrt{\frac{g}{2}}\delta\right| &\leq \left|e^{-|\beta_i|^2 \delta/2}\sqrt{\delta}\left(\sin\sqrt{\frac{g\delta}{2}} - \sqrt{\frac{g\delta}{2}}\right)\right| + \left|\sqrt{\frac{g}{2}}\delta\left(e^{-|\beta_i|^2\delta/2} - 1\right)\right|\\
        &\leq \left(\frac{1}{6}\left(\frac{g}{2}\right)^{3/2} + \sqrt{\frac{g}{2}}\frac{|\beta_i|^2}{2}\right)\delta^2.
    \end{align}
    In the second line, we use that $e^{-x} \leq 1$ for $x \geq 0$, $|\sin x - x| \leq |x|^3/6$, and $|e^{-x} - 1| \leq x$.
    Putting everything together, we see that 
    \begin{equation}
        \norm{K_0 - I + \ii \delta u_i Z_0 + \frac{\delta}{2}J_i^\dagger J_i}_{\mathrm{op}} \leq C_0 \delta^2
    \end{equation}
    for some $C_0 > 0$ depending on $g, W$.
    Now, for the second term in \Cref{eq:diamond-bound}, we have
    \begin{equation}
        \norm{K_0 - I}_{\mathrm{op}}^2 \leq \delta^2\left(|u_i| + \frac{1}{2}\norm{J_i}_{\mathrm{op}}^2 + C_0\right)^2.
    \end{equation}
    For the third term in \Cref{eq:diamond-bound}, we have
    \begin{align}
        \norm{\frac{K_1}{\sqrt{\delta}} - J_i}_{\mathrm{op}} &= \left\|\left(e^{-|\beta_i|^2\delta/2}\frac{1}{\sqrt{\delta}}\sin\sqrt{\frac{g\delta}{2}} - \sqrt{\frac{g}{2}}\right)Z_i + \beta_i\left(e^{-|\beta_i|^2\delta/2} \cos\sqrt{\frac{g\delta}{2}} - 1\right)Z_0 Z_i \right\|_{\mathrm{op}}\\
        &\leq \left|e^{-|\beta_i|^2\delta/2}\frac{1}{\sqrt{\delta}}\sin\sqrt{\frac{g\delta}{2}} - \sqrt{\frac{g}{2}}\right| + |\beta_i|\left|e^{-|\beta_i|^2\delta/2}\cos\sqrt{\frac{g\delta}{2}} - 1\right|.
    \end{align}
    Then, we have
    \begin{align}
        \left|e^{-|\beta_i|^2\delta/2}\frac{1}{\sqrt{\delta}}\sin\sqrt{\frac{g\delta}{2}} - \sqrt{\frac{g}{2}}\right| &\leq e^{-|\beta_i|^2\delta/2}\left|\frac{1}{\sqrt{\delta}}\sin\sqrt{\frac{g\delta}{2}} - \sqrt{\frac{g}{2}}\right| + \sqrt{\frac{g}{2}}|e^{-|\beta_i|^2\delta/2} - 1|\\
        &\leq \left(\frac{g^{3/2}}{6\cdot 2^{3/2}} + \frac{|\beta_i|^2\sqrt{g/2}}{2}\right)\delta,
    \end{align}
    where in the second line we use $e^{-x} \leq 1$ for $x \geq 0$, $|\sin x - x| \leq |x|^3/6$, and $|e^{-x} - 1| \leq x$.
    Similarly,
    \begin{equation}
        |\beta_i|\left|e^{-|\beta_i|^2\delta/2}\cos\sqrt{\frac{g\delta}{2}} - 1\right| \leq e^{-|\beta_i|^2\delta/2}\left|\cos\sqrt{\frac{g\delta}{2}} - 1\right| + |e^{-|\beta_i|^2\delta/2} - 1| \leq |\beta_i|\left(\frac{g/2 + |\beta_i|^2}{2}\right) \delta,
    \end{equation}
    where in the last inequality, we use $|\cos x-1| \leq x^2/2$ and $|e^{-x} - 1| \leq x$.
    Putting everything together, this says that
    \begin{equation}
        \norm{K_1 - \sqrt{\delta}J_i}_{\mathrm{op}} \leq C_1 \delta^{3/2}
    \end{equation}
    for some $C_1 > 0$.
    Moreover, this also implies that 
    \begin{equation}
        \norm{K_1}_{\mathrm{op}} \leq \sqrt{\delta}(\norm{J_i}_{\mathrm{op}} + C_1).
    \end{equation}
    Finally, for the last term in \Cref{eq:diamond-bound}, using the definition of $K_r$, we have
    \begin{equation}
        \sum_{r=2}^\infty \norm{K_r}_{\mathrm{op}}^2 = e^{-|\beta_i|^2\delta/2} \sum_{r=2}^\infty \frac{(|\beta_i|^2\delta)^r}{r!} = 1 - e^{-|\beta_i|^2\delta/2}(1 + |\beta_i|^2 \delta) \leq \frac{|\beta_i|^4\delta^2}{2}.
    \end{equation}
    Putting everything together and gathering all of the constants into \Cref{eq:diamond-bound}, we can clearly see that the claim follows.
\end{proof}

\begin{corollary}
    \label{coro:compose}
    Let $t > 0$.
    Let $r \in \mathbb{Z}_{\geq 0}$ be such that $\delta = t/r \in (0,1)$.
    Consider the channels $\Phi_{i,\delta}$ from \Cref{lem:one-term-sim}.
    Then,
    \begin{equation}
        \norm{\left(\Phi_{n-1,\delta} \circ \cdots \circ \Phi_{1,\delta}\right)^r - \left(e^{\delta \lind_{u,n-1}} \circ \cdots \circ e^{\delta \lind_{u,1}}\right)^r}_\diamond \leq C' t \delta
    \end{equation}
    for some $C' > 0$, which may depend on $g, W, m$.
    Moreover, one can implement the channel $\left(\Phi_{n-1,\delta} \circ \cdots \circ \Phi_{1,\delta}\right)^r$ using one copy of the program state $\Psi_u(t)$.
\end{corollary}

\begin{proof}
    This follows easily by \Cref{lem:one-term-sim}:
    \begin{equation}
        \norm{\left(\Phi_{n-1,\delta} \circ \cdots \circ \Phi_{1,\delta}\right)^r - \left(e^{\delta \lind_{u,n-1}} \circ \cdots \circ e^{\delta \lind_{u,1}}\right)^r}_\diamond \leq r\sum_{i=1}^{n-1} \norm{\Phi_{i,\delta} - e^{\delta \lind_{u,i}}}_{\diamond} \leq \frac{t}{\delta} (n-1) C \delta^2 = (n-1)Ct\delta.
    \end{equation}
    We also show that the channel $(\Phi_{n-1,\delta} \circ \cdots \circ \Phi_{1,\delta})^r$ can be implemented using one copy of the program state.
    Recall that
    \begin{equation}
        \Psi_u(t) = \bigotimes_{i=1}^{n-1} \ket{\beta_i\sqrt{t}} \bra{\beta_i\sqrt{t}}.
    \end{equation}
    Each $\Phi_{i,\delta}$ channel is implemented using one state $\ket{\beta_i\sqrt{\delta}}$.
    Thus, applying \Cref{lem:splitting} $r$ times, we can use $\Psi_u(t)$ to get $r$ copies of $\ket{\beta_i \sqrt{\delta}}$ for each $i \in [n-1]$.
\end{proof}

\begin{lemma}[Trotterization]
    \label{lem:trotter}
    Let $t > 0$.
    Let $r \in \mathbb{Z}_{\geq 0}$ be such that $\delta = t/r \in (0,1)$.
    Consider the channels $\Phi_{i,\delta}$ from \Cref{lem:one-term-sim}.
    Then,
    \begin{equation}
        \norm{\left(\Phi_{n-1,\delta} \circ \cdots \circ \Phi_{1,\delta}\right)^r - e^{t\mathcal{L}_u}}_\diamond \leq C'' t \delta
    \end{equation}
    for some $C'' > 0$, which may depend on $g, W, m$.
    Moreover, one can implement the channel $\left(\Phi_{n-1,\delta} \circ \cdots \circ \Phi_{1,\delta}\right)^r$ using one copy of the program state $\Psi_u(t)$.
\end{lemma}

\begin{proof}
    This follows by standard product formulas for Lindbladians, see, e.g., Theorem 6 in~\cite{wang2026lindbladian}.
    Namely, this shows that
    \begin{equation}
        \norm{e^{\delta \lind_{u,n-1}} \circ \cdots \circ e^{\delta \lind_{u,1}} - e^{\delta \lind_u}}_\diamond \leq \frac{\delta^2}{2}\sum_{\substack{i,j=1\\i < j}}^{n-1} \norm{[\lind_{u,i}, \lind_{u,j}]}_{\diamond}.
    \end{equation}
    Composing $r$ times, then we have
    \begin{equation}
        \norm{\left(e^{\delta \lind_{u,n-1}} \circ \cdots \circ e^{\delta \lind_{u,1}}\right)^r - e^{t \lind_u}}_\diamond \leq t\delta (n-1)^2 \max_{i\in[n]}\norm{\lind_{u,i}}_{\diamond}^2.
    \end{equation}
    By \Cref{eq:lind-diamond}, this can be bounded in terms of $g,W,m$.
    Thus, by triangle inequality with \Cref{coro:compose}, we arrive at the claim.
\end{proof}

This shows how to simulate one non-adaptive query $e^{\mathcal{L}_u t}$.
To complete the proof of \Cref{prop:simulation}, we show how to extend this to simulate adaptive queries to the time evolution.

\begin{proof}[Proof of {\Cref{prop:simulation}}]
    Consider a distinguishing algorithm, which has query access to a Lindbladian $\lind_u \in \mathcal{C}_{m,g,\epsilon}$ and uses a total evolution time of $T$.
    Let $\Delta \triangleq \min\{1/2, 2\eta/(C''T)\}$, where $C'' > 0$ is from \Cref{lem:trotter}.
    Then, the channel $\Phi$ which takes the program state $\Psi_u(T)$ and produces an output close to the distinguishing algorithm's output acts as follows.
    We start with $\Psi_u(T) = \bigotimes_{i=1}^{n-1} \ket{\beta_i\sqrt{T}}\bra{\beta_i \sqrt{T}}$.
    When a query of duration $t > 0$ is made, use \Cref{lem:splitting} to transform $\ket{\beta_i\sqrt{T}}$ into $\ket{\beta_i\sqrt{t}}\ket{\beta_i\sqrt{T-t}}$ for each $i \in [n-1]$.
    Then, use the first register and \Cref{lem:trotter} (with $r = \lceil t/\Delta \rceil$, $\delta = t/r$) to simulate $e^{t\lind_u}$.
    Call the channel from \Cref{lem:trotter} $\Phi_t$.
    Inductively, for further queries, we can repeat this process using the second register as our new initial state.

    By \Cref{lem:trotter}, we know that
    \begin{equation}
        \norm{\Phi_t - e^{t\lind_u}}_\diamond \leq C'' t \Delta.
    \end{equation}
    We use this to show the claim of \Cref{prop:simulation}.
    We prove this via a hybrid argument.
    Suppose the distinguishing algorithm uses $N$ queries to the time evolution under $\lind_u$.
    Let $\omega_u^{(j)}$ denote the output distribution of the distinguishing algorithm when the first $j$ queries are exact (i.e., directly querying $e^{t\lind_u}$) and the remaining queries are simulated via the program state.
    Then, $\omega_u^{(0)} = \Phi(\Psi_u(T))$ and $\omega_u^{(N)} = \omega_u$, and we wish to bound the trace distance between these two states.

    First, consider the difference between two states $\omega_u^{(j)}$ and $\omega_u^{(j-1)}$.
    Notice that, in the procedure to obtain such states, the first $j-1$ queries are the same (i.e., exact queries).
    Thus, conditioned on the first $j-1$ queries, the $j$-th query is of duration $t_j$ and is the same for both states.
    Thus, by \Cref{lem:trotter}, conditioned on the first $j-1$ queries, the trace distance between $\omega_u^{(j)}$ and $\omega_u^{(j-1)}$ is bounded by $C''t_j\Delta/2$, where the factor of $1/2$ comes from the definition of trace distance.
    Hence, we have
    \begin{equation}
        \label{eq:dtr-omega}
        \dtr(\omega_u^{(j)}, \omega_u^{(j-1)}) \leq \frac{C''\Delta}{2}\E[t_j],
    \end{equation}
    where the expectation is over the randomness of the first $j-1$ queries.
    By triangle inequality, we have
    \begin{equation}
        \dtr(\omega_u, \Phi(\Psi_u(T))) \leq \sum_{j=1}^N \dtr(\omega_u^{(j)}, \omega_u^{(j-1)}) \leq \frac{C''\Delta}{2}\E\left[\sum_{j=1}^N t_j\right] \leq \frac{C''T\Delta}{2} \leq \eta.
    \end{equation}
    In the second inequality, we use \Cref{eq:dtr-omega}.
    In the third inequality, we use that the total time evolution of the distinguishing algorithm is bounded by $T$.
    In the last inequality, we use our choice of $\Delta$.
\end{proof}

\subsection{Lower bound proof}

Finally, we put everything together to prove \Cref{thm:lower-bound}.
From the previous two sections, we see that it suffices to upper bound the trace distance between two program states when $u \sim \Pi_+$ versus when $u \sim \Pi_-$.
Recall that $u_i = (\epsilon W/m) B_i X_i$, where $B_i \sim \mathrm{Bernoulli}(1/m^2)$ and $X_i \sim \mathcal{D}_{\pm}$.
Define
\begin{equation}
    \gamma \triangleq \frac{\epsilon W}{m} \sqrt{\frac{2T}{g}}.
\end{equation}
Then, we can write $\ket{\beta_i\sqrt{T}} = \ket{\ii \gamma B_i X_i}$.
Define
\begin{equation}
    \label{eq:rho-pm}
    \rho_{\pm} \triangleq \left(1-\frac{1}{m^2}\right)\ket{0}\bra{0} + \frac{1}{m^2}\sigma_{\pm},\quad \sigma_{\pm}\triangleq\E_{X \sim \mathcal{D}_{\pm}}[\ket{\ii \gamma X}\bra{\ii \gamma X}].
\end{equation}
Then, the expected value of the program state is
\begin{equation}
    \E_{u\sim \Pi_{\pm}}[\Psi_u(T)] = \rho_{\pm}^{\otimes (n-1)}.
\end{equation}
To upper bound the trace distance between these two states, we first consider $\sigma_{\pm}$.
We show that $\sigma_{\pm}$ are only distinguishable based on the first moment when considering moments up to order $2q+1$, and higher moments are small.
Then, because of the Bernoulli weights in $\rho_{\pm}$, the distinct first moments are still not enough to distinguish the two states.

To analyze $\sigma_{\pm}$, we want to use the moment matching property of the distributions from \Cref{lem:matching-moments}.
This requires first showing that only the moments up to a certain order contribute non-trivially, as \Cref{lem:matching-moments} only guarantees matching moments up to order $2q + 1$.

\begin{lemma}
    \label{lem:coherent-truncate}
    Let $z \in \mathbb{R}$.
    Then,
    \begin{equation}
        \ket{\ii z}\bra{\ii z} = \sum_{\ell=0}^\infty C_\ell z^\ell,\quad C_\ell \triangleq \sum_{\substack{a,b,c =0\\a + b +2c = \ell}}^\infty \frac{(-1)^c \ii^a (-\ii)^b}{c!\sqrt{a!b!}}\ket{a}\bra{b}.
    \end{equation}
    Moreover, for $L \in \mathbb{Z}_{\geq 0}$ and $\gamma \geq 0$, then
    \begin{equation}
        \sup_{|x| \leq 1}\left\|\ket{\ii \gamma x} \bra{\ii \gamma x} - \sum_{\ell=0}^L C_\ell (\gamma x)^\ell\right\|_{\tr} \leq 2e^{12\gamma^2 - L\log 2}.
    \end{equation}
\end{lemma}

\begin{proof}
    To prove the first statement, by the definition of coherent states, we have
    \begin{equation}
        \ket{\ii z}\bra{\ii z} = e^{-z^2} \sum_{a,b=0}^\infty \frac{\ii^a (-\ii)^b z^{a+b}}{\sqrt{a!b!}}\ket{a}\bra{b} = \sum_{a,b,c=0}^\infty \frac{(-1)^c \ii^a (-\ii)^b z^{a+b+2c}}{c!\sqrt{a!b!}}\ket{a}\bra{b},
    \end{equation}
    where the last equality follows by Taylor expanding $e^{-z^2}$.
    For the trace norm bound, let $x$ be such that $|x| \leq 1$.
    Then, we have
    \begin{align}
        \label{eq:truncate-bound}
        \left\|\ket{\ii \gamma x} \bra{\ii \gamma x} - \sum_{\ell=0}^L C_\ell (\gamma x)^\ell\right\|_{\tr} \leq \sum_{\ell=L+1}^\infty \norm{C_\ell}_{\tr} |\gamma x|^\ell \leq 2^{-L}\sum_{\ell = L+1}^\infty \norm{C_\ell}_{\tr} (2\gamma)^\ell \leq 2^{-L}\sum_{\ell=0}^\infty \norm{C_\ell}_{\tr} (2\gamma)^\ell.
    \end{align}
    In the first inequality, we use the expansion of $\ket{\ii z}\bra{\ii z}$ that we just proved and triangle inequality.
    In the second inequality, we use that $|x| \leq 1$ and $\ell > L$ in the sum.
    In the third inequality, we add more terms to the sum.
    To bound this, we have
    \begin{equation}
        \norm{C_\ell}_{\tr} \leq \sum_{\substack{a,b,c=0\\a+b+2c=\ell}}^\infty \frac{1}{c!\sqrt{a!b!}} \triangleq s_\ell.
    \end{equation}
    Then, for $y \geq 0$
    \begin{equation}
        \sum_{\ell=0}^\infty s_\ell y^\ell = e^{y^2}\left(\sum_{a=0}^\infty \frac{y^a}{\sqrt{a!}}\right)^2 = e^{y^2}\left(\sum_{a=0}^\infty \frac{(\sqrt{2}y)^a}{\sqrt{a!}} 2^{-a/2}\right)^2 \leq e^{y^2}\left(\sum_{a=0}^\infty \frac{(2y^2)^a}{a!}\right)\left(\sum_{a=0}^\infty 2^{-a}\right) = 2e^{3y^2}.
    \end{equation}
    The first equality follows by Taylor expanding $e^{y^2}$.
    The inequality uses Cauchy-Schwarz.
    The last equality follows via a geometric series and the Taylor expansion of $e^x$.
    Plugging this back into \Cref{eq:truncate-bound} with $y = 2\gamma$ completes the proof.
\end{proof}

\begin{corollary}
    \label{coro:sigma-diff}
    Let $\sigma_{\pm}$ be defined as in \Cref{eq:rho-pm}.
    Then,
    \begin{equation}
        \sigma_+ - \sigma_- = \frac{2\gamma}{W} C_1 + E,
    \end{equation}
    where $\norm{E}_{\tr} \leq 4e^{12\gamma^2 - (2q+1)\log 2}$ and $C_1$ is as defined in \Cref{lem:coherent-truncate}.
\end{corollary}

\begin{proof}
    This follows by applying \Cref{lem:coherent-truncate} with $L = 2q+1$.
    Explicitly, we have
    \begin{align}
        &\sigma_+ - \sigma_-\\
        &= \E_{X \sim \mathcal{D}_+}[\ket{\ii \gamma X}\bra{\ii \gamma X}] - \E_{X \sim \mathcal{D}_-}[\ket{\ii \gamma X}\bra{\ii \gamma X}]\\
        &= \left(\E_{X \sim \mathcal{D}_+}[\ket{\ii \gamma X}\bra{\ii \gamma X}] - \sum_{\ell=0}^{2q+1}C_\ell\gamma^\ell\E_{X\sim \mathcal{D}_+}[X^\ell]\right) + \left(\sum_{\ell=0}^{2q+1} C_\ell \gamma^\ell\E_{X\sim \mathcal{D}_-}[X^\ell] - \E_{X \sim \mathcal{D}_-}[\ket{\ii \gamma X}\bra{\ii \gamma X}]\right) + \frac{2\gamma}{W}C_1\\
        &= E_1 + E_2 + \frac{2\gamma}{W}C_1.
    \end{align}
    In the second line, we use the definition of $\sigma_{\pm}$.
    In the third line, we add and subtract $\sum_\ell C_\ell \gamma^\ell \E_{X\sim \mathcal{D}_{\pm}}[X^\ell]$, where by \Cref{lem:matching-moments}, $2\gamma C_1/W$ is the right correction term to be added because these sums only differ for $\ell =1$.
    The last line follows by applying \Cref{lem:coherent-truncate} with $L = 2q+1$ twice, where here $\norm{E_1}_{\tr}, \norm{E_2}_{\tr} \leq 2e^{12\gamma^2 - (2q+1)\log 2}$.
    Then, taking $E = E_1 + E_2$ completes the claim.
\end{proof}

This tells us that the only term that contributes non-trivially to the difference between $\sigma_+$ and $\sigma_-$ is the first moment term $C_1$.
In the following lemma, we prove that this contribution is not enough to distinguish $\rho_{\pm}$.

\begin{lemma}
    \label{lem:rho-diff}
    Let $\rho_{\pm}, \sigma_{\pm}$ be defined as in \Cref{eq:rho-pm}.
    Let $\theta = 2\gamma/(W(m^2 - 1))$.
    If $|\theta| \leq \pi/2$, then
    \begin{equation}
        \dtr(\rho_+^{\otimes (n-1)}, \rho_-^{\otimes (n-1)}) \leq (n-1)\left(\frac{2}{m^2}e^{12\gamma^2 - (2q+1)\log 2} + \left(1-\frac{1}{m^2}\right)\theta^2 + \frac{1}{m^2}|\theta|\right) + \sqrt{n-1}|\theta|.
    \end{equation}
\end{lemma}

\begin{proof}
    For conciseness, in this proof, let $p = 1/m^2$.
    Then, $\theta = 2\gamma p/(W(1-p))$.
    Let $G \triangleq \ii (\ket{1}\bra{0} + \ket{0}\bra{1})$ and $U \triangleq e^{\theta G}$.
    Then, by triangle inequality
    \begin{equation}
        \label{eq:trace-dist-u}
        \dtr(\rho_+^{\otimes (n-1)}, \rho_-^{\otimes (n-1)}) \leq \dtr(\rho_+^{\otimes (n-1)}, (U\rho_-U^\dagger)^{\otimes (n-1)}) + \dtr((U\rho_-U^\dagger)^{\otimes (n-1)}, \rho_-^{\otimes (n-1)}),
    \end{equation}
    so it suffices to bound each of these terms individually.
    First, we bound the first term.
    Note that
    \begin{align}
        \rho_+ - U\rho_-U^\dagger &= (1-p)(\ket{0}\bra{0} - U\ket{0}\bra{0}U^\dagger) + p\left(\sigma_+ - U\sigma_-U^\dagger\right)\\
        &= (1-p)(\ket{0}\bra{0} - U\ket{0}\bra{0}U^\dagger) + p(\sigma_+ - \sigma_- +\sigma_- - U\sigma_- U^\dagger)\\
        &= (1-p)(\ket{0}\bra{0} - U\ket{0}\bra{0}U^\dagger) + \frac{2\gamma p}{W}C_1 + pE + p(\sigma_- - U\sigma_-U^\dagger)\\
        &= (1-p)(\ket{0}\bra{0} + \theta C_1 - U\ket{0}\bra{0}U^\dagger) + pE + p(\sigma_- - U\sigma_-U^\dagger).\label{eq:rho-plus-dist}
    \end{align}
    In the third line, we use \Cref{coro:sigma-diff}.
    In the fourth line, we use that $\theta = 2\gamma p/(W(1-p))$.
    Then, we can bound the trace norm of each of these terms.
    By the Fundamental Theorem of Calculus,
    \begin{equation}
        U\sigma_- U^\dagger - \sigma_- = \int_0^1 \partial_s (e^{s\theta G} \sigma_- e^{-s\theta G})\,ds = \int_0^1 e^{s\theta G}[\theta G, \sigma_-] e^{-s \theta G}\,ds.
    \end{equation}
    Thus, we can bound the trace norm as
    \begin{equation}
        \norm{U\sigma_- U^\dagger - \sigma_-}_{\tr} \leq 2|\theta|.
    \end{equation}
    Similarly, using the Fundamental Theorem of Calculus (twice), then
    \begin{equation}
        U\ket{0}\bra{0}U^\dagger - \ket{0}\bra{0} - \theta C_1 = \theta^2 \int_0^1 (1-s) e^{s\theta G}[G, [G, \ket{0}\bra{0}]] e^{-s\theta G}\,ds.
    \end{equation}
    Thus, we can bound the trace norm as
    \begin{equation}
        \norm{U\ket{0}\bra{0}U^\dagger - \ket{0}\bra{0} - \theta C_1}_{\tr} \leq 2\theta^2.
    \end{equation}
    Plugging back into \Cref{eq:rho-plus-dist}, we have
    \begin{equation}
        \dtr(\rho_+, U\rho_- U^\dagger) \leq \frac{p}{2}\norm{E}_{\tr} + (1-p)\theta^2 + p|\theta| \leq 2pe^{12\gamma^2 - (2q+1)\log 2} + (1-p)\theta^2 + p|\theta|,
    \end{equation}
    where in the last inequality, we use \Cref{coro:sigma-diff}.
    Then, we have
    \begin{equation}
        \label{eq:rho-plus-dist-final}
        \dtr(\rho_+^{\otimes (n-1)}, (U \rho_- U^\dagger)^{\otimes (n-1)}) \leq (n-1)\left(2pe^{12\gamma^2 - (2q+1)\log 2} + (1-p)\theta^2 + p|\theta|\right).
    \end{equation}
    Now, we bound the second term in \Cref{eq:trace-dist-u}.
    Let $\ket{\psi}$ be a purification of $\rho_-$.
    Then, $(U \otimes I) \ket{\psi}$ is a purification of $U \rho_- U^\dagger$.
    Notice that the overlap between these purifications is $\bra{\psi} U\otimes I \ket{\psi} = \tr(\rho_- U)$.
    Then, we have
    \begin{equation}
        \dtr(\rho_-^{\otimes (n-1)}, (U\rho_- U^\dagger)^{\otimes (n-1)}) \leq \sqrt{1- |\tr(\rho_- U)|^{2(n-1)}} \leq \sqrt{1- \cos^{2(n-1)}\theta} \leq \sqrt{n-1}|\sin \theta| \leq \sqrt{n-1} |\theta|.
    \end{equation}
    In the second inequality, we use
    \begin{equation}
        \frac{1}{2}(U + U^\dagger) = \cos \theta (\ket{0}\bra{0} + \ket{1}\bra{1}) + I - (\ket{0}\bra{0} + \ket{1}\bra{1})
    \end{equation}
    so that
    \begin{equation}
        |\tr(\rho_- U)| \geq \mathrm{Re}\tr(\rho_- U) = 1- (1-\cos\theta) \tr(\rho_- (\ket{0}\bra{0} + \ket{1}\bra{1})) \geq \cos\theta.
    \end{equation}
    In the third inequality, we use $1-(1-x)^{n-1} \leq (n-1)x$ for $x \in [0,1]$ for $x = \sin^2\theta$.
    Putting everything together in \Cref{eq:trace-dist-u}, we have
    \begin{equation}
        \dtr(\rho_+^{\otimes (n-1)}, \rho_-^{\otimes (n-1)}) \leq (n-1)\left(2pe^{12\gamma^2 - (2q+1)\log 2} + (1-p)\theta^2 + p|\theta|\right) + \sqrt{n-1}|\theta|,
    \end{equation}
    as required.
\end{proof}

Finally, with this, we can finish the proof of \Cref{thm:lower-bound}.

\begin{proof}[Proof of {\Cref{thm:lower-bound}}]
    Suppose for the sake of contradiction that there exists a learning algorithm which uses total evolution time $T \leq gm^2(2q+1)/(3200\epsilon^2 W^2)$ and succeeds with probability at least $0.99$.
    By \Cref{lem:learn-to-distinguish}, this implies there exists a distinguishing algorithm between $\Pi_{\pm}$ using total evolution time $T$ which succeeds with probability $0.98$.
    To contradict this, we bound the probability of success of any distinguishing algorithm.

    Let $\omega_u$ be a diagonal density matrix representing the output distribution of a distinguishing algorithm.
    Then, we have
    \begin{align}
        \dtr\left(\E_{u\sim \Pi_+}[\omega_u], \E_{u\sim \Pi_-}[\omega_u]\right) &\leq \dtr\left(\E_{u\sim \Pi_+}[\Psi_u(T)], \E_{u\sim \Pi_-}[\Psi_u(T)]\right)\\
        &= \dtr(\rho_+^{\otimes (n-1)}, \rho_-^{\otimes (n-1)})\\
        &\leq (n-1)\left(\frac{2}{m^2}e^{12\gamma^2 - (2q+1)\log 2} + \left(1-\frac{1}{m^2}\right)\theta^2 + \frac{1}{m^2}|\theta|\right) + \sqrt{n-1}|\theta|.\label{eq:omega-dist}
    \end{align}
    The first line follows by \Cref{coro:simulation}.
    The second line follows from the discussion at the beginning of this subsection.
    The last line follows by \Cref{lem:rho-diff}, where $\theta = 2\gamma/(W(m^2-1))$ and $\gamma = \frac{\epsilon W}{m}\sqrt{\frac{2T}{g}}$.
    Here, note that $|\theta| \leq \pi/2$, as required by \Cref{lem:rho-diff}, because, by our upper bound on $T$,
    \begin{equation}
        \gamma^2 = \frac{2\epsilon^2 W^2 T}{m^2g} \leq \frac{2q+1}{100}
    \end{equation}
    so that
    \begin{equation}
        |\theta| = \left|\frac{2\gamma}{W(m^2 - 1)}\right| \leq \frac{4}{m^2}|\gamma| \leq \frac{2\sqrt{2q+1}}{5m^2} < \frac{\pi}{2},
    \end{equation}
    where we use that $W \geq 2(q+1) > 1$ in the first inequality.
    In the last inequality is true because $q = \lceil 10\log m \rceil$ and $m \geq 10^8$.
    Moreover, the upper bound on $\gamma^2$ shows that
    \begin{equation}
        \frac{2(n-1)}{m^2}e^{12\gamma^2 - (2q+1)\log 2} \leq 2me^{-0.57 (2q+1)},
    \end{equation}
    where we also use $n-1 = m^3$.
    Using this and the bound on $|\theta|$ in \Cref{eq:omega-dist}, we have
    \begin{equation}
        \dtr\left(\E_{u\sim \Pi_+}[\omega_u], \E_{u\sim \Pi_-}[\omega_u]\right) \leq 2me^{-0.57(2q+1)} + \frac{4(2q+1)}{25m} + \frac{2\sqrt{2q+1}}{5m} + \frac{2\sqrt{2q+1}}{5\sqrt{m}} < 10^{-3}, 
    \end{equation}
    where the last inequality follows because $q = \lceil 10 \log m \rceil$ and $m \geq 10^8$.

    Let $0 \leq M \leq I$ be an operator such that $\{M, I - M\}$ specifies a two-outcome measurement representing the distinguishing algorithm's output.
    When $u\sim \Pi_+$ with probability $1/2$ and $u \sim \Pi_-$ with probability $1/2$, then the success probability of the distinguishing algorithm is
    \begin{align}
        \Pr[\mathrm{success}] &= \frac{1}{2}\tr\left(M \E_{u\sim \Pi_+} [\omega_u]\right) + \frac{1}{2}\tr\left((I-M)\E_{u\sim \Pi_-}[\omega_u]\right)\\
        &= \frac{1}{2} + \frac{1}{2} \tr\left(M\left(\E_{u\sim \Pi_+}[\omega_u] - \E_{u\sim \Pi_-}[\omega_u]\right)\right)\\
        &\leq \frac{1}{2}\left(1 + \dtr\left(\E_{u\sim \Pi_+}[\omega_u], \E_{u\sim \Pi_-}[\omega_u]\right)\right) < 0.501.
    \end{align}
    Thus, a distinguishing algorithm cannot have success probability at least $0.98$ using $T \leq gm^2(2q+1)/(3200\epsilon^2 W^2)$, which contradicts \Cref{lem:learn-to-distinguish}.
    Hence, we can conclude that a total evolution time of
    \begin{equation}
        \tet > \frac{gm^2(2q+1)}{3200\epsilon^2 W^2} \geq \frac{gm^2}{12800e^4\epsilon^2 (q+1)} = \Omega\left(\frac{gm^2}{\epsilon^2 \log m}\right)
    \end{equation}
    is required.
    In the second inequality, we use that $W < 2e^2(q+1)$ from \Cref{lem:matching-moments}.
    In the last equality, we use that $q = \lceil 10\log m \rceil$.
    This completes the proof.
\end{proof}

\section*{Acknowledgments}
L.L. thanks Aram Harrow, Jordi Montana-Lopez, Quynh Nguyen, Umesh Vazirani, and Thomas Vidick for helpful discussions.
L.L.\ is supported by the U.S. Department of
Energy, Office of Science, Office of Advanced Scientific Computing Research, Department of Energy Computational Science Graduate Fellowship under Award Number DE-SC0026073.
E.T.\ is supported by the Miller Institute for Basic Research in Science, University of California Berkeley.
J.W.\ is supported by the NSF CAREER award CCF-233971 and a Sloan Fellowship.

This report was prepared as an account of work sponsored by an agency of the United States Government. Neither the United States Government nor any agency thereof, nor any of their employees, makes any warranty, express or implied, or assumes any legal liability or responsibility for the accuracy, completeness, or usefulness of any information, apparatus, product, or process disclosed, or represents that its use would not infringe privately owned rights. Reference herein to any specific commercial product, process, or service by trade name, trademark, manufacturer, or otherwise does not necessarily constitute or imply its endorsement, recommendation, or favoring by the United States Government or any agency thereof. The views and opinions of authors expressed herein do not necessarily state or reflect those of the United States Government or any agency thereof.

\section*{AI Disclosure}
All ideas and technical proofs excluding the lower bound were first generated by the authors of the manuscript without AI usage.
ChatGPT 5.6 Sol was used to check the correctness of these proofs, where it identified some issues, which the authors repaired with the help of ChatGPT 5.6 Sol.
The fixes affected technical details in Sections 4.1 and 5.3 and did not change the algorithm or central ideas of the proofs.

The proof of the lower bound was generated using ChatGPT 6 Astra.
The final proof presented in this manuscript was written and verified by the authors, and we take full responsibility for its correctness and exposition.

%\AtNextBibliography{\small}
\printbibliography

\end{document}